\documentclass[%
 twocolumn,
 superscriptaddress,
 nofootinbib,
 amsmath,amssymb,
 aps,
 showkeys,
 a4paper,
]{revtex4-1}

\usepackage[utf8]{inputenc}
\usepackage[english]{babel} 
\usepackage{graphicx}
\usepackage{dcolumn}
\usepackage{bm}

\usepackage{mathtools}

\usepackage{fancybox}
\usepackage{enumitem}
\usepackage{natbib}
\usepackage{footmisc}
\usepackage[bookmarks, colorlinks=true, urlcolor=blue, linkcolor=black, citecolor=blue]{hyperref}
\usepackage{xcolor}
\usepackage{float}
\usepackage{physics}
\usepackage[paperwidth=210mm,paperheight=297mm,centering,hmargin=1.5cm,vmargin=3cm]{geometry}
\usepackage{bbm}
\usepackage{textcomp}
\usepackage[caption=false]{subfig}
\usepackage[export]{adjustbox}
\usepackage{ulem}

\usepackage{algorithm}
\usepackage{algorithmicx}
\usepackage{algpseudocode}
\usepackage{amsthm}
\newtheorem{theorem}{Theorem}

\newtheorem{lemma}{Lemma}
\allowdisplaybreaks

\begin{document}

\normalem


\title{Quantum enhancements for deep reinforcement learning in large spaces}

\author{Sofiene Jerbi}
\affiliation{Institute for Theoretical Physics, University of Innsbruck, Technikerstr. 21a, A-6020 Innsbruck, Austria}
\author{Lea M. Trenkwalder}
\affiliation{Institute for Theoretical Physics, University of Innsbruck, Technikerstr. 21a, A-6020 Innsbruck, Austria}
\author{Hendrik Poulsen Nautrup}
\affiliation{Institute for Theoretical Physics, University of Innsbruck, Technikerstr. 21a, A-6020 Innsbruck, Austria}
\author{Hans J. Briegel}
\affiliation{Institute for Theoretical Physics, University of Innsbruck, Technikerstr. 21a, A-6020 Innsbruck, Austria}
\affiliation{Department of Philosophy, University of Konstanz, Fach 17, 78457 Konstanz, Germany}
\author{Vedran Dunjko}
\affiliation{LIACS, Leiden University, Niels Bohrweg 1, 2333 CA Leiden, The Netherlands}

 
\date{\today}

\begin{abstract}
\noindent In the past decade, the field of quantum machine learning has drawn significant attention due to the prospect of bringing genuine computational advantages to now widespread algorithmic methods. However, not all domains of machine learning have benefited equally from quantum enhancements. Notably, deep learning and reinforcement learning, despite their tremendous success in the classical domain, both individually and combined, remain relatively unaddressed by the quantum community. Arguably, one reason behind this is the systematic use in these domains of models and methods without prominent computational bottlenecks, leaving little room for quantum improvements. In this work, we study the state-of-the-art neural-network approaches for reinforcement learning with quantum enhancements in mind. We demonstrate the substantial learning advantage that models with a sampling bottleneck can provide over conventional neural network architectures in complex learning environments. These so-called energy-based models, like deep energy-based reinforcement learning, and deep projective simulation that we also introduce in this work, effectively allow to trade off learning performance for efficiency of computation. To alleviate the additional computational costs, we propose to leverage future and near-term quantum algorithms, resulting in overall more advantageous learning algorithms. This is achieved using cutting-edge and new quantum computing machinery to speed-up classical sampling methods and by employing generalized models to gain an additional quantum advantage.
\end{abstract}

\maketitle


\section{\label{sec:intro}Introduction}

Over the course of the last years, machine learning (ML) has emerged not just as a promising domain of application for quantum computing \cite{biamonte17,dunjko18,ciliberto18,schuld15,adcock15}, but is often listed as one of its potential ``killer applications" \cite{schuld18}.
For instance, in the contexts of classification and big data problems, genuine quantum effects have been exploited to speed-up computational bottlenecks stemming from either genuinely hard computational tasks \cite{neven08,lloyd16} or from data volume and dimensionality \cite{rebentrost14,lloyd14,wiebe12}.
More recently, new types of machine learning algorithms for classification and generative modeling have been put forward to leverage the potential of noisy intermediate-scale quantum (NISQ) \cite{preskill18} devices.
In this flavor of QML, both quantum generalizations of classical methods \cite{amin18}, and genuinely quantum models \cite{schuld20} have also been proposed.

\begin{figure}
	\centering
	\includegraphics[width=1.0\columnwidth]{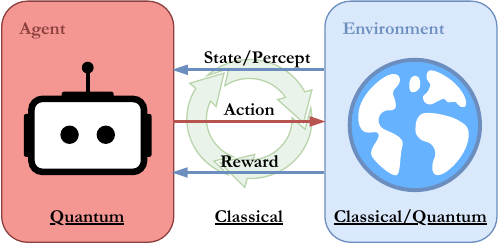}
	\caption{\textbf{Quantum-enhanced reinforcement learning.} Agent and environment maintain a cyclic interaction consisting of \emph{states} (or percepts) that the agent uses to deliberate on its \emph{actions} according to its policy, and perceives feedback on its behavior in the form of \emph{rewards}. In the quantum-enhanced framework, the agent is granted access to a, generally universal, quantum computer to enhance its deliberation process.}\vspace{-1.25\baselineskip}
	\label{fig:QC-RL_framework}
\end{figure}

However, the promises of quantum enhancements have not had as profound an impact in all ML domains. In particular, perhaps the most celebrated area of ML based on deep feed-forward neural networks (NNs) has essentially remained unaddressed from a quantum perspective, besides a few exceptions \cite{allcock18,kerenidis19}. Similarly, the field of reinforcement learning (RL), which focuses on an interactive mode of learning (see Fig.\ \ref{fig:QC-RL_framework}) and has recently demonstrated its potential with the AlphaGo achievements \cite{silver17,silver18}, has only seen speed-ups for special models \cite{paparo14,levit17,neukart18} and in settings where learning environments are quantum in nature \cite{dunjko17b,dunjko16,dunjko15b,cornelissen18,ronagh19,dunjko17}. One of the main reasons for this oversight is simple. Since RL and deep NNs are predominantly studied as methods for real-world applications, emphasis has consistently been placed on variants that favor general applicability and can be run efficiently on classical computers, even with large datasets. Consequently, the standardized algorithms are already very efficient in a complexity-theoretic sense, leaving little room for quantum improvements. For example, the evaluation of feed-forward NNs and their gradients has a negligible contribution to their overall training cost. An exception applies to training algorithms, where quantum methods need to rely on special settings including QRAM architectures \cite{giovannetti08,allcock18}, or quantum environments \cite{dunjko17b}. However, these settings significantly limit real-world applicability and are incompatible with NISQ constraints.\\

\textit{Contributions. }In this work, we study the cutting-edge NN-based approaches for RL in the context of learning environments with large state and action spaces. We demonstrate that the conventional paradigm, where the computations of an agent's decisions using a NN are simple and efficient, may have substantial learning disadvantages in complex environments. Further, we show that rather different models, which additionally require approximate sampling from distributions specified by the NN, can do significantly better. These so-called energy-based models \cite{lecun06} are however computationally more challenging, introducing a trade-off between learning performance and efficiency. We show how to mitigate these computational costs using existing and new quantum methods. In particular, we apply a number of known quantum algorithmic methods to ameliorate computationally expensive aspects of training our models. Additionally, by combining results for efficient classical Gibbs sampling on a quantum computer with previous methods for preparing Gibbs states of quantum Hamiltonians, we obtain the currently most efficient method for the quantum Gibbs sampling problem. This is critical in our RL algorithms since we also describe how to use quantum generalizations of energy-based models, e.g., so-called quantum Boltzmann machines \cite{amin18,kappen20}, for state-of-the-art RL methods. Our work provides a new set of results, further highlighting the substantial benefits that may arise from combining machine learning with quantum computation, especially in the near term.\\

The paper is organized as follows: Sec.\ \ref{sec:background} serves as an introduction to the well-known problem of large state and action spaces in reinforcement learning and the (in)adequacy of different methods to deal with this issue. Sec.\ \ref{sec:DEB-RL} presents a step-by-step construction of the deep energy-based networks we consider. We evaluate their performance on benchmarking tasks and compare them to standard NN architectures used in deep RL in Sec.\ \ref{sec:numerical-simulations}. Sec.\ \ref{sec:hybrid-EBRL} introduces a hybrid quantum-classical algorithm to train energy-based models for RL and describes the quantum subroutines it relies on. In the last Sec.\ \ref{sec:conclusion} we discuss our results and present an outlook.\break

\vspace{-1.5\baselineskip}
\section{Background \& preliminaries\label{sec:background}}

In this section, we introduce the well-studied case of reinforcement learning environments with large state and action spaces and discuss modern approaches \cite{mnih15,sallans04} that attempt to deal with such scenarios by contrasting them with traditional RL methods \cite{sutton98,briegel12}.

\subsection{Tabular methods}

Reinforcement learning (RL) is a framework for the design of learning agents in interactive environments (see Fig.\ \ref{fig:QC-RL_framework}). Typically, the goal of RL is to learn a policy $\pi(\bm{a}|\bm{s})$ that specifies an agent's actions $\bm{a}$ given a perceived state $\bm{s}$ towards maximizing its future rewards in a certain environment. There exist various approaches to learn optimal policies. However, in this work, we focus on two particular methods, namely value-based methods and projective simulation, detailed shortly, that we treat within a unifying framework. This framework allows us to combine and take advantage of their respective techniques and previous work. We give here a short description of the two methods of interest, further detailed in Appendix \ref{sec:intro-RL}:
\begin{itemize}[leftmargin=3mm]
	\item Methods deriving from dynamic programming focus on a so-called \emph{action-value function} $Q(\bm{s},\bm{a})$ to link policy and expected rewards. These \emph{value-based methods} (VBMs) \cite{sutton98} include, for instance, $Q$-learning and SARSA.
	\item Projective simulation (PS) \cite{briegel12} is a physics-inspired model for agency\footnote{Note that, in standard VBMs, states and actions are part of a mathematical object describing the agent's task environment: a Markov Decision Process. In PS, although similar to VBMs in its two-layer form, the perspective is different and more agent-centered: representations of states and actions have a separate status as entities in the agent's memory, they can be created, erased, or modified as part of the agent's dynamics of memory processing and deliberation.} that relies on weighted connections in a directed graph to reflect experienced rewards and derive a policy. These weights are referred to as $h$-values. In its simplest (i.e., two-layered) form, the underlying weighted graph only contains directed connections from states to actions, resulting in state-action $h$-values $h(\bm{s},\bm{a})$.
\end{itemize}

To unify the treatment of these two approaches, we define the notion of \emph{merit functions} $M$, of which both auxiliary functions $Q$ and $h$ are instances. These are real-valued functions defined on the state-action space that assign a certain value of merit to all actions in any given state. A standard procedure to learn merit functions and their associated policies is the so-called \emph{tabular} method, which is shared by both VBMs and two-layered PS: scalar values $\Hat{M}(\bm{s},\bm{a})$ are stored for each state-action pair, initialized such as to produce a random behavior. These can be normalized into a policy, e.g., by using a Gibbs/Boltzmann distribution, also called ``softmax" distribution in ML:
\begin{equation}\label{eq:energy-based_policy}
	\pi(\bm{a}|\bm{s}) = \frac{e^{\beta \Hat{M}(\bm{s},\bm{a})}}{\sum_{\bm{a'}} e^{\beta \Hat{M}(\bm{s},\bm{a'})}}
\end{equation}
where $\beta>0$ is an inverse temperature hyperparameter. \linebreak An update rule given by VBMs or PS allows one to iteratively update the values $\Hat{M}(\bm{s},\bm{a})$ associated with \emph{experienced} state-action pairs. This update depends on the rewards received from the environment after following this policy and increases the estimated merit value of rewarded state-action pairs such that the agent is more likely to encounter them in the future. For instance, in the case of $Q$-learning (which is a VBM):
\begin{multline}
	\Hat{Q}(\bm{s^{(t)}},\bm{a^{(t)}}) \leftarrow (1-\alpha) \Hat{Q}(\bm{s^{(t)}},\bm{a^{(t)}}) \\+ \alpha [r^{(t+1)} + \gamma\max_{\bm{a}} \Hat{Q}(\bm{s^{(t+1)}},\bm{a})]
\end{multline}
where $\bm{a^{(t)}}$ is the action performed by the agent in the state $\bm{s^{(t)}}$ at timestep $t$, which lead it to the state $\bm{s^{(t+1)}}$ and associated reward $r^{(t+1)}$; $\alpha$ is a learning rate taking value in $[0,1]$ and used to deal with statistical fluctuations.\\

The fundamental problem with tabular approaches is their lack of generalization capabilities: there is no mechanism enabling an agent to generalize learned values to unobserved states and actions. This puts a strong lower bound on the number of interactions with the environment required to learn an optimal merit function: $\Omega(\abs{\mathcal{S}}\times\abs{\mathcal{A}})$. That is, the number of interactions scales at least linearly in the number of state-action pairs \cite{ronagh19}. Hence, learning through tabular methods becomes intractable for environments with large (potentially continuous) state and action spaces. However, this is a rather common aspect of practical RL tasks:
\begin{itemize}[leftmargin=3mm]
	\item Large (and continuous) state spaces appear, for instance, in vision problems, where the dimension of the state space grows linearly with the number of pixels in their input images.
	\item Large action spaces are common for recommender systems such as YouTube or Netflix, where the actions are the possible recommendations of contents that can be made to a user. They also appear in many control problems of large processes, including, for instance, quantum control and quantum experiments \cite{august18,fosel18,bukov18,bukov18b,melnikov18b,nautrup19,sweke18,palittapongarnpim16,vedaie18,porotti19,wallnofer19,xu2019,dalgaard20}. For tasks such as these, actions are characterized by continuous parameters that are commonly discretized to efficiently search the action space. The drawback of this approach is that the action space grows polynomially with the inverse precision of discretization.
\end{itemize}
While generalization mechanisms do not overcome lower bound complexities, in practice, RL environments are not arbitrary and we do not necessarily care about learning the optimal policy exactly, just ``very good" policies (in a practical sense). In this case, much faster learning is possible through generalization heuristics, as we explain next.

\subsection{Generalization with neural networks\label{sec:nn-generalization}}

A common way to provide RL agents with generalization mechanisms is through \emph{function approximation} methods, closely related to regression problems in supervised learning. Such approaches introduce a parametrized representation $M^{\bm{\theta}}(\bm{s},\bm{a})$ of the merit function, replacing an otherwise sparsely-filled merit table by a function that is defined on the entire state-action space. Updates of the merit function act no longer on individual values $\Hat{M}(\bm{s},\bm{a})$ but on a set of parameters $\bm{\theta}$, allowing one to indirectly update the values $M^{\bm{\theta}}(\bm{s},\bm{a})$ of many -- including unobserved -- state-action pairs simultaneously. Using the terminology of supervised learning, $\{M^{\bm{\theta}}\}_{\bm{\theta}}$ constitutes the set of all possible hypothesis functions and learning consists in updating the parameters $\bm{\theta}$ as to find the setting that best approximates the true merit function $M$. Due to their success in supervised and unsupervised learning, neural networks, parametrized by their weights, have gained popularity as function approximators in RL. As we clarify next, while generalization over the state space is shared by all function approximation methods based on neural networks, generalization over the action space, i.e., learning a useful similarity measure between actions, is less common.\\

\begin{figure}
	\centering
	\includegraphics[width=1.0\columnwidth]{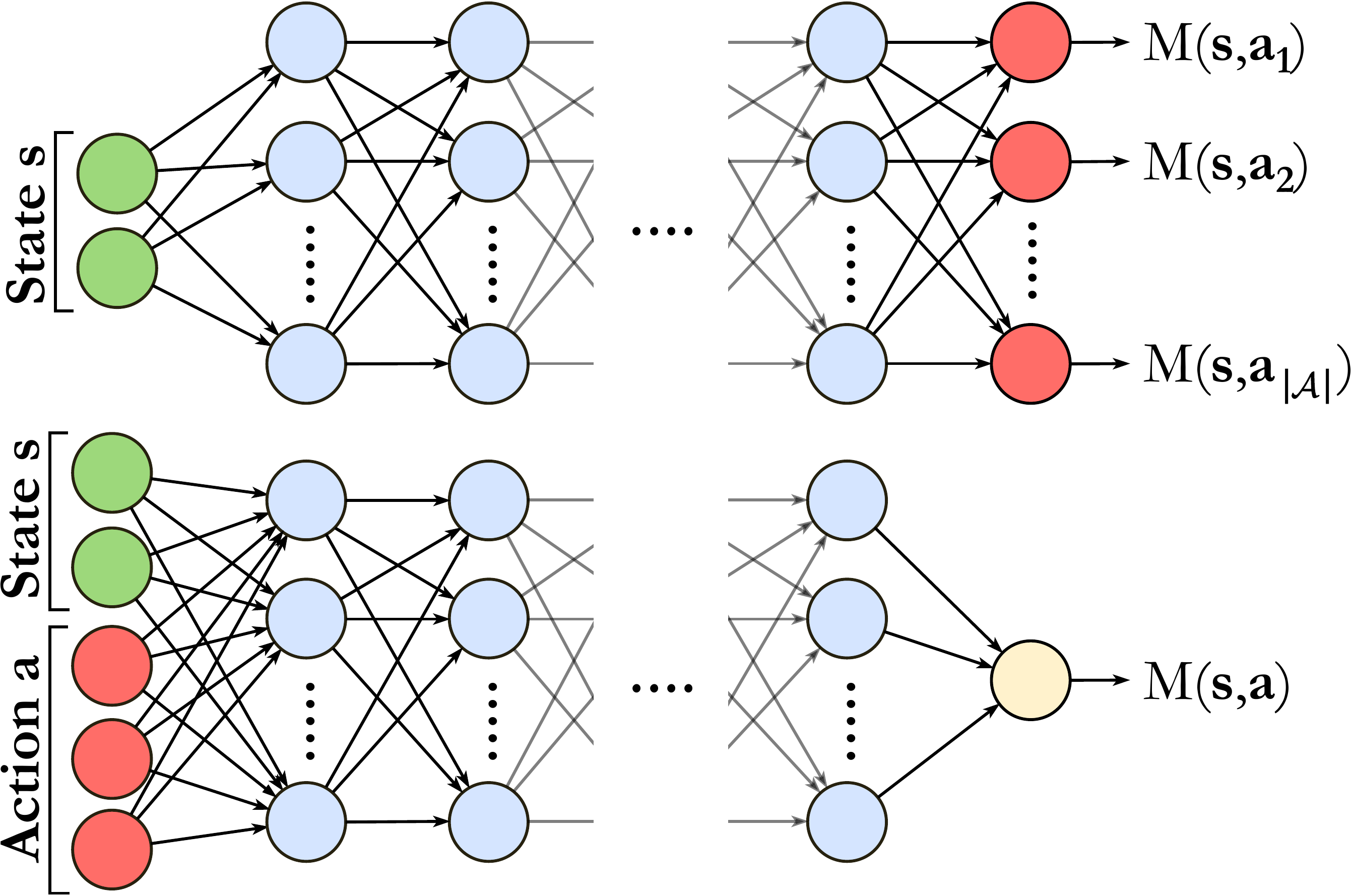}
	\caption{\textbf{A sketch of comparison between the layouts of DQN-like networks \cite{mnih15} and the DEBN we consider.} The DQN-like model (top) estimates the merit function for all actions simultaneously while the DEBN (bottom) estimates the merit function for a single state-action pair. The first hidden layers may be convolutional.}
	\label{fig:DQN-DEM}
\end{figure}

When using feedforward neural networks as function approximators, one can choose between two different approaches to represent the merit function (see Fig.\ \ref{fig:DQN-DEM}). The first representation (\emph{i}) takes as input a state $\bm{s}$ and outputs \emph{all values} $\left( M^{\bm{\theta}}(\bm{s},\bm{a_1}),\ \ldots,\ M^{\bm{\theta}}(\bm{s},\bm{a_{\abs{\mathcal{A}}}}) \right)$ associated to that state; while the second (\emph{ii}) takes as input a single state-action pair $(\bm{s},\bm{a})$ and outputs its corresponding merit value $M^{\bm{\theta}}(\bm{s},\bm{a})$. Each of these representations has its advantages in different contexts, that we clarify below.\\

In their seminal work \cite{mnih15}, Mnih et al.\ introduce a so-called deep $Q$-network (DQN) of type (\emph{i}) that quickly became the standard in value-based deep RL \cite{hessel18}. DQNs spiked interest notably due to their success in vision problems, e.g., playing Atari games, as their so-called convolutional hidden layers made them perfectly suited to deal with the high-dimensional space of the images they receive as input. Importantly for our later analyses, utilizing a network of type (\emph{i}) makes it particularly efficient to evaluate the Boltzmann policy of the agent (Eq.\ (\ref{eq:energy-based_policy})), since it requires only a single evaluation of the neural network followed by the application of a single additional softmax layer that normalizes all the merit values of its output.\\
However, this representation suffers from shortcomings when it comes to tasks with large action spaces, as we elaborate. During the first stages of learning, it is important for the agent to keep an explorative behavior to not get stuck in a possibly sub-optimal policy. For this reason, the function approximator should be able to capture several of the local maxima of the merit function, which may originate from complex state-action correlations. In other words, the representation should be \emph{multimodal}. However, it is reasonable to expect a neural network that can only represent merit values that are linear combinations of (non-linear) state features to be unlikely to learn highly multimodal merit functions. This is why several works have chosen to avoid DQNs in environments with large action spaces \cite{haarnoja17,dulac15}.\\

We refer to representations of type (\emph{ii}) as \emph{deep energy-based networks} (DEBN) due to their interesting connection to so-called energy-based models, as explored in machine learning \cite{lecun06}. Inspired by statistical physics, energy-based models represent a target function using the energy function of a statistical system (that can be in one of combinatorially many configurations, each having a certain energy). When normalized, this energy function gives rise to a probability distribution over all possible configurations. An example of energy-based model is the Boltzmann machine, introduced by Hinton \& Sejnowski \cite{hinton83} and which is inspired by the Ising model for spin systems \cite{ising25}. Restricted Boltzmann machines (RBMs), i.e., Boltzmann machines with a further constrained set of energy functions, were the first energy-based models reportedly used in RL with the work of Sallans \& Hinton \cite{sallans04}, followed by various extensions \cite{heess12,otsuka10,elfwing16,crawford18,haarnoja17}. Interestingly, although RBMs are commonly associated to a class of so-called stochastic recurrent neural networks, it turns out that they can be expressed as feedforward neural networks of type (\emph{ii}) \cite{martens13} with a single hidden layer. The DEBNs we consider here are then straightforwardly obtained by adding additional hidden layers to this feedforward neural network \cite{ngiam11}. Energy-based models are also known to be good at representing multimodal distributions \cite{haarnoja17,du19}. Hence, since DEBNs are able to encode deep and multimodal representations, they are well suited for tasks with both large state and action spaces.\\
A major drawback of DEBNs is the high cost of sampling from the policy (\ref{eq:energy-based_policy}): $|\mathcal{A}|$ evaluations of the network are needed to evaluate the policy exactly. \emph{Approximate} sampling methods based on random walks, such as Gibbs sampling, can speed-up the sampling process. However, Long et al.\ \cite{long10} showed that approximate sampling is still an NP-hard problem in the case of RBMs, an argument that can be trivially generalized to DEBNs.

In this work, we experimentally explore these intuitively motivated advantages of deep energy-based models in complex RL environments and show how quantum algorithms can help mitigate the computational bottleneck that they introduce.

\section{Deep energy-based RL\label{sec:DEB-RL}}

In this section, we give a short introduction to energy-based models, which we build upon to design our deep energy-based learning agents. We first consider RBM function approximators for RL \cite{sallans04} and present an established connection between RBMs and shallow feedforward networks \cite{martens13} which allows us to consider a natural extension to a deep architecture -- our so-called DEBNs. We show that this function approximation method is compatible with both VBM and PS update rules. This construction thereby also provides an extension of PS to NN function approximation, which we call \emph{deep PS}, in analogy to deep Q-learning.

\subsection{Energy-based RL\label{sec:EB-RL}}

\begin{figure}
	\centering
	\includegraphics[width=1.0\columnwidth]{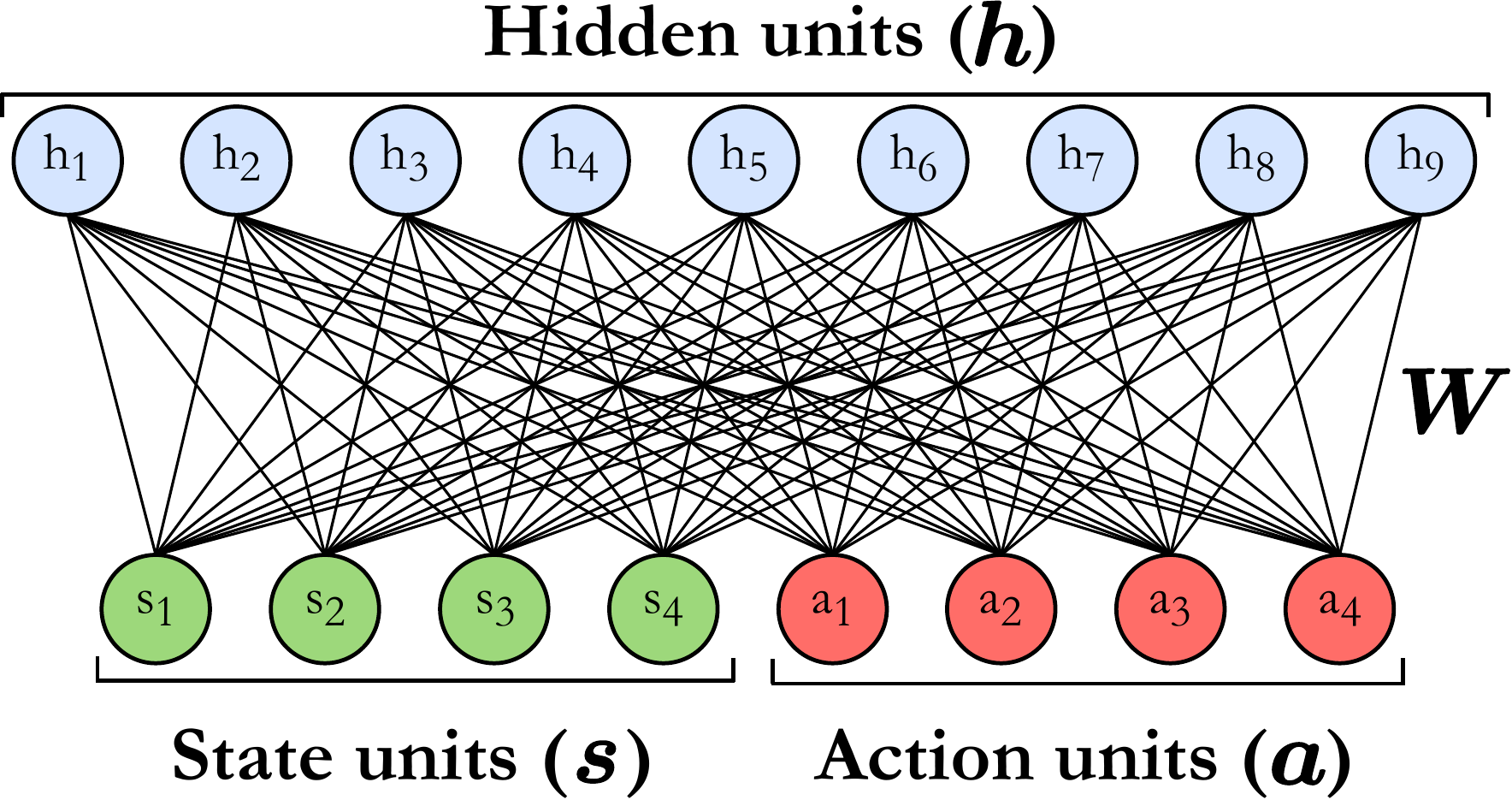}
	\caption{\textbf{A restricted Boltzmann machine used for reinforcement learning.} The visible units are divided between state and action units. Their connections to hidden units are weighted by a weight matrix $W$. Biases are not represented for compactness.}
	\label{fig:RBM-RL}
\end{figure}

\textbf{Energy-based models (EBMs). }EBMs are a popular class of probabilistic models that find applications in both supervised and unsupervised learning \cite{lecun06}. They are often used as generative models, that is, to model the empirical probability distribution of observed data vectors $\{ \bm{v} \}_{\bm{v} \in \mathcal{D}}$, e.g., strings of binary variables. EBMs assign a scalar energy value to each configuration of these variables, e.g., all possible binary strings of length $\abs{\bm{v}}$, whose structure should reflect dependencies (i.e., correlations) between the variables. This energy function is usually parametrized by a vector $\bm{\theta}$ that determines the strengths of these dependencies. In the case of a Boltzmann machine, for instance, the energy function takes the form:
\begin{equation}
	- E^{\bm{\theta}}(\bm{v}) = \sum_{i<j} w_{ij} v_i v_j + \sum_{i} b_i v_i
\end{equation}
where $\bm{\theta} = \left( (w_{ij})_{1\leq i < j \leq \abs{\bm{v}}},\ (b_{i})_{1\leq i \leq \abs{\bm{v}}} \right)$. Note that this particular Boltzmann machine is said to be fully-visible, since its energy only depends on the visible variables $v_i$ that encode the data (further explained shortly).
The topology of the model specifies which variables are assumed (in)dependent. Combined, topology and parametrization constrain the family of available energy functions $\{E^{\bm{\theta}}\}_{\bm{\theta}}$ that an EBM can represent.\\
A probability distribution over the configuration space can be derived from the energy function by a certain normalization over all possible configurations (example below). Normalization assigns a high probability to configurations with low energy and vice versa. This generative probability distribution can be trained, i.e., iteratively modified, to best fit an (empirical) probability distribution through updates of the parameters $\bm{\theta}$. These updates effectively select new energy functions from $\{E^{\bm{\theta}}\}_{\bm{\theta}}$, and, equivalently, new generative distributions.\\
In order to capture more complex dependencies between the variables, an additional set of auxiliary variables $\bm{h}$, called \emph{hidden} or \emph{latent} variables, can be added to the set of so-called visible variables $\bm{v}$. The model has now an energy $E^{\bm{\theta}}(\bm{v},\bm{h})$ that includes terms characterized by the parameters between visible and hidden variables. The probability distribution over the visible variables is then specified by a function $F^{\bm{\theta}}(\bm{v})$ called the \emph{free energy} of the model due to its connection to the equilibrium free energy in statistical physics. It is obtained by tracing out (or marginalizing) the hidden variables from the normalizing distribution. For instance, in the case of a softmax normalization, we have:
\begin{equation}\label{eq:energy-distribution}
	P^{\bm{\theta}}(\bm{v}) = \frac{\sum_{\bm{h}}e^{-E^{\bm{\theta}}(\bm{v},\bm{h})}}{\sum_{\bm{v'},\bm{h'}}e^{-E^{\bm{\theta}}(\bm{v'},\bm{h'})}} = \frac{e^{-F^{\bm{\theta}}(\bm{v})}}{\sum_{\bm{v'}}e^{-F^{\bm{\theta}}(\bm{v'})}}
\end{equation}
and
\begin{equation}\label{eq:free-energy}
	F^{\bm{\theta}}(\bm{v}) = E^{\bm{\theta}}(\bm{v},\langle \bm{h}\rangle_{P(\bm{h} | \bm{v})}) + H(\bm{h} | \bm{v}) 
\end{equation}
where $\langle \bm{h}\rangle_{P(\bm{h} | \bm{v})}$ is the expectation value of the hidden variables under $P(\bm{h} | \bm{v})$ (defined similarly to distribution (\ref{eq:energy-distribution})) and $H(\bm{h} | \bm{v})$ is their Shannon conditional entropy under this same probability distribution. A derivation of Eq.\ (\ref{eq:free-energy}) is provided in Appendix \ref{sec:bm_energy-based}.\\

\textbf{Energy-based function approximation. }Following the idea of Sallans \& Hinton \cite{sallans04}, we use the free energies of energy-based models as \emph{parametrized} approximations of merit functions (such as the $Q$-function in VBMs or the $h$-values of PS), that is, we take $M^{\bm{\theta}}$ to be defined by the model's $-F^{\bm{\theta}}$. We start with a simple energy-based model: a restricted Boltzmann machine (RBM), further described in Appendix \ref{sec:bm_energy-based}. To represent a policy using an RBM, we divide its visible units into state and action units (see Fig.\ \ref{fig:RBM-RL}). This allows one to define a conditional probability distribution:
 \begin{equation}\label{eq:DEBN-policy}
	\pi^{\bm{\theta}}(\bm{a}|\bm{s}) = \frac{e^{-F^{\bm{\theta}}(\bm{s},\bm{a})}}{\sum_{\bm{a'}}e^{-F^{\bm{\theta}}(\bm{s},\bm{a'})}} = \frac{e^{M^{\bm{\theta}}(\bm{s},\bm{a})}}{\sum_{\bm{a'}}e^{M^{\bm{\theta}}(\bm{s},\bm{a'})}}
\end{equation}
effectively encoding the policy of an agent as prescribed by Eq.\ (\ref{eq:energy-based_policy}) (the inverse temperature $\beta$ can be usually be absorbed by the parameters $\bm{\theta}$).\\
This approximation of the merit function by the RBM free energy, as opposed to a tabular approach, has the advantage of being defined on the entire state-action space, which allows the agent to act non-randomly on previously unseen states. Moreover, (implicitly) stored values $M^{\bm{\theta}}(\bm{s},\bm{a})$ are not updated independently from one another, as updates on $\bm{\theta}$ act on many values at a time, which allows one to generalize learned merit values to unobserved states and actions.
\begin{figure}
	\centering
	\includegraphics[width=1.0\columnwidth]{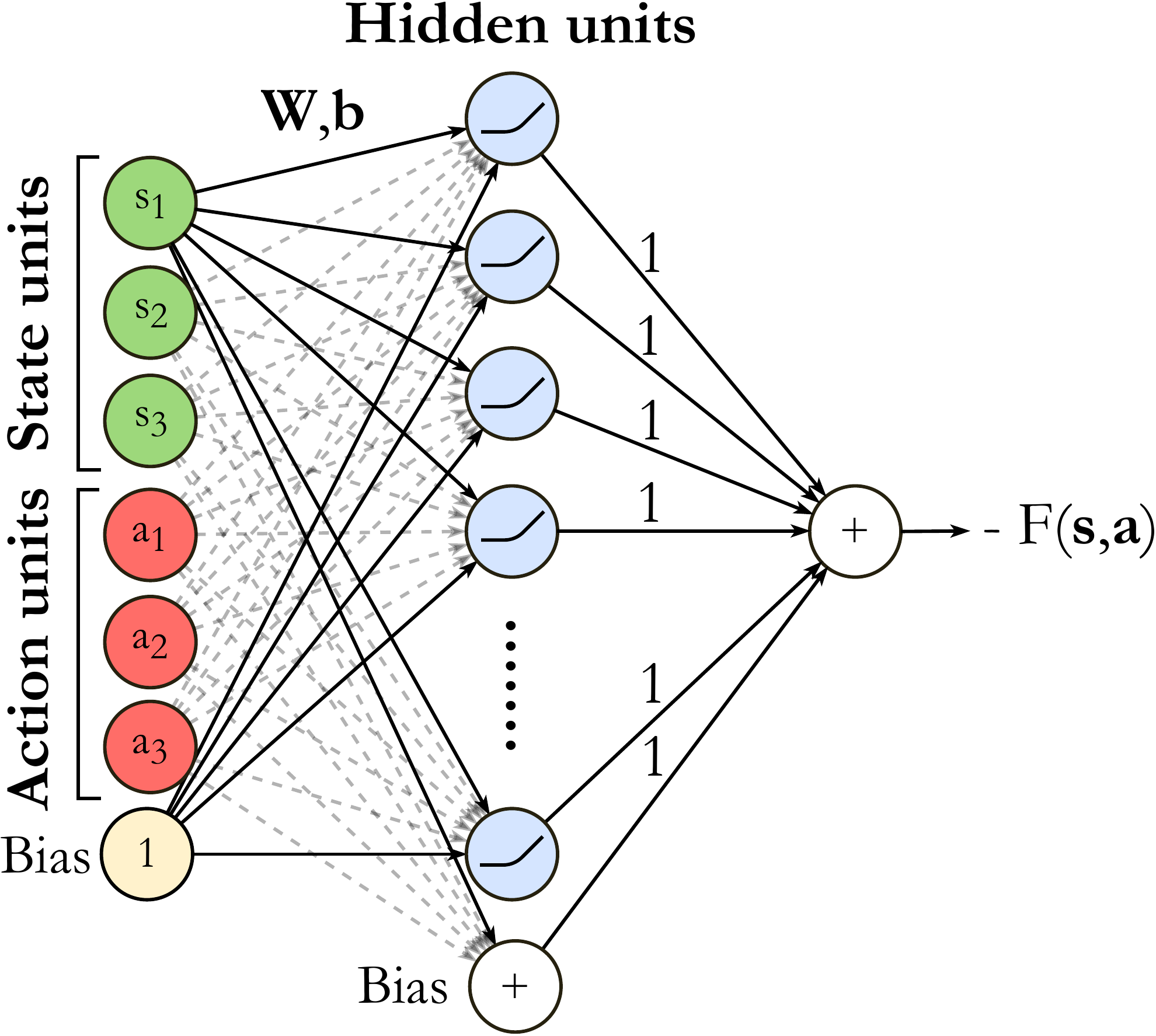}
	\caption{\textbf{A feedforward neural network computing the exact free energy $F$ of an RBM.} $W$ and $b$ are the vectors of weights and biases parametrizing the RBM. The activations of the hidden units are softplus functions (see Appendix \ref{sec:bm_energy-based}). These activations and the input bias are summed up by the output unit.}
	\label{fig:RBM-FF}
\end{figure}

At first sight, this form of energy-based function approximation may look completely unrelated to an approach relying on feedforward neural networks, such as the DQN of Mnih et al., making their structure hard to compare. Indeed, the expression of the RBM free energy given in Eq.\ (\ref{eq:free-energy}) does not seem to be expressible as a neural-network function, i.e., linear combinations of (input) variables interlaid by non-linear activation functions. But a quick derivation, given in Appendix \ref{sec:bm_energy-based}, shows that the free energy of an RBM can be evaluated \emph{exactly} by a shallow feedforward neural network taking the form of Fig.\ \ref{fig:RBM-FF}. We come back to this interesting connection in the next subsection when we consider deep neural networks.\\

\textbf{Updating the merit function. }Updates on an RBM are typically performed through gradient descent applied on the weights $\bm{\theta}$ parametrizing its (free) energy. When used in an unsupervised learning scenario, the descent direction is given by the gradient of the Kullback-Leibler divergence  (also known as relative entropy) between the RBM distribution $P^{\bm{\theta}}(\bm{v})$ and the empirical data distribution $\Hat{p}(\bm{v})$ \cite{hinton12}\footnote{Under maximum likelihood learning. In practice, a different gradient is used to train RBMs, derived by the \emph{contrastive divergence} method \cite{welling02,carreira05}. Even though this gradient is computed more efficiently, it can lead to very large approximation errors in the learned distributions \cite{fischer11}}. In a reinforcement learning scenario, however, we do not have direct access to (samples from) the optimal policy. Instead, the agent receives feedback from the environment in the form of delayed rewards. That is, feedback is given only after long sequences of actions, preventing the direct evaluation of any measure of distance to the optimal policy. For this reason, the update rule of the RBM used in RL is derived from the update rule of the merit function that the model approximates.
To explain what we mean, we detail how this is done in the case of $Q$-learning, which has a tabular update rule \cite{sutton98} given by:
\begin{multline}\label{eq:q-learning}	
	-F^{\bm{\theta}}(\bm{s^{(t)}},\bm{a^{(t)}}) \leftarrow -(1-\alpha) F^{\bm{\theta}}(\bm{s^{(t)}},\bm{a^{(t)}}) \\+ \alpha [r^{(t+1)} - \gamma\max_{\bm{a}} F^{\bm{\theta}}(\bm{s^{(t+1)}},\bm{a})]
\end{multline}
where $- F^{\bm{\theta}}$ corresponds to a $Q$-function approximation in this case and $\alpha \in [0,1]$ is the so-called learning rate. Notice that this expression can be viewed of the form
\begin{equation}\label{eq:current-target}
\textrm{current} \leftarrow (1-\alpha)\ \textrm{current} + \alpha\ \textrm{target}
\end{equation}
where the target value is an approximation of the optimal $Q$-value (expressed in its bootstrapped expansion in terms of the current reward and the discounted $Q$-value of the next timestep) that becomes exact upon convergence to the optimal policy.\\
Stated differently, this update aims at decreasing the so-called \emph{temporal difference} error, i.e., the difference between target and current value:
\begin{equation*}
	\mathcal{E}^{\bm{\theta}}_\textrm{TD} (t) = r^{(t+1)} - \gamma\max_{\bm{a}} F^{\bm{\theta}}(\bm{s^{(t+1)}},\bm{a}) + F^{\bm{\theta}}(\bm{s^{(t)}},\bm{a^{(t)}})
\end{equation*}
Analogously, we can derive similar expressions for SARSA:
\begin{equation*}
	\mathcal{E}^{\bm{\theta}}_\textrm{TD'} (t) = r^{(t+1)} - \gamma F^{\bm{\theta}}(\bm{s^{(t+1)}},\bm{a^{(t+1)}}) + F^{\bm{\theta}}(\bm{s^{(t)}},\bm{a^{(t)}})
\end{equation*}
and PS (see Appendix \ref{sec:bm_energy-based}):
\begin{equation*}
	\mathcal{E}^{\bm{\theta}}_\textrm{PS} (t) = \widetilde{r}^{(t+1)} + \gamma_\textrm{ps}F^{\bm{\theta}}(\bm{s^{(t)}},\bm{a^{(t)}})
\end{equation*}

It is then straightforward to derive an update rule for the weights of the RBM ($\bm{\theta^{(t+1)}} = \bm{\theta^{(t)}} + \Delta\bm{\theta}$) by performing gradient descent on the squared error $\left(\mathcal{E}^{\bm{\theta}} (t)\right)^2$:\footnotetext{The approximation comes in SARSA/Q-learning from considering $Q(\bm{s}^{(t+1)},\bm{a})$ constant with respect to $\bm{\theta}$, which is an assumption made in temporal difference learning.}\addtocounter{footnote}{-1}
\begin{align}\label{eq:update-rule_shallow}
	\begin{split}
    \Delta \bm{\theta} &= - \frac{\alpha}{2}\nabla_{\bm{\theta}}\left(\mathcal{E}^{\bm{\theta}} (t)\right)^2\\
    &= - \alpha\mathcal{E}^{\bm{\theta}} (t)\nabla_{\bm{\theta}}\mathcal{E}^{\bm{\theta}} (t)\\
    &\approx\footnotemark - \alpha\mathcal{E}^{\bm{\theta}} (t)\nabla_{\bm{\theta}} F^{\bm{\theta}}(\bm{s^{(t)}},\bm{a^{(t)}})
	\end{split}
\end{align}
It turns out that, in the case of RBMs, these update values can be computed efficiently since the gradient of their free energy $\nabla_{\bm{\theta}} F^{\bm{\theta}}(\bm{v})$ has a tractable expression:
\begin{equation*}
    \begin{cases}
		\frac{\partial F^{\bm{\theta}}(\bm{v})}{\partial w_{ik}} = -v_i\langle h_k\rangle_{P(\bm{h}|\bm{v})}\\
		\frac{\partial F^{\bm{\theta}}(\bm{v})}{\partial b_{i}} = -v_i\ ;\ \frac{\partial F^{\bm{\theta}}(\bm{v})}{\partial b_{k}} = -\langle h_k\rangle_{P(\bm{h}|\bm{v})}
    \end{cases}
\end{equation*}
where $\langle h_k\rangle_{P(\bm{h}|\bm{v})} = 1/\left(1+\exp\left(-\sum_{i} w_{ik}v_i - b_k\right)\right)$.\\

Hence, these formulas allow us to adapt RL update rules to train the free energies of RBMs as merit function approximators.

\subsection{From shallow to deep: the deep energy models for RL\label{sec:going-deep}}

\begin{figure*}
	\center
	\includegraphics[width=1\textwidth]{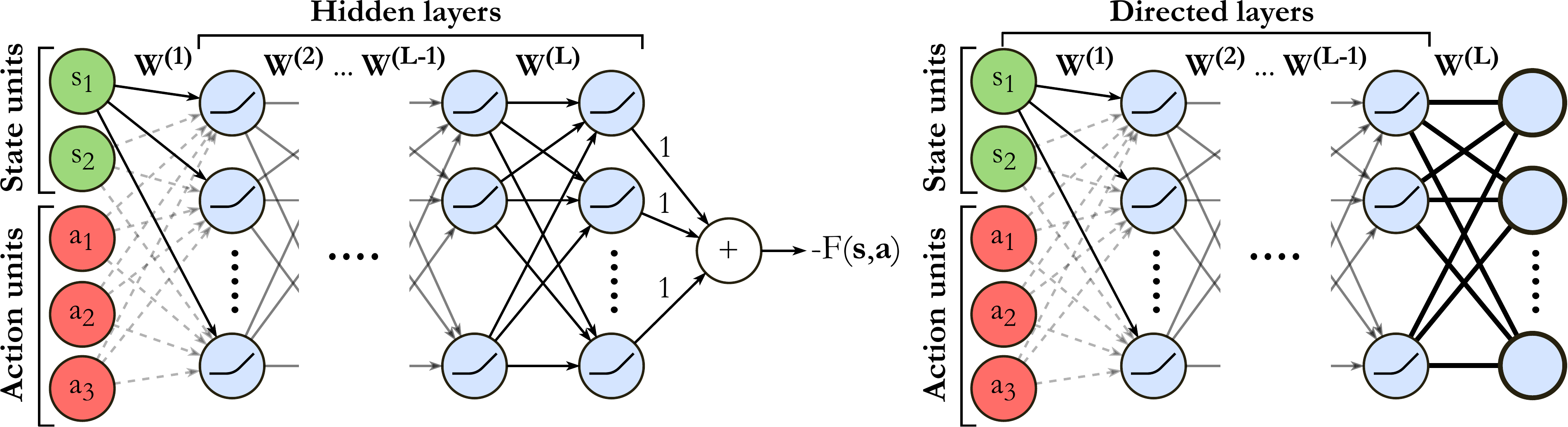}
	\caption{\textbf{The DEBN we consider (left) and its corresponding DEM (right).} The DEM is obtained by removing the summation unit at the output of its corresponding DEBN and by replacing its last hidden layer with a layer of stochastic units of the same size. Just as in an RBM, these stochastic units can take binary values depending on the input given by their undirected connections. Biases are not represented for compactness.}
	\label{fig:DEBN-DEM}
\end{figure*}

\textbf{Universal function approximation. }The RBM-based model presented in the previous subsection fulfills one of the goals behind generalization in RL: having a parametrized representation of the merit function that allows us to generalize learned behavior to previously unobserved states and actions. A natural consideration that arises from this approximation method is a possible compromise in representational power: a tabular method can represent arbitrary real-valued functions over the state-action space (assuming it is finite), but there is no guarantee, a priori, that the RBM free energy can represent these functions with finitely many hidden units. Fortunately, a result from Le Roux \& Bengio \cite{leroux08} shows that RBMs are universal function approximators: an RBM can represent arbitrary probability distributions over its visible units with any given precision, using a finite number of hidden units that depends on this distribution and the desired precision. This same result can be extended to prove that RBMs can represent arbitrary conditional probability distributions (i.e., policies) \cite{montufar15}. Unfortunately, these universality results also suggest that certain policies may need a number of hidden units that is linear in the number of state-action pairs to be represented precisely \cite{montufar11}\footnote{Representing arbitrary merit functions using an RBM free energy (Fig.\ \ref{fig:RBM-FF}) has a similar cost in the number of hidden units \cite{martens13}.}. This leads to a neural network parametrized by a \emph{very large} number of weights, rendering it difficult to train and likely to get stuck on suboptimal policies for general environments.\\

\textbf{Deep energy-based network (DEBN). }In contrast, deep neural networks are strictly more expressive: the Vapnik–Chervonenkis dimension \cite{vapnik71}, a measure of function approximators' expressive power\footnote{Originally used to quantify the expressive power of classifiers, its definition can also be extended to regression models.}, of neural networks with ReLU activation functions (similar to the softplus function we consider here, see Appendix \ref{sec:bm_energy-based}) grows linearly with their number of hidden layers when keeping their total number of weights fixed \cite{harvey17}. This added expressive power can give some intuition on how come deep neural networks require only logarithmically many artificial neurons compared to their shallow counterparts to represent a certain general class of functions \cite{liang16}. Given the greater representational power of deep neural networks, we come back to the connection between the RBM-based approach and feedforward neural networks made in the previous section. To enable the representation of a richer class of merit functions, we extend the free energy network in Fig.\ \ref{fig:RBM-FF} to a deep architecture by duplicating its hidden layer. The resulting neural network represents the free energy (i.e., the energy after tracing out the \emph{stochastic} hidden units) of a so-called deep energy model (DEM) (see Fig.\ \ref{fig:DEBN-DEM}), first introduced by Ngiam et al.\ \cite{ngiam11}. As pointed out by the authors, this model is different from both deep Boltzmann machines (DBMs) \cite{salakhutdinov09} and deep belief networks (DBNs) \cite{hinton06} due to its \emph{directed deterministic} connections between its first layers, followed by a \emph{single} undirected stochastic layer. Conceptually, this model transforms state-action pairs given as input into non-linear features that are then treated as the visible units of an RBM. This construction preserves the tractability of the free energy inherited from the RBM (as opposed to DBMs and DBNs) while representing high-level abstractions of state-action pairs. Hence, DEBNs constitute models with great representational power without the associated drawback of having a large number of weights to train.\\

\textbf{Updating the merit function. }Same as for the RBM, updates of the merit function are performed through gradient descent on the weights of the DEBN. However, due to the deep architecture of the neural network, we have to resort to backpropagation \cite{rumelhart88} to propagate derivatives of the temporal difference error through the multiple layers of the DEBN. It can be easily shown that backpropagation of the temporal difference error on the shallow DEBN considered in the previous subsection leads to the same update rule as in Eq.\ (\ref{eq:update-rule_shallow}).

\textbf{Stable learning. }Prior to the DQN of Mnih et al.\ \cite{mnih15}, it was widely believed that non-linear function approximation using (deep) neural networks was not suitable for learning value functions due to its high instability, which is likely to lead to divergence from the optimal value function \cite{tsitsiklis97}. This instability originates from the approximation of the target value function (see Eq.\ (\ref{eq:current-target})) that leads to a non-stationary temporal difference error and from the important correlations between the training samples generated by the agent. These two features, arising in RL yet not in supervised learning, both violate the assumptions behind convergence proofs in (un)supervised learning. Mainly two additions enabled stable learning with DQNs: 1.\ an \emph{experience replay memory} storing (recent) samples of learning experience, used to generate decorrelated training data for the neural network; 2.\ a \emph{target network}, i.e., a copy of the trained DQN updated at a slower rate, used to supply consistent targets for the temporal difference error during updates of the primary DQN. Naturally, we also rely on these tools to stabilize learning with our DEBNs.\\

A description of the full learning algorithm is presented in Algorithm \ref{alg:hybrid-EBRL} of Sec.\ \ref{sec:hybrid-EBRL}.

\section{Performance comparison between various models\label{sec:numerical-simulations}}

Having introduced how energy-based models can be used for function approximation, we now justify their use in RL. To that end, we study the differences in performance of both standard and energy-based function approximation methods for RL when facing environments with large state and action spaces. In a similar fashion to the analysis of Fu et al.\ \cite{fu19}, we establish several hypotheses regarding the performance of these models with respect to a number of learning aspects appearing in RL. These hypotheses are based on structural considerations along with previous experimental evidence. We then describe and carry out experiments to test our hypotheses.\\

Let us list our hypotheses, numbered (1) to (4) below.\\

We start by demonstrating that, similarly to a standard RL scenario based on feed-forward NNs, the use of function approximation and especially deep models is also beneficial when considering an energy-based approach.\\
(1) We expect very simple (shallow) energy-based networks to quickly learn in environments with large state spaces but simple target merit functions, while tabular methods would fail to learn as fast. This is due to the ability of (even shallow) neural networks to capture the simple correlations of the merit function, enabling it to generalize on the large input space.\\
(2) We expect shallow neural networks with a restricted number of hidden units to be unable to learn complex merit functions. A separation in learning performance between shallow and deep energy-based networks should then be noticeable when dealing with more complex task environments.\\
We test hypotheses (1) and (2) using common benchmarking tasks, introduced in Sec.\ \ref{sec:tabular-shallow-deep}.\\

In the second part of this section, we investigate the advantage of energy-based function approximators when facing large action spaces. Our hypotheses stem from previous experimental results showing the improved exploration performance of energy-based policies when facing multimodal reward functions in RL \cite{haarnoja17,dulac15}, as well as results in generative modeling showing the ability of energy-based models to produce robust multimodal representations of large image datasets \cite{du19}.\\
(3) We postulate that DEBNs are better suited than standard DQNs (see Section \ref{sec:nn-generalization}) to learn in RL environments with large action spaces: given the same number of parameters, we expect DEBNs to learn richer representations that allow for better generalization performance when learning merit functions. This should be due to the ability of the DEBN architecture to easily learn complex state-action correlations by having non-linear hidden layers that act on both states and actions. Since in RL an agent has to rely on its own approximation of the merit function to explore the environment and update itself, we expect the RL learning setting to amplify poor generalization performance in DQN-type architectures. We isolate two learning aspects that are specific to a RL scenario and that we believe characterize these amplifying effects:
\begin{itemize}[leftmargin=4mm]
	\item \emph{Policy-sampling. }We refer to the first aspect as the policy-sampling property of the training process.
It comes from the fact that an agent's policy influences the experienced data that constitutes its training set. Since this policy is derived from the agent's approximation of the merit function, a model that lacks to generalize properly its learned merit function will then render the training set suboptimal and hence limit the agent's learning performance.
	\item \emph{Reward-discounting. }The second aspect comes from the update rules of RL. As opposed to a supervised learning scenario, the target values used to train the merit function approximators are not fixed but themselves computed using an approximation of the merit function. For instance, for SARSA: $$\mathcal{E}^{\bm{\theta}}_\textrm{TD'} = \left(r + \gamma M^{\bm{\theta}}(\bm{s'},\bm{a'})\right) - M^{\bm{\theta}}(\bm{s},\bm{a}),$$ and for PS: $$\mathcal{E}^{\bm{\theta}}_\textrm{PS}= \widetilde{r} - \gamma_\textrm{ps}M^{\bm{\theta}}(\bm{s},\bm{a}).$$ The trained model has then to cope with a noisy and inaccurate target function\footnote{Note that, combined with the policy-sampling effect, we also have that the target function is moving, as the policy of the agent changes over training.}. These properties are a consequence of statistical fluctuations of sampled trajectories and their associated rewards (main source of error in PS), as well as an imperfect approximation of the merit function associated to the current policy (main source of error in SARSA/Q-learning).
\end{itemize}
To test our hypothesis (3), we design two experiments that isolate each of these learning aspects. In each experiment, we test the generalization performance of DEBN and DQN agents for increasing dimensions of the action space, for which we expect to see a growing separation in performance between the two models. Our experiments have the specificity that the target merit function to be learned by the agents is fixed and clearly defined, which allows us to measure the error of the agents' approximations, that we refer to as generalization error. We also run a control experiment verifying that the learning aspect is indeed responsible for the observed separation.

(4) Our last hypothesis posits that these gaps in generalization abilities manifest in a separation in learning performance for a sufficiently complex RL environment, e.g., with a high number of good sub-optimal policies. We give evidence of this claim in such an RL environment where both effects highlighted above contribute to the learning performance.\\

The experiments that test hypotheses (1) and (2) are presented in Sec.\ \ref{sec:tabular-shallow-deep}, while the experiments that test hypotheses (3) and (4) are presented in Sec.\ \ref{sec:DEBN-DQN}. The hyperparameters used in these experiments are listed in Appendix \ref{sec:app_num_sim}.

\subsection{Shallow v.s.\ tabular and deep v.s.\ shallow\label{sec:tabular-shallow-deep}}

\begin{figure*}
	\subfloat[][]{\label{fig:grid}\includegraphics[width=0.375\textwidth, valign=c]{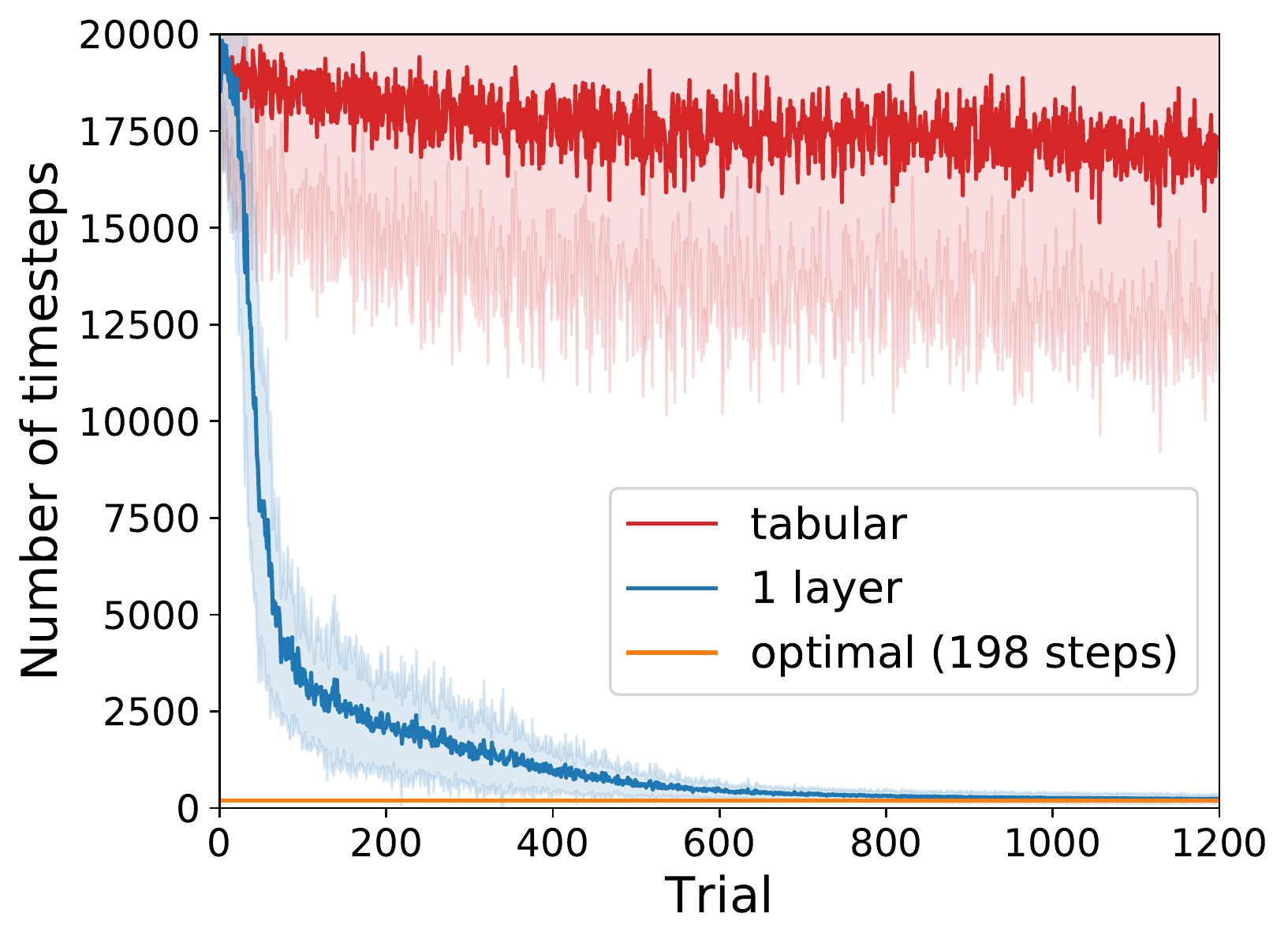}}\hspace{0em}%
	\subfloat[][]{\label{fig:cart}\includegraphics[width=0.375\textwidth, valign=c]{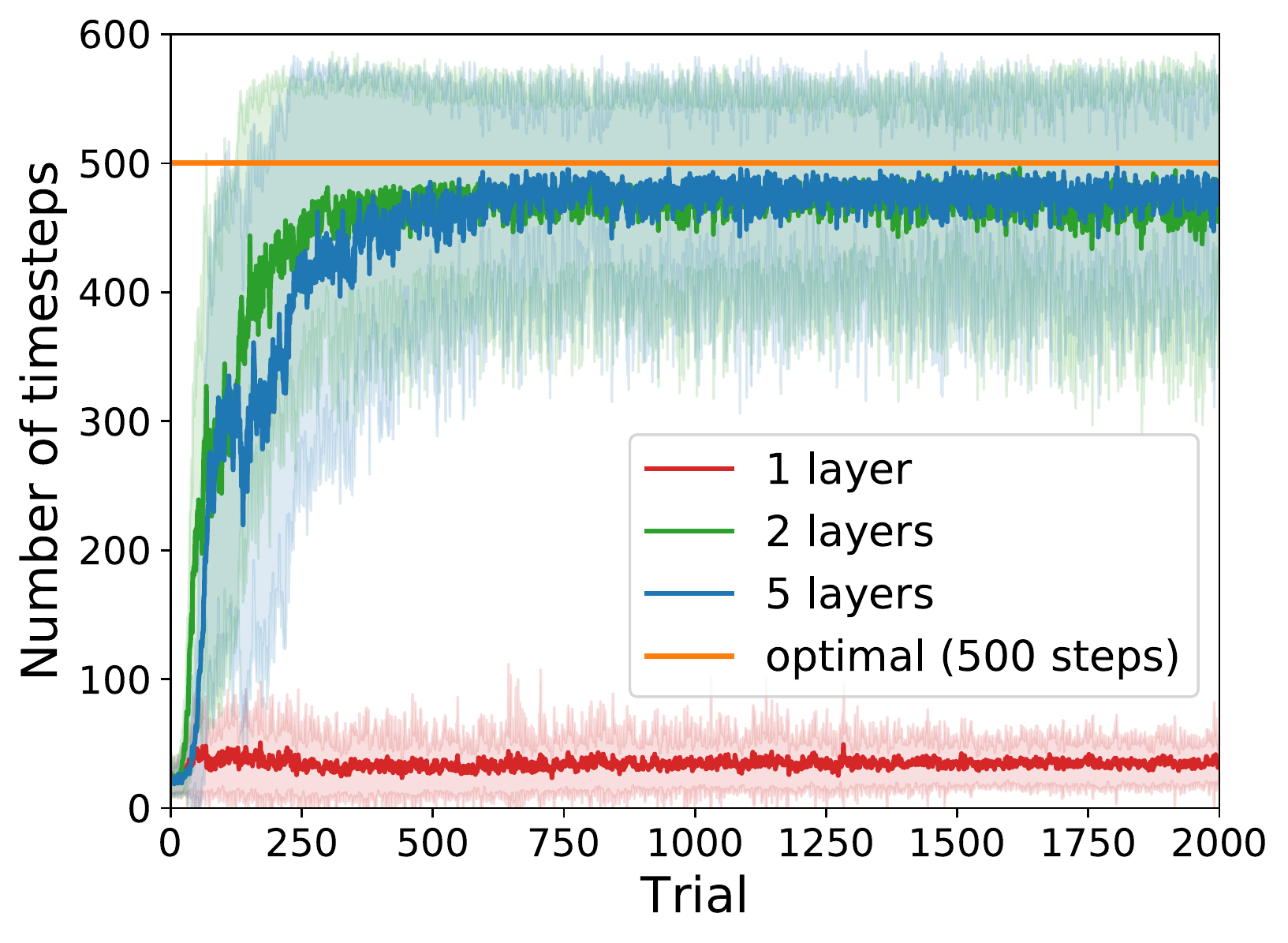}}%
	\subfloat{\includegraphics[width=0.25\textwidth, height=0.25\textwidth, valign=c]{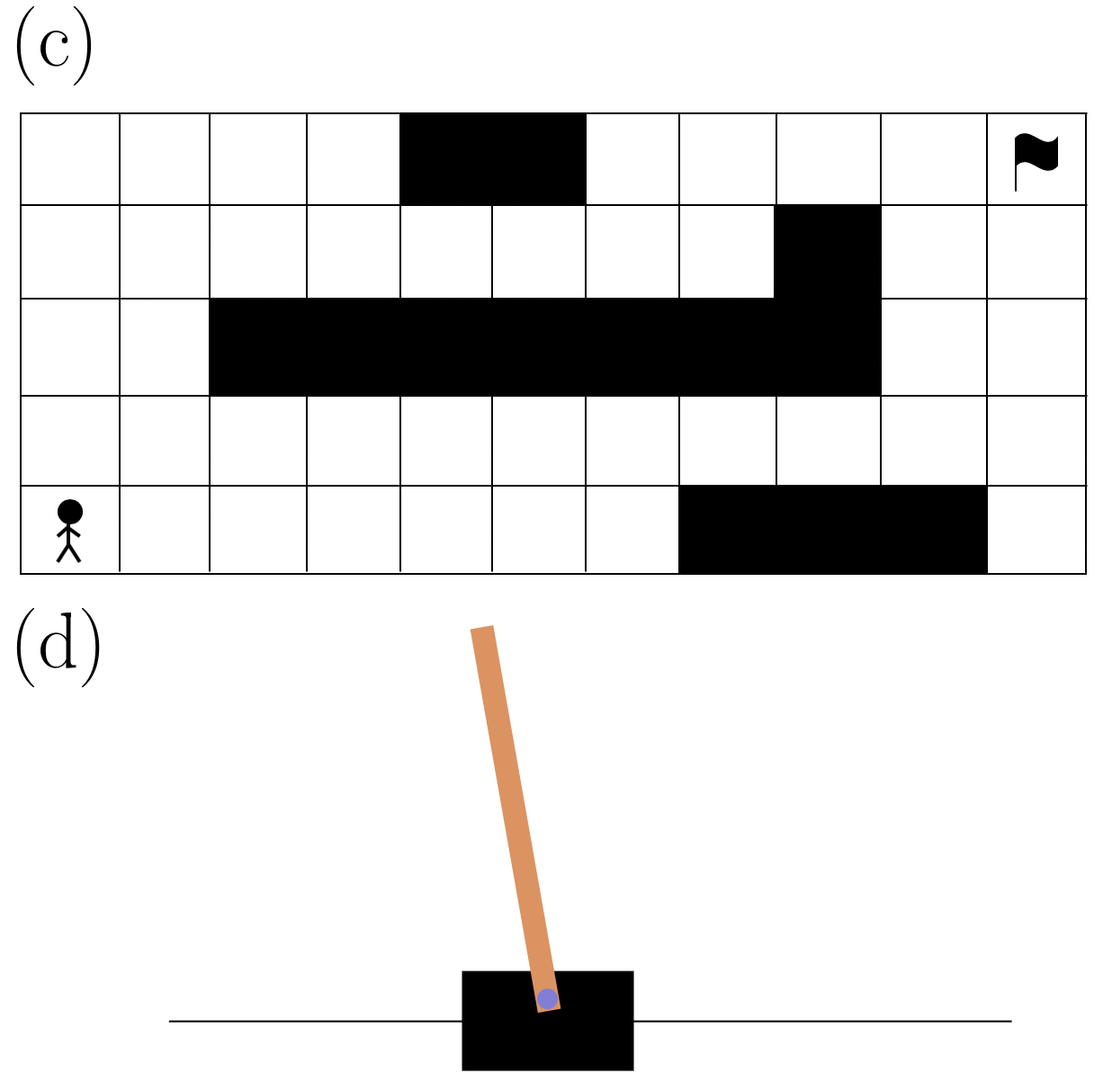}}
  \caption{\textbf{Numerical evidence of the advantage of (deep) function approximation models in a $100\times 100$ GridWorld task and the CartPole-v1 environment.} (a) and (b) In both plots, the optimal performance ($198$ and $500$ steps respectively) is indicated by the orange line. The remaining curves indicate the average performance and standard deviations of $50$ agents trained using a PS update rule. (a) In the $100\times 100$ GridWorld task, the performance of tabular agents (red curve) remains close to that of a random walk on the grid while agents with a shallow neural network (blue curve) reach close to optimal performance. (b) In the Cartpole-v1 environment, agents with a shallow neural network (red curve) exhibit close to random behavior while neural networks with at least $2$ hidden layers achieve close to optimal performance. (c) An illustrative depiction of the GridWorld environment. Note that this image is purely illustrative as it does not reflect the actual size and layout of the grid used in the simulations. (d) A game-screen image of the CartPole-v1 environment.}
	\label{fig:fa-demonstration}
\end{figure*}

In this subsection, we numerically demonstrate the importance of function approximation and deep models in RL environments with large and continuous state spaces. The results are presented in Fig.\ \ref{fig:fa-demonstration}.\\

\textbf{GridWorld simulations. }In order to demonstrate the advantage of shallow energy-based networks over tabular methods, we compare their performance in the GridWorld benchmarking task \cite{melnikov18} (see Fig.\ \ref{fig:fa-demonstration}c).
In this task, the agent navigates through a two-dimensional array of cells by choosing, at each timestep of interaction, one of the four main cardinal directions to move from one cell to the next. The training is divided into trials. At the beginning of each trial, the position of the agent is initialized at a fixed starting cell $(0,0)$ and a reward is placed in a fixed goal cell $(n,n)$, where $n$ is the grid size.
A trial terminates if the agent reaches the goal cell or if the length of the trial exceeds $20000$ timesteps. 
In the former case, the agent receives a reward of $1$. In the latter, the agent receives a reward of $-1$. 
For the simulations presented here, we chose a $100\times100$-grid with closed boundary conditions and no obstacles such that the state space is made large while maintaining a simple optimal policy, close to being state-independent (going up or right with equal probability, except at the boundaries). This last property brings a very helpful but yet very simple way to generalize the learned merit function / policy on the large state space.
In Fig.\ \ref{fig:grid}, we demonstrate that agents using a shallow neural network architecture significantly outperform their tabular counterparts in this task. This improved performance attests the generalization capabilities provided by the neural network. The tabular methods are unable to detect and exploit the similarities between different cells.\\

\textbf{CartPole simulations. }We now show, by means of an example, that DEBNs may provide advantages over shallow architectures in more complex task environments. More specifically, we consider the OpenAI Gym CartPole-v1 environment \cite{brockman16} which features a continuous state space.
In CartPole, the goal is to balance an inverted pole-pendulum mounted on a movable cart by applying to the latter, at each timestep of interaction, a unit of force to the left or right (see Fig.\ \ref{fig:fa-demonstration}d).
A trial terminates in one of three situations: a) the angle between the vertical line above the pivot point and the pole exceeds $12$\textdegree, b) the center of mass of the cart is at a horizontal position outside the range $[-2.4,2.4]$, c) the pendulum was successfully balanced for $500$ consecutive timesteps. 
For every timestep the agent manages to balance the pole, it receives a reward of $1$. 
The state space of this task can be described by four continuous parameters: the position of the cart, the velocity of the cart, the angle between the vertical line above the pivot point and the pole, and the angular velocity of the pole. 
In Fig.\ \ref{fig:cart}, we compare the performance of architectures with one, two and five hidden layers but the same total number of weights. 
We observe that the network with a single hidden layer does not manage to learn the complex optimal policy of this task.
However, already 2 hidden layers are sufficient to achieve near-optimal performance, demonstrating the importance of \emph{deep} energy-based models in RL tasks.

\begin{figure*}
	\center
	\includegraphics[width=\linewidth]{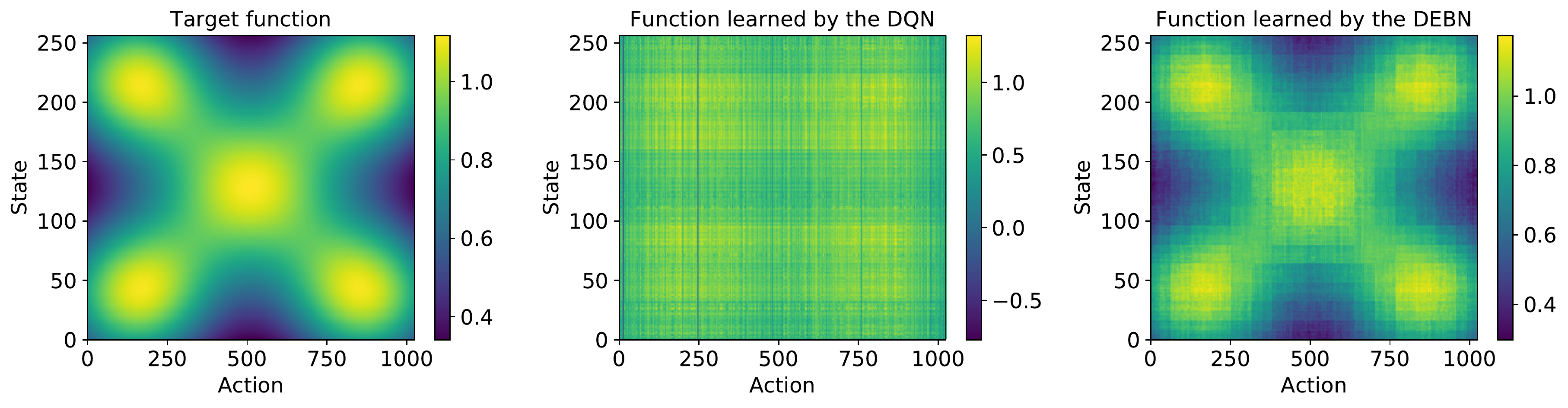}
	\caption{\textbf{The policy-sampling experiment. }At each learning step, a state is uniformly sampled from the state space (of size 256). Its corresponding binary encoding is perceived by the agent, which samples an action from a Boltzmann policy computed using its current approximation of the merit function (initially random). In the policy-sampling experiment, the target value associated to this sampled state-action pair is compared to the approximation of the neural network, resulting in the loss used to train its weights. In the analogue multidimensional regression problem, the entire row of target values associated to that state is compared to the values approximated by the network, each contributing to the loss. We additionally use mini-batches of states of size $10$ to compute average losses and improve convergence rates. The DEBNs also receive binary encodings of the actions.}
	\label{fig:on-policy}
\end{figure*}

\subsection{DEBN v.s.\ DQN\label{sec:DEBN-DQN}}

After considering RL environments with large state spaces, we now focus on environments with large action spaces, for which we expect an advantage of the DEBN architecture over DQN-like networks (see Fig.\ \ref{fig:DQN-DEM}). In order to keep a fair comparison between the two models, we constrain our neural networks in two ways. First, we set their total number of parameters to be equal whenever their performance is compared. Moreover, for these models to be practically relevant, we restricted their total number of parameters to a \emph{reasonable} bound (constant in our experiments, but in practice simply significantly smaller than the state-action size). However, these constraints combined lead to a particular regime where the separation between the two architectures is straightforward: as the size of the action space grows, the restricted parameters of the DQN tend to accumulate in its output layer (which is linear in the number of actions) until the model becomes unusable, while the DEBN architecture is significantly less affected (since its input layer can grow as slowly as logarithmically with the action space). Since this \emph{bottleneck effect} would cause DEBNs to trivially outperform DQNs, we choose our bounds on the number of parameters such that our instances of DQNs stay out of the bottleneck regime.\\

Let us now introduce in more detail the RL-specific aspects that we expect to cause a separation between the two architectures and the experiments we performed to test this claim.\\

\begin{figure*}
	\subfloat[][]{\includegraphics[width=0.5\linewidth]{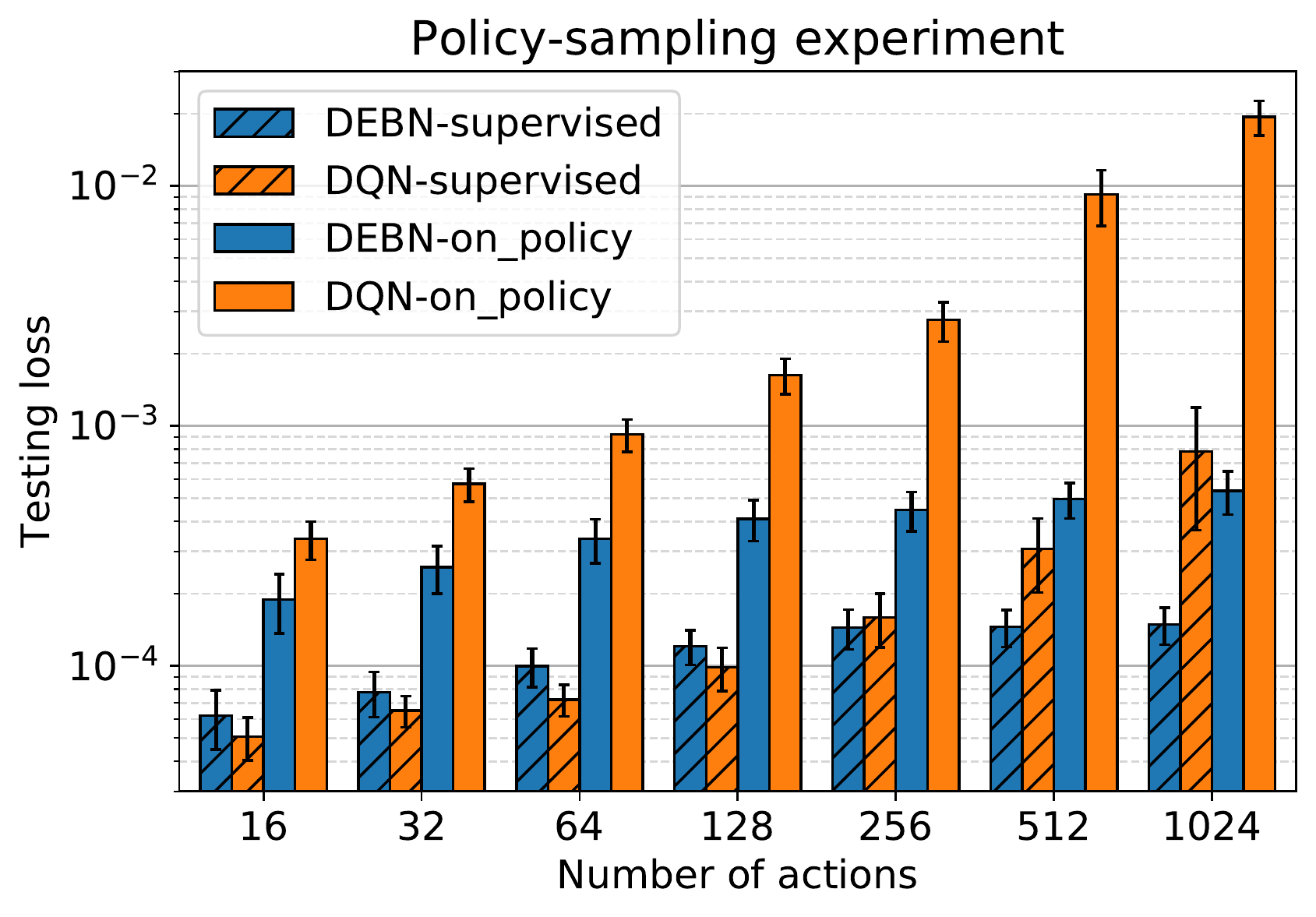}\label{fig:plot_on-policy}}\hspace{0em}%
	\subfloat[][]{\includegraphics[width=0.5\linewidth]{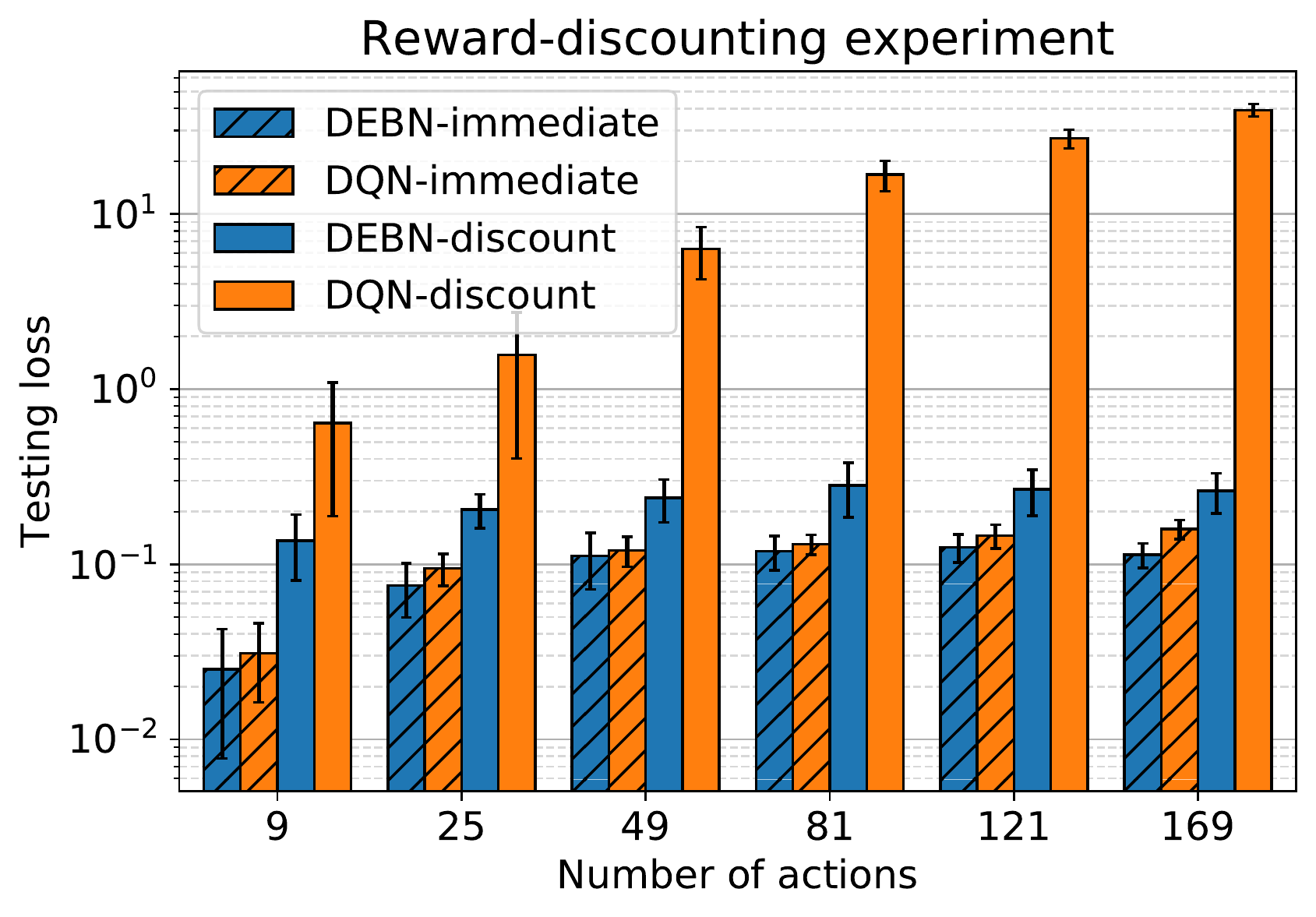}\label{fig:plot_glow}}
	\caption{\textbf{Separation in testing loss between DEBN and DQN agents in the policy-sampling and reward-discounting experiments. }(a) and (b) In both plots, each bar indicates the average testing loss (i.e., generalization error) of $20$ agents in their respective tasks. The number of weights of both neural networks is kept constant and equal in each experiment. Hatched bars are associated to the control experiments where the learning setting doesn't have the tested property. (a) In the policy-sampling experiment, we test the ability of both architectures to faithfully represent a multimodal merit function (depicted in Fig.\ \ref{fig:on-policy}) when training samples are generated by the current approximation of said function. The testing loss corresponds to the average smoothed $l_1$-loss over all merit values associated to $1000$ randomly sampled states after $5000$ learning steps. (b) In the reward-discounting experiment, we test the ability of both architectures to cope with the statistical fluctuation on their target merit values brought by a SARSA update rule. The testing loss corresponds to the average $l_2$-loss over all merit values associated to $100$ randomly sampled states after $2000$ learning steps ($20000$ total sampled transitions and an update period of $10$).}
\end{figure*}

\textbf{Policy-sampling. }The first aspect comes from the policy-sampling property of the training process: at each interaction step, an agent samples one action $\bm{a}$ from a policy specified by its current approximation $M^\theta(\bm{s},\ldots)$ of the merit function, and updates this approximation according to a loss that is (in general) independent of the values $M^\theta(\bm{s},\bm{a'}),M(\bm{s},\bm{a'})$ for different actions $\bm{a'}$. As opposed to a supervised (multidimensional) regression task, where, for each sampled state $\bm{s}$, the agent could compute its loss on the entire vector $\left(M^\theta(\bm{s},\bm{a}_1), \ldots, M^\theta(\bm{s},\bm{a}_{\abs{\mathcal{A}}})\right)$, the policy-sampling aspect can lead to \emph{masking} effects on sampled target values. In particular, when the current approximation incorrectly assigns large values to certain actions, this approximation error skews the sampling distribution in an inopportune way and prevents learning the merit function faithfully on the entire state-action space. Since sampling from the correct data distribution is already an imperative to get good generalization error in regression tasks, it should as well be important in the more general setting of policy-sampling. We expect these masking effects to be particularly relevant when the target merit function is multimodal, as this multimodality (in general) adds complexity to the state-action correlations to be learned. Intuitively, since DEBNs take some (compact) representation of individual actions as input, they are able to construct non-linear features of both states and actions. Hence, this makes DEBNs more likely to lean complex state-action correlations/masks, in a way that is not possible for a network that can only construct state features in its hidden layers. To investigate the magnitude of this effect, we compare the learning performance of these two models in the policy-sampling learning task described in Fig.\ \ref{fig:on-policy}. Interestingly, we observe that the DQNs learn a very noisy approximation of the target function, and even completely fail to learn one of the five modes of the target function. In order to give more quantitative results, we test the generalization performance of the two models on the same task with different levels of discretization of the action space. As we can see in Fig.\ \ref{fig:plot_on-policy}, the separation in testing loss appears to increase linearly with the size of the action space. Here, one could argue that this separation is solely caused by the bottleneck effect, as previously discussed. However, we show that this separation indeed originates from the policy-sampling aspect of training by running a control experiment on the analogue supervised learning task, where we indeed observe that both architectures have the same performance up to $256$ actions (after which the bottleneck effect might be responsible for part of the separation).
\\

\begin{figure*}
\centering
\includegraphics[width=\linewidth]{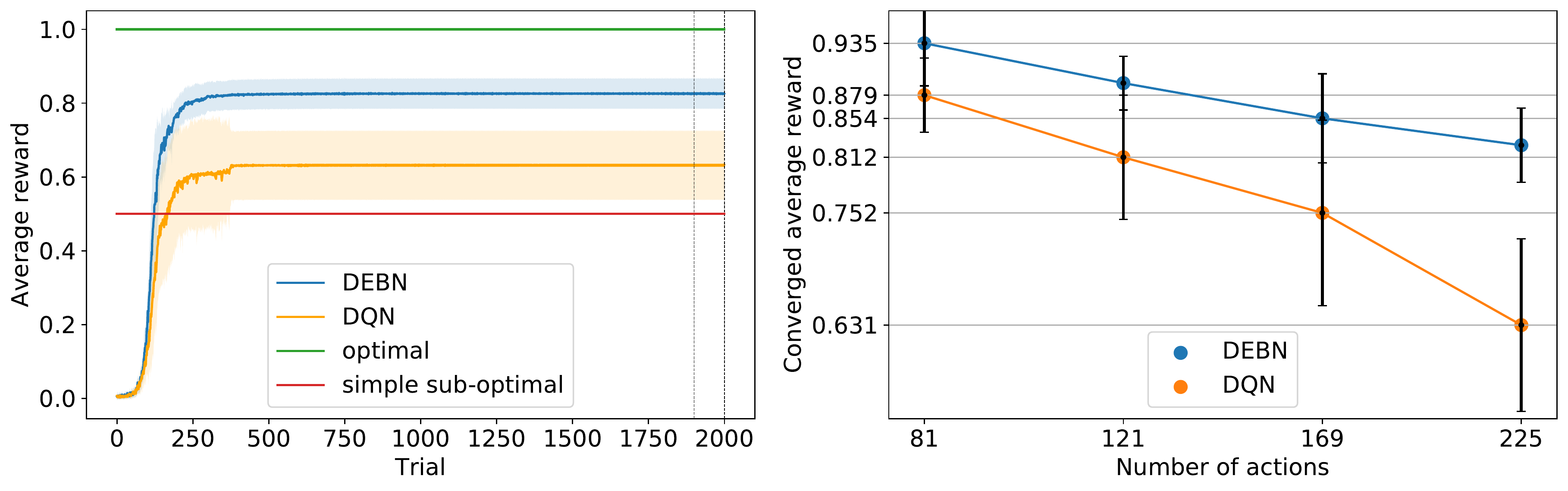}
	\caption{\textbf{Separation in average rewards between DEBN and DQN agents in a RL task. }Each dot on the right plot indicates the average reward of 20 agents over $100$ trials after training for $1900$ trials of length $100$ interactions. The left plot shows the evolution of this average performance over the course of training for the toy circular-GridWorld task with $225$ actions. Very simple sub-optimal policies can achieve an average reward of $0.5$, represented by the red line. The optimal policy achieves an average reward of $1$, indicated by the green line.}
	\label{fig:full-RL}
\end{figure*}

\textbf{Discounted rewards. }The second aspect of RL for which we expect a separation comes from the reward discounting introduced by RL update rules. For both projective simulation (PS) and value-based methods (VBMs), merit values do not correspond to the immediate rewards of state-action pairs but also include estimates of future rewards discounted by a discount (or glow) factor (see bootstrapped expression of Eq.\ (\ref{eq:current-target})). This is an intrinsic property of RL update rules as they all aim at optimizing long-term rewards. For VBMs, these estimates of future rewards are given by a sample of the immediate reward and the merit value of the next experienced state as approximated by the model; while, for PS, these correspond to the rewards collected during one sampled trajectory in the environment. For \emph{stochastic} policies and/or environments, the sampled terms lead in general to statistical fluctuations in the target merit values on which the agents compute their loss. As for the approximated terms, errors in the approximation of the merit function lead to noise in the target values. We refer to both these effects as fluctuations of the merit values. To study the influence of this RL aspect in isolation from other sources of fluctuations, we consider the following task: an agent is given access to sample rewards experienced by a fixed stochastic policy (also called \emph{off}-policy) in a toy-environment where states are randomly sampled at each interaction step, and needs to learn its \emph{discounted} merit function. The constraints of this task effectively cancel the policy-sampling aspect considered in the previous experiment as well as temporal fluctuations induced by the stochastic policy at hand. The setting for which we observed the highest separation between the two models was for an off-policy very close to optimal (performing a rewarded action with probability $0.99$ and a random action otherwise) and with a discount factor of $0.9$. Again, to assess the influence of reward discounting on this separation, we run a control experiment where the discount factor is set to $0$. The results of these simulations are presented in Fig.\ \ref{fig:plot_glow}. Note that we observed this separation despite the use of the learning mechanisms that help counter the statistical fluctuation of target values: notably, the experience replay memory and, more importantly, the target network (see Sec.\ \ref{sec:going-deep}).\\

\textbf{A complete RL environment. }In the previous experiments, we have studied the separations between DEBNs and DQNs in terms of their generalization error when learning merit functions for restricted RL environments. We now give evidence of a separation in learning performance (i.e., average collected rewards) in a complete RL environment. Here, agents learn to solve a toy GridWorld task with a special transition function by directly interacting with it (policy-sampling) and are subject to the statistical fluctuations induced by a discount factor. This task environment has the same description in terms of states and actions as the one used for the discounted-rewards experiment, but differs in its transition function such that it adds more structure to the task. The description goes as follows: an agent is situated in a one-dimensional GridWorld with $N$ cells, enumerated $0$ to $N-1$, circular boundary conditions and a goal state at position $\lfloor N/2 \rfloor$. At each interaction step, the agent is at a given position $n$ and can do one of $(\lfloor N/2 \rfloor+1)^2$ actions. Each action is described by two sub-actions: moving $i$ steps to the left and $j$ steps to the right for $i,j$ in $0, \ldots, \lfloor N/2 \rfloor$. Call $n'$ the position $n+j-i \text{ mod } \lfloor N/2 \rfloor$. From the actions for which $n'$ corresponds to the goal position $\lfloor N/2 \rfloor$, only the shortest (i.e., smallest $i+j$) is rewarded with a reward of $1$. All other actions are not rewarded. Now for the transition function: if $n'=\lfloor N/2 \rfloor$, the agent is moved to the position $n+1$ (or $n+2$ if that happens to be the goal position), otherwise it is moved to the position $n'$. This choice of transition function, while quite artificial, leads to the interesting property that this task allows for several sub-optimal policies, e.g., the agent runs into a loop of $l\in\{2,\ldots,N-2\}$ positions with only one non-rewarded transition, achieving an average reward of $1-1/l \geq 0.5$. The only optimal policy on the other hand, runs through a loop of length $N-1$ and achieves an average reward of $1$. This makes it a very challenging task if one seeks the optimal policy, and in part explains why the separation in performance between the two architectures appears small. We do however observe in the results of our experiment (see Fig.\ \ref{fig:full-RL}) an increasing separation in average rewards with respect to the number of actions in the environment. This illustrates how the gaps in generalization performance observed for the RL-specific learning aspects discussed above translate into a separation in learning performance when these aspects act in concert.\\

\textbf{Discussion. }Due to the simplicity of the chosen tasks, the separations observed in our experiments with DEBNs and DQNs may get less significant in different learning regimes. Similarly to tabular methods being able to learn optimal policies given arbitrarily long learning time, we expect the observed separations to get smaller when looking at longer learning times, or, for instance, when increasing the total number of weights of the networks and hyperparameters such as the update period of the target networks. However, throughout our exploration of relevant regimes, we did not observe any inversion of the performance of both networks, indicating that a large separation might as well be present in broader regimes for more complex environments. We used a PS update rule for the simulations on the benchmarking tasks and a SARSA update rule for the separation results in the reward-discounting and full-RL experiment. We however believe that the same experiments with interchanged update rules exhibit a similar behavior, possibly with different separations. This claim is supported by an extensive numerical study that showed that the tabular versions of PS and VBMs behave similarly in this type of benchmarking tasks \cite{melnikov18} and additional numerical simulations that were not included in this paper for conciseness.
In order not to penalize the DEBN agents with approximation errors in sampling, we used \emph{exact} sampling (see next section) from the policy of the agents in the experiments above. Note that the quantum speed-up presented in the next section applies to problems with large action spaces and hence can not be numerically demonstrated as classical simulations are intractable.

\section{Hybrid deep energy-based RL with quantum enhancements\label{sec:hybrid-EBRL}}

So far, we have shown that one instance of energy-based models, the DEBNs, can provide a learning advantage over standard DQNs in complex RL environments. However, this learning advantage comes at a cost: the inefficiency of sampling from energy-based policies, needed to act on the environment and train the RL agent. It is easy to see that \emph{exact} sampling is intractable, since it requires computing the merit values $M^{\bm{\theta}}(\bm{s},\bm{a})$ for all possible actions $\bm{a}$ given a state $\bm{s}$ to evaluate the expression of Eq.\ (\ref{eq:energy-based_policy}), which costs $\mathcal{O}(\abs{\mathcal{A}})$ feed-forward evaluations in the case of DEBNs. Moreover, this hardness of sampling even extends to \emph{approximate} sampling: it is actually an NP-hard problem to generate samples from a distribution that is \emph{provably} close in total variation distance from a general RBM distribution \cite{long10}. We address this computational bottleneck by applying quantum algorithms for approximate sampling from Gibbs distributions, that is algorithms that can offer quantum speed-ups to \emph{heuristic} sampling methods. We investigate the potential of using both (long-term) fault-tolerant algorithms with provable quadratic speed-ups \cite{harrow20,yung12, van20}, as well as near-term approximate optimization approaches \cite{chowdhury20,wu19,wang20}. 

We have focused in our simulations on DEBNs rather than other energy-based models due to their interesting properties: they offer deep NN-based representations of the merit function while allowing efficient computation of this function and its gradient (through a forward and backward evaluation of the DEBN). A drawback of DEBNs is that, in order to be evaluated in quantum algorithms, one needs to carry out the feed-forward computation coherently, which is out of reach of near-term quantum devices. Hence, this encourages us to additionally consider different deep energy-based models such as deep Boltzmann machines (DBMs) as well as their quantum extensions (QBMs) that have the advantage of being easily simulatable. In view of their similar properties, we expect these models to be as robust to the learning aspects that hinder DQN-type models. With DBMs and QBMs however, we face the issue that the evaluation of the merit function (i.e., the negative free energy of the BM after tracing out its hidden units), along with its gradient, is intractable for a large number of hidden layers/units. The reason behind this intractability is similar to the one behind the hardness of sampling from RBMs or DEBN-policies: in both cases, one needs to compute the free energy of the model after tracing out mutually dependent (i.e., connected) units, which is a hard task. However, similarly to the task of sampling from RBMs/DEBN-policies, approximate sampling algorithms can be used to estimate the free energy of DBMs and its gradient.\\

In this section, we introduce a hybrid quantum-classical scheme for training deep energy-based models (Algorithm \ref{alg:hybrid-EBRL}). We list several classical and quantum models that can be used in this hybrid scheme, and present quantum subroutines to speed-up sampling and the evaluation of the approximate merit function and its gradient. 

\subsection{Description of the algorithm}

\begin{figure*}
\begin{minipage}{\linewidth}
\begin{algorithm}[H]
\begin{algorithmic}
\State Initialize replay memory $\mathcal{D}$ to capacity $N$
\State Initialize main and target merit function approximators $M^{\bm{\theta}}$, $\widetilde{M}^{\bm{\theta'}}$ with the same random parameters
\State Initialize update rule $R$ in \{PS, Q-learning, SARSA\} 
\For{episode $=1, \ldots, E$} 
\State Initialize state $s_1$
\For {$t=1, \ldots, T$}
	\State Sample an action $\bm{a_t}$ with approximate probability $\pi(\bm{a_t}|\bm{s_t}) = e^{\beta M^{\bm{\theta}}  (\bm{s_t}, \bm{a_t})}/Z_\beta(\bm{s_t})$ using \textbf{Theorem \ref{thm:gibbs-state-preparation}}
	\State Execute action $\bm{a_t}$ on the environment and observe reward $r_t$ and state $\bm{s_{t+1}}$
	\If {R $\in$ \{SARSA, Q-learning\}}
	\State Store transition $\left(\bm{s_t},\bm{a_t},r_t,\bm{s_{t+1}}\right)$ in $\mathcal{D}$
	\EndIf
	\State Sample random minibatch $\mathcal{B}$ of transitions $\left(\bm{s_j},\bm{a_j},r_j,\bm{s_{j+1}}\right)$ from $\mathcal{D}$
	\For {transition $\in \mathcal{B}$}
	\State Evaluate $M^{\bm{\theta}}(\bm{s_{j}}, \bm{a_{j}})$ (and  $\widetilde{M}^{\bm{\theta'}}(\bm{s_{j}}, \bm{a_{j}})$ for PS) using \textbf{Theorem \ref{thm:merit-function-evaluation}}
	\If {R $\in$ \{SARSA, Q-learning\}}
	\State Set $\beta' \gg 1$ for Q-learning, otherwise $\beta'=\beta$
	\State Sample an action $a'$ with approximate probability $\pi(\bm{a'}|\bm{s_{j+1}}) = e^{\beta' \widetilde{M}^{\bm{\theta'}}(\bm{s_{j+1}}, \bm{a'})}/Z_{\beta'} (\bm{s_{j+1}})$ using \textbf{Theorem \ref{thm:gibbs-state-preparation}}
	\State Evaluate $\widetilde{M}^{\bm{\theta'}}(\bm{s_{j+1}}, \bm{a'})$ using \textbf{Theorem \ref{thm:merit-function-evaluation}}
	\EndIf
	\State Set
	$y_j =
    \left\{
    \begin{array}{l l}
      r_j + \widetilde{M}^{\bm{\theta'}}(\bm{s_{j}}, \bm{a_{j}}) - \gamma_\text{ps} M^{\bm{\theta}}(\bm{s_{j}}, \bm{a_{j}}) \quad & \text{if } R= \text{PS}\\
      r_j + \gamma \widetilde{M}^{\bm{\theta'}}(\bm{s_{j+1}}, \bm{a'}) \quad & \text{if } R= \text{SARSA}\\
      r_j + \gamma \max_{a'} \widetilde{M}^{\bm{\theta'}}(\bm{s_{j+1}}, \bm{a'}) \quad & \text{if } R= \text{Q-learning}
    \end{array} \right.$
	\State Evaluate $\nabla_{\bm{\theta}} M^{\bm{\theta}}(\bm{s_j},\bm{a_j})$ \textbf{Theorem \ref{thm:gradient-merit-function}}
	\State Accumulate $\Delta\bm{\theta}\ +\!=  \alpha \left(y_j - M^{\bm{\theta}}(\bm{s_{j}}, \bm{a_{j}}) \right) \nabla_{\bm{\theta}} M^{\bm{\theta}}(\bm{s_j},\bm{a_j})$
	\EndFor
	\State Perform a gradient descent step $\bm{\theta} + \Delta\bm{\theta}$ and reset $\Delta\bm{\theta}=0$
	\State Every C steps, set $\widetilde{M}^{\bm{\theta'}} = M^{\bm{\theta}}$ 
\EndFor
\If {R $=$ PS}
	\State Compute discounted rewards $\widetilde{r_t}$ using the edge-glow mechanism \cite{mautner15}
	\State Store experienced transitions $\left(\bm{s_t},\bm{a_t},\widetilde{r_t}\right)$ in $\mathcal{D}$
	\EndIf
\EndFor
\end{algorithmic}
\caption{Hybrid deep energy-based RL}
\label{alg:hybrid-EBRL}
\end{algorithm}
\vspace{-1.75\baselineskip}
\caption{\textbf{Pseudo-code of the algorithm used to train energy-based agents in RL environments. }Possible quantum enhancements are highlighted by Theorems \ref{thm:gibbs-state-preparation}, \ref{thm:merit-function-evaluation} and \ref{thm:gradient-merit-function}. }\vspace{-1.25\baselineskip}
\end{minipage}
\end{figure*}

In Algorithm \ref{alg:hybrid-EBRL}, we give a general procedure to train RL agents with energy-based merit function approximators. This procedure applies to both VBM and PS update rules and takes into account common mechanisms to stabilize learning, notably a replay memory and a target network (see Sec.\ \ref{sec:going-deep}). These mechanisms, along with the interaction with the environment, are kept classical. Throughout training though, the agent needs to perform three subroutines that are amenable to quantum speed-ups:
\begin{enumerate}[leftmargin=4mm]
	\item sampling from a distribution specified by the current approximation of the merit function
	\item estimating the approximate merit values associated to experienced states and sampled actions
	\item estimating the gradient of the approximate merit function for experienced state-action pairs 
\end{enumerate} 
So far, we have been considering DEBNs, for which subroutines 2 and 3 can be efficiently performed classically and hence do not require quantum speed-ups. For different energy-based models however, this can no longer be the case as they generally involve computing expectations values with respect to intractable probability distributions. Nonetheless, we show in Sec.\ \ref{sec:function+gradient-evaluation} that these two subroutines can be performed using the same sampling algorithms of subroutine 1, presented in Sec.\ \ref{sec:quantum-gibbs-preparation}. To see this, note that the free energy of an energy-based model takes the form $$F(\bm{v}) = \langle E(\bm{v,h}) \rangle_{P(\bm{h} | \bm{v})} + S(\bm{h} | \bm{v}),$$ that is the sum of the expected energy and Shannon entropy (or von Neumann entropy in the quantum case) of the model under the conditional Gibbs distribution $P(\bm{h} | \bm{v})$. Both these terms can be estimated by sampling from this same distribution. As for the gradient of the free energy, we show that for the models we consider, these gradients take the form of expected values of the model parameters under the distribution $P(\bm{h} | \bm{v})$.\\

Essentially, at the core of these three subroutines is a sampling problem. As explained in more detail in Sec.\ \ref{sec:quantum-gibbs-preparation}, classical algorithms that solve this problem should generate samples from a distribution that is close to a target Gibbs distribution of a classical or quantum Hamiltonian, while, in the quantum case, the analogue goal is to prepare quantum Gibbs states that encode such distributions. We focus on the complexities of such algorithms to quantify the runtime of our subroutines. 

\subsection{Examples of suitable energy-based models}

\begin{table*}
	\centering
	\includegraphics[width=1.0\textwidth]{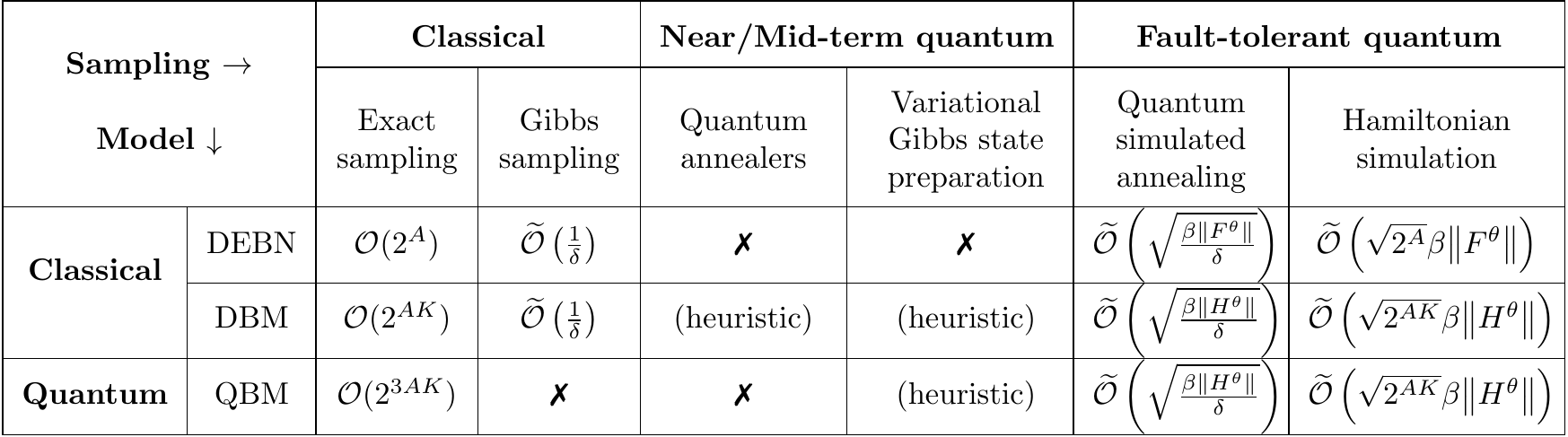}
	\caption{\textbf{Applicability of different sampling algorithms to several energy-based models.} $\norm{F^{\theta}}$ corresponds to the maximum free energy of the model, while $\norm{H^{\theta}}$ is its maximum clamped energy, i.e., the energy after tracing out the hidden units or simply fixing the input units respectively. $K$ counts the number of hidden units of the model while $A$ counts its number of action units. $\beta$ is the inverse temperature of the sampled Gibbs distribution and $\delta$ the spectral gap of the Gibbs sampling Markov chain using for sampling (when applicable). The complexities of exact sampling come from basic linear algebra techniques, while the complexities of classical Gibbs sampling are due to lower bound results on mixing times of Markov chains \cite{aldous82} (see also Appendix \ref{sec:walks}). The complexities for the quantum algorithms are proven in Theorem \ref{thm:gibbs-state-preparation}.}\vspace{-1.25\baselineskip}
	\label{table}
\end{table*}

Before discussing our quantum subroutines, we start by describing the energy-based models studied in this section. Having already introduced in Sec.\ \ref{sec:DEB-RL} restricted Boltzmann machines (RBMs), deep energy models (DEMs) and their associated free-energy DEBNs, we now focus on DBMs and QBMs. As is common in statistical and quantum physics, we switch from energy functions to a description in terms of Hamiltonians, involving here Pauli $\sigma^{x}$ and $\sigma^{z}$ matrices. Similarly to Sec.\ \ref{sec:DEB-RL}, we define our merit function approximators $M^{\bm{\theta}}(\bm{s,a})$ using the negative free energy $-F^{\bm{\theta}}(\bm{s,a})$, i.e., the energy after tracing out hidden units of these models.\\

\textbf{Deep Boltzmann machines.} DBMs \cite{salakhutdinov09,srivastava12} are generalizations of RBMs with additional hidden layers. They are described by a diagonal Hamiltonian of the form:
		\begin{align}\label{eq:DBM-Hamiltonian}
			H^{\bm{\theta}}_\text{DBM} &= - \sum_{i,k_{1}} w_{i,k_{1}} \sigma^z_i \sigma^z_{k_{1}} - \sum_{j,k_{1}} w_{j,k_{1}} \sigma^z_j \sigma^z_{k_1}\\
			&\quad - \sum_{\eta,k_{\eta},k_{\eta+1}} w_{k_{\eta},k_{\eta+1}} \sigma^z_{k_{\eta}} \sigma^z_{k_{\eta+1}} - \sum_{l} b_{l} \sigma^z_l \nonumber
		\end{align}
		parametrized by a real-valued vector $\bm{\theta} = (w_{i,k_{1}}, w_{j,k_{1}}, w_{k_{\eta},k_{\eta+1}}, b_l)_{i,j,k_\eta,l}$ and where we adopt the notation that assigns the indices $i,j,l$ to state, action and all units respectively and the indices $k_{\eta}$ to hidden units of the $\eta$-th hidden layer. Note that general Boltzmann machines with unrestricted connections (UBM) can also be mapped to Hamiltoninans of the form $H^{\bm{\theta}}_\text{DBM}$. One way of doing this is by simulating the intra-layer connections of a UBM using new connections to additional hidden units (that can be placed in the next hidden layer) \cite{gao17}.

		The DBM free energy however does not have a neural network formulation, and is instead expressed as:
		\begin{align}\label{eq:DBM-free-energy}
			F^{\bm{\theta}}_\text{DBM}(&\bm{s},\bm{a}) =  -\mathrm{log}\left(\sum_{\bm{h}}{e^{-H^{\bm{\theta}}(\bm{s,a},\bm{h})}}\right)\\
			&= - \sum_{i,k_{1}} w_{i,k_{1}} s_i \langle h_{k_{1}}\rangle - \sum_{j,k_{1}} w_{j,k_{1}} a_j \langle h_{k_1}\rangle \nonumber\\
			& - \sum_{\eta,k_{\eta},k_{\eta+1}} w_{k_{\eta},k_{\eta+1}} \langle h_{k_{\eta}} h_{k_{\eta+1}}\rangle - \sum_{k} b_{k} \langle h_k\rangle \nonumber\\
			& - \sum_{i} b_{i}s_i - \sum_{j} b_{j}a_j + \sum_{\bm{h}} P(\bm{h}|\bm{s,a})\mathrm{log}(P(\bm{h}|\bm{s,a})) \nonumber
		\end{align}
		where $\langle h_k \rangle$ is the expected value of the hidden unit $k$ under the conditional distribution $P(\bm{h}|\bm{s,a})$.\\

\textbf{Semi-transverse quantum Boltzmann machines.} QBMs \cite{amin18,kappen20} are a generalization of Boltzmann machines with additional transverse-field terms $\sigma^{x}$ in their Hamiltonians that allow for eigenvectors that are superpositions of computational basis states. The resulting non-diagonal Hamiltonians are of the form:
		\begin{align}\label{eq:QBM-Hamiltonian}
			H^{\bm{\theta}}_\text{QBM} &= - \sum_{i,k_{1}} w_{i,k_{1}} \sigma^z_i \sigma^z_{k_{1}} - \sum_{j,k_{1}} w_{j,k_{1}} \sigma^z_j \sigma^z_{k_1} \\
			& - \sum_{\eta,k_{\eta},k_{\eta+1}} w_{k_{\eta},k_{\eta+1}} \sigma^z_{k_{\eta}} \sigma^z_{k_{\eta+1}} - \sum_{l} b_{l} \sigma^z_l - \sum_{k} \Gamma_{k} \sigma^x_k \nonumber
		\end{align}
		parametrized by a real-valued vector $\bm{\theta} = (w_{i,k_{1}}, w_{j,k_{1}}, w_{k_{\eta},k_{\eta+1}}, b_l, \Gamma_k)_{i,j,k_\eta,l,k}$ and where we adopt the notation that assigns the indices $k$ to all hidden units. We choose to only apply the transverse-field terms $\Gamma_k \sigma^x_k$ on hidden units (hence the qualifier \emph{semi-transverse}). This choice allows to express the Hamiltonian in the form $H^{\bm{\theta}}_\text{QBM} = \sum_{\bm{v},\bm{h_{v}}} \lambda_{\bm{v},\bm{h_{v}}} \ket{\bm{v,h_v}}\bra{\bm{v,h_v}}$, where $\ket{\bm{v}} = \ket{\bm{s,a}}$  is a given configuration of the state and action units (i.e., a computational basis state) but $\ket{\bm{h_v}}$ is allowed to be an arbitratry superposition of computational basis states (see Appendix \ref{sec:bm_energy-based}). In turn, this leads to a straightforward derivation of $\Lambda_{\bm{v}}H^{\bm{\theta}}_\text{QBM}\Lambda_{\bm{v}} = \sum_{\bm{h_v}} \lambda_{\bm{v},\bm{h_{v}}} \ket{\bm{v,h_v}}\bra{\bm{v,h_v}}$ for  $\Lambda_{\bm{v}} = \Lambda_{\bm{s,a}} = \ket{\bm{s,a}}\bra{\bm{s,a}} \otimes I_{\bm{h}}$ a projector on a given configuration $\bm{v} = (\bm{s,a})$, which helps evaluate the QBM's free energy and its gradient via sampling (see Sec.\ \ref{sec:function+gradient-evaluation}). These considerations are similar to the ones of Wiebe et al.\ \cite{wiebe19} (albeit visible units are the ones allowed to be in superposition in their case), who used QBMs for generative training of quantum states.

		This free energy again does not take a neural-network formulation and has the following expression:
		\begin{equation}\label{eq:QBM-free-energy}
			F^{\bm{\theta}}_\text{QBM}(\bm{s},\bm{a}) = \text{Tr}[\rho_{\bm{s,a}} H^{\bm{\theta}}_\text{QBM}] + \text{Tr}[\rho_{\bm{s,a}} \mathrm{log}(\rho_{\bm{s,a}})]
		\end{equation}
		where $\rho_{\bm{s,a}} = \frac{\Lambda_{\bm{s,a}}e^{-H^{\bm{\theta}}_\text{QBM}}\Lambda_{\bm{s,a}}}{\text{Tr}\left[\Lambda_{\bm{s,a}}e^{-H^{\bm{\theta}}_\text{QBM}}\right]}$ is the QBM Gibbs state after a projective measurement $\Lambda_{\bm{s,a}}$.\\

\textbf{Access models. }In the case of RBMs and DEMs, we take advantage of the computability of their free energies to avoid using their full Hamiltonian, and consider instead the Hamiltonians defined by their associated DEBNs.

All the Hamiltonians considered here allow sparse access \cite{berry09}, meaning that we can efficiently (i.e., in polynomial time) construct the two unitary oracles:
\begin{equation*}
	\hat{O}_H \ket{j}\ket{k}\ket{z} = \ket{j}\ket{k}\ket{z \oplus H_{jk}}, \quad \hat{O}_F \ket{j}\ket{l} = \ket{j}\ket{f(j,l)}
\end{equation*}
For a $d$-sparse Hamiltonian $H$, $\hat{O}_H$ accepts a row index $j$ and column index $k$ and returns $H_{jk}$ in a binary format, while $\hat{O}_F$ accepts a row index $j$ and a number $l \in [d]$ and computes in-place the column index $f(j,l)$ of the $l$-th non-zero element in row $j$.

While sparse access is sufficient to run all the quantum subroutines described below, its implementation requires resources far beyond those NISQ devices can offer. This is the case of DEBN (and DEM) Hamiltonians for instance, where the implementation of $\hat{O}_H$ requires the computation of the entire neural network function coherently, which is out of reach of near-term implementations. DBMs and QBMs however also naturally satisfy access in the form of linear combinations of (local) unitaries $H = \sum_{n=1}^N \alpha_n U_n$, where $N$ scales at most quadratically with the number of units of the model, which eases up significantly their implementation.

\subsection{Quantum algorithms for Gibbs state preparation\label{sec:quantum-gibbs-preparation}}

At the crux of training energy-based models is the ability to sample from Gibbs distributions of the form
\begin{equation}\label{eq:gibbs-distribution}
p_\beta(\bm{x}) = \frac{e^{- \beta H(\bm{x})}}{\text{Tr}_{\bm{x}}[e^{-\beta H}]}
\end{equation}
of a Hamiltonian $H$ specifying our learning model, where we assume $\left\{\bm{x}\right\}$ to be a subset of the eigenvectors of $H$, for simplicity\footnote{For QBMs, we are as well interested in Gibbs distributions for configurations $\left\{\bm{x}\right\}$ that are not part of the eigenbasis of $H$, leading to a different expression of (\ref{eq:gibbs-distribution}) that involves the eigendecomposition of $H$. For classical Hamiltonians (e.g., DBMs), these eigenvectors are simply computational basis states.}. While the energy $H(\bm{x})$ of a single configuration can be efficiently computed, the partition function $\text{Tr}_{\bm{x}}[e^{-\beta H}] = \sum_{\bm{x}} e^{-\beta H(\bm{x})}$ involves the energies of combinatorially many configurations, making it hard to approximate in general \cite{long10}.

Classically, this distribution can be computed explicitly by evaluating the energy of all the configurations $\bm{x}$ we want to sample over and normalizing the resulting vector of energies. This is the approach we took in the numerical simulations of Sec.\ \ref{sec:numerical-simulations}.
Alternatively, one can use random walk algorithms such as Gibbs sampling to generate samples from an approximation of this distribution (for completeness, we describe these random walk algorithms in Appendix \ref{sec:walks}).

In a quantum setting, one is interested in the similar goal of preparing approximately a state of the form\footnote{Actually, the general goal of purified Gibbs state preparation is to prepare a state $\ket{\phi}_{AB}$ such that $\text{Tr}_{B}[\ket{\phi}\bra{\phi}_{AB}] \approx \ket{\psi}\bra{\psi} = e^{- \beta H}/\text{Tr}[e^{-\beta H}]$. The state of Eq.\ (\ref{eq:gibbs-state}) is a particular instance, called a \emph{coherent encoding}, which can be more useful than a general purified Gibbs state, e.g., when estimating the partition function of the associated Gibbs distribution (\ref{eq:gibbs-distribution}) \cite{wocjan08b}.}
\begin{equation}\label{eq:gibbs-state}
\ket{\psi_\beta} = \frac{1}{\sqrt{\text{Tr}_{\bm{x}}[e^{-\beta H}]}}\sum_{\bm{x}} \sqrt{e^{- \beta H(\bm{x})}} \ket{\bm{x}},
\end{equation}
called a purified Gibbs state of our model Hamiltonian.\break This is because measuring this state in the basis $\{\bm{x}\}_{\bm{x}}$ effectively returns a sample from the Gibbs distribution (\ref{eq:gibbs-distribution}). Interestingly, the exact and Gibbs sampling approaches mentioned above both have their quantum analogue, each providing a quadratic speed-up in sampling over the classical algorithms (up to overheads in some parameters). However, these quantum algorithms utilize many gates and qubits and will likely require fault-tolerant implementations to be used. To consider approaches that may be more amenable to NISQ constraints, we also study a variational approach to Gibbs state preparation that trades off guarantees of speed-up and convergence with less resource-consuming primitives.\\

We summarize the complexities of these sampling algorithms in the following informal theorem, before providing their short descriptions.

\begin{theorem}[Purified Gibbs state preparation (informal)]\label{thm:gibbs-state-preparation}
Given sparse access to one of the model Hamiltoninans listed above and a bound $\norm{M^{\theta}}$ on its free energy or clamped energy (i.e., the energy after tracing out the hidden units or simply fixing the input units respectively), we can prepare a purified Gibbs state at inverse temperature $\beta$ in either $$T_\text{Gibbs} = \widetilde{\mathcal{O}} \left( \sqrt{\frac{\beta \norm{M^{\bm{\theta}}}}{\delta}} \right)$$ quantum walk steps, where $\delta$ is a lower bound on the spectral gaps of the Markov chains associated to this quantum walk; or in $$T_\text{Gibbs}' =\widetilde{\mathcal{O}} \left(\sqrt{2^{n}} \beta\norm{M^{\bm{\theta}}}\right)$$ gates and queries to the Hamiltonian, where $n$ is the number of action (and hidden) units of the model.\\ Alternatively, for DBMs and QBMs, one can opt for a heuristic variational approach, where the cost of evaluating the objective function is in $$\widetilde{\mathcal{O}}\left(\norm{M^{\theta}}/\varepsilon + 1/(p_{\min}^2\beta\varepsilon)\right)$$ gates and preparations of the parametrized quantum circuit, where $\varepsilon$ is the precision of the estimate and $p_\text{min}$ an approximation parameter.
\end{theorem}

A recurring quantity in the subroutines of Theorem \ref{thm:gibbs-state-preparation} is a bound on the model Hamiltonian's norm $\norm{M^{\theta}}$, that one needs to estimate to be able to apply them. In order to do so, one has several possibilities:
\begin{itemize}[leftmargin=2mm]
	\item Use $\norm{M^{\theta}} = \norm{\bm{\theta}}_1$, the $l_1$-norm of all the weights that parametrize a chosen model, as this quantity also bounds the Hamiltonian's maximum eigenvalue
	\item Use the maximum finding algorithm to get an approximate value \cite{durr96,van20} of $\norm{M^{\theta}}$ (even though this algorithm has a complexity in $\widetilde{\mathcal{O}}\left(\sqrt{2^n}\right)$ for $n = A \text{ or } AK$)
	\item Assume that the approximated merit function is close to the true function and use a bound on the latter, e.g., $\frac{\abs{R}_\text{max}}{1-\gamma}$ in the case of VBMs, where $\abs{R_\text{max}}$ is the largest obtainable reward and $\gamma$ the environment's discount factor.
\end{itemize}

The proof of Theorem \ref{thm:gibbs-state-preparation} is a combination of known results along with an additional Lemma derived from previous works. In the following, we summarize these results, namely fault-tolerant quantum algorithms based on Hamiltonian simulation and quantum simulated annealing to prepare purified Gibbs states, followed by variational approaches to Gibbs state preparation suitable for near-term quantum devices.\\

\textbf{Gibbs state preparation using simulation. }The approach of van Apeldoorn et al.\ \cite{van20} consists in using Hamiltonian simulation techniques \cite{low19} to implement an approximation of the operation $f(H) = e^{-H}$. Since this operation is not unitary, one needs to rely on so-called block-encodings to apply it, that is, construct a unitary $\tilde{U}$ such that $\norm{(\bra{0}\otimes I)\tilde{U}(\ket{0}\otimes I) - e^{-H}} \leq \varepsilon$. This construction has a complexity in $\widetilde{O}\left(\norm{H}d \right)$ for a $d$-sparse Hamiltonian. One can then apply $\tilde{U}$ on the A-register of the maximally entangled state $\sqrt{2^{-n}}\sum_{i} \ket{i}_A\ket{i}_B$ and use $\widetilde{\mathcal{O}}\left(\sqrt{2^n/\text{Tr}[e^{-H}]}\right)$ rounds of amplitude amplification \cite{yoder14}\break to amplify the relevant $(\ket{0}\otimes I)$-subspace. This results in a good approximation of the desired purified Gibbs state. To get rid of the $1/\text{Tr}[e^{-H}]$ term in the complexity, the authors additionally apply a linear transformation on $H$, that we defer, along with a detailed restatement of their result, to Appendix \ref{sec:deferred-thms}. Note that, since this approach is based on Hamiltonian simulation techniques, it is applicable to both classical and quantum (i.e., non-diagonal in the computational basis) Hamiltonians without the need for special considerations.\\

\textbf{Quantum simulated annealing (QSA). }In the case of quantum simulated annealing, classical and quantum Hamiltonian require a separate treatment, that we describe below. We start by an overview of the classical-Hamiltonian case and then explain how its methods can be generalized to the quantum-Hamiltonian case.

\textit{Classical Hamiltonians. }The idea behind the QSA approach is to quantize Gibbs sampling random walks to achieve a quadratic speed-up in runtime. Random walk algorithms have a runtime complexity that depends on the spectral gap $\delta$ of their associated Markov chains (MCs). This runtime, also called mixing time of the MC, is in $\tilde{\mathcal{O}}(\delta^{-1})$ (see Appendix \ref{sec:walks} for more details). An example of such MCs is specified by  the Gibbs sampling (or more generally, Metropolis-Hastings) algorithm. Pictorially, these MCs define the transition probabilities of a random walk over the eigenstates of the Hamiltonian that is governed by their respective eigenvalues, such that, after the MC mixing time,  eigenstates are sampled according to their Gibbs distribution. A quantum analogue of classical random walks are Szegedy-type quantum walks \cite{szegedy04,magniez11}. They generalize a MC $P$ to a unitary $W(P)$ called a Szegedy walk operator: by taking $P$ to be a Gibbs sampling MC whose stationary distribution is the Gibbs distribution $p_\beta$ of Eq.\ (\ref{eq:gibbs-distribution}), $W(P)$ has the quantum state $\ket{\psi_\beta}$ of Eq.\ (\ref{eq:gibbs-state}) as its eigenvalue-$1$ eigenvector. Using a combination of phase estimation \cite{brassard02} and amplitude amplification \cite{yoder14}, one can transform an initial state $\ket{\phi}$ into the desired state $\ket{\psi_\beta}$ using a total of $\tilde{\mathcal{O}}\big(\big(\abs{\braket{\psi_\beta}{\phi}}\sqrt{\delta}\big)^{-1}\big)$ applications of $W(P)$ \cite{wocjan08}.\break
Direct preparation of $\ket{\psi_\beta}$ from an arbitrary state $\ket{\phi}$, e.g., the uniform superposition $\ket{u}=\sqrt{2^{-n}}\sum_{i} \ket{i}$, can however be intractable since the overlap $\abs{\braket{u}{\pi_\beta}}^2$ can be as small as $2^{-n}$ (or even worse for a different $\ket{\phi}$) \cite{orsucci18}. To get around this issue, one can instead use these same tools to prepare several intermediary states $\ket{u} = \ket{\psi_{\beta_0}},\ \ket{\psi_{\beta_1}},\ \ldots,\ \ket{\psi_{\beta_l}}$ that follow an \emph{annealing schedule} on the inverse temperature $\beta$ of the Gibbs distribution. Intuitively, if the overlap between these intermediate states is made large by construction, then the whole transformation should be efficient. Formalizing this idea, Wocjan et al.\ have shown \cite{wocjan08} that given a schedule $\beta_0=0,\ \beta_1,\ \ldots,\ \beta_l=\beta$ and Gibbs sampling Markov chains $P_{\beta_0},\ P_{\beta_1},\ \ldots,\ P_{\beta_l}$ with stationary distributions $p_{\beta_0},\ p_{\beta_1},\ \ldots,\ p_{\beta_l}$ that fulfill the slow-varying condition $|\braket{\psi_{\beta_i}}{\psi_{\beta_{i+1}}}|^2 \geq c$, preparing an $\varepsilon$-approximation of $\ket{\psi_\beta}$ has complexity $\mathcal{O}(\frac{l}{c}\sqrt{\delta^{-1}}\log^2(\frac{l}{\varepsilon}))$. Using an \emph{adaptive} annealing schedule, Harrow \& Wei \cite{harrow20} were able to provide a construction with length $l=\tilde{\mathcal{O}}(\sqrt{\log(\text{Tr}[e^{-\beta H}]^{-1})})\leq\tilde{\mathcal{O}}(\sqrt{\beta \left\lVert M^{\bm{\theta}} \right\rVert})$, leading to a total preparation complexity in $\tilde{\mathcal{O}}(\sqrt{\delta^{-1}\beta \left\lVert M^{\bm{\theta}} \right\rVert})$.\\

\textit{Quantum Hamiltonians. }In the case of a quantum Hamiltonian (e.g., for QBMs), while we can still construct sparse access to this Hamiltonian, its eigenstates are not computational basis states in general but instead superpositions that are a priori unknown. Since random walk algorithms need access to the eigenstates of a problem Hamiltonian in order to navigate its energy landscape (and eventually sample from its associated Gibbs distribution), this prevents us from constructing the necessary Gibbs sampling Markov chains and their associated quantum walk operators\footnote{This is also the reason why we deem classical Gibbs sampling to be inapplicable to QBMs in Table \ref{table}, and why exact sampling has larger complexity for QBMs than DBMs: one should additionally diagonalize the QBM Hamiltonian to be able to exponentiate it.}. A way to circumvent this issue was introduced by Temme et al.\ \cite{temme11}, then further generalized to work with quantum walk operators and QSA by Yung \& Aspuru-Guzik \cite{yung12}. While these approaches use different constructions, both rely at their core on Hamiltonian simulation and quantum phase estimation to get access to the eigenvalues of a quantum Hamiltonian without knowledge of the structure of its eigenvectors. Because we're interested here in quantum walk approaches, we focus here on the method of Yung \& Aspuru-Guzik. This method, as presented, doesn't use \emph{adaptive} annealing schedules and relies on projective measurements instead of quantum amplitude amplification to adiabatically transform intermediary states, analogously to the approach of Somma et al.\ \cite{somma08}. For these reasons, the method has a complexity in $\widetilde{\mathcal{O}} \left( \beta^2 \norm{M^{\bm{\theta}}}^2 / \sqrt{\delta} \right)$, quartically worse in $\beta$ and $\norm{M^{\bm{\theta}}}$ than with diagonal Hamiltonians. 
We prove however, in the following Lemma, that the advances in QSA for classical Hamiltonians presented above are also applicable to quantum Hamiltonian, leading to the same complexities.

\begin{lemma}
	Given sparse access to a quantum Hamiltonian $H$ and a bound $\norm{M^{\theta}}$ on its free energy or clamped energy, we can prepare via adaptive quantum simulated annealing a quantum state $\varepsilon$-close to the coherent encoding of its Gibbs distribution at inverse temperature $\beta$
\begin{equation}\label{eq:coherent-encoding-quantum-Gibbs}
\frac{1}{\sqrt{\textnormal{Tr}[e^{-\beta H}]}}\sum_{\bm{\varphi}} \sqrt{e^{- \beta \bra{\bm{\varphi}}H\ket{\bm{\varphi}}}} \ket{\bm{\varphi}} \ket{\tilde{\bm{\varphi}}}
\end{equation}
where $\ket{\bm{\varphi}}$ are eigenstates of $H$ and $\ket{\tilde{\bm{\varphi}}}$ their complex conjugates; using $$\widetilde{\mathcal{O}} \left( \sqrt{\frac{\beta\norm{M^{\bm{\theta}}}}{\delta}} \right)$$ quantum walk steps, where $\delta$ is a lower bound on the spectral gaps of the Markov chains associated to this quantum walk.
\end{lemma}
\begin{proof}
The proof of this lemma goes in two steps, each extending the construction of Yung \& Aspuru-Guzik \cite{yung12} to the algorithmic improvements of Wocjan et al.\ \cite{wocjan08} and Harrow \& Wei \cite{harrow20} respectively:\\

\textit{Step 1. }We first note that the generalized Szegedy walk operators defined by Yung \& Aspuru-Guzik for quantum Hamiltonians satisfy the relevant properties that are required for the QSA approach of Wocjan et al.\ (Theorem 2 in Ref.\ \cite{wocjan08}) that is based on amplitude amplification instead of projective measurements. These properties are namely having a purified Gibbs state of the form (\ref{eq:coherent-encoding-quantum-Gibbs}) as eigenvalue-$1$ eigenstate and a phase gap $\Delta \geq 2\sqrt{\delta}$. Note that, as shown by Yung \& Aspuru-Guzik, the eigenbasis of the quantum Hamiltonian is not needed to construct the generalized walk operators and input states of these QSA algorithms. Additionally, by working in the eigenbasis of the quantum Hamiltonian instead of the computational basis, the error analysis in Lemma 3 and Corollary 3 of Ref.\ \cite{wocjan08} to derive the length and overlaps of the annealing sequence is straightforwardly applicable to quantum Hamiltonians, leading to a preparation complexity in $\widetilde{\mathcal{O}} \left( \beta \norm{M^{\bm{\theta}}} / \sqrt{\delta} \right)$.\\

\textit{Step 2. }By tracing the construction steps of Harrow \& Wei (Theorem 4 of Ref.\ \cite{harrow20}), we can see that the proof of existence of an annealing (or cooling) schedule which ensures the slow-varying condition (i.e., large overlaps between intermediary states) in their adiabatic preparation is indenpendent of the Hamiltonian eigenbasis. Hence the proof can be directly expressed in the eigenbasis of quantum Hamiltonians without any other change. Moreover, the tools they use for the adaptive construction of these schedules, namely nondestructive amplitude estimation for binary search of the next inverse temperature $\beta'$ (Theorems 9 and 10 of Ref.\ \cite{harrow20}), are also applicable to generalized Szegedy walk operators. Hence, by applying the same methods, we obtain an overall preparation complexity in $\widetilde{\mathcal{O}} \left( \sqrt{\beta \norm{M^{\bm{\theta}}} / \delta} \right)$.
\end{proof}

\textbf{Variational Gibbs state preparation. }Recent advances in the field of quantum computation have motivated the development of algorithmic approaches that are compatible with near/mid-term devices. One particularly promising approach for the preparation of Gibbs states is based on a variational optimization method. The general idea is to train a parametrized Ansatz $\ket{\psi_{\bm{\theta}}} = U(\bm{\theta}) \ket{\bm{0}}$ such that it approximates the state of Eq.\ (\ref{eq:gibbs-state}), or equivalently, that its density matrix $\rho = \ket{\psi_{\bm{\theta}}} \bra{\psi_{\bm{\theta}}}$ approximates the Gibbs state $e^{-\beta H}/\text{Tr}[e^{-\beta H}]$. Since this target Gibbs state  is the unique state that minimizes the free energy
\begin{equation}\label{eq:variational-free-energy}
	F = \text{Tr}[\rho H] + \frac{1}{\beta}\text{Tr}[\rho \mathrm{log}(\rho)]
\end{equation}
with respect to the Hamiltonian $H$, one can use this free energy as an objective function to be minimized by $\ket{\psi_{\bm{\theta}}}$. The existing methods in variational Gibbs state preparation \cite{wu19,chowdhury20,wang20} vary in three aspects:
\begin{itemize}[leftmargin=4mm]
	\item The choice of Ansatz: Chowdhury et al.\ \cite{chowdhury20} use an Ansatz derived from Trotterized adiabatic state preparation, while Wu \& Hsieh \cite{wu19} opt for an Ansatz motivated by QAOA, and Wang et al.\ \cite{wang20} choose a so-called hardware-efficient Ansatz.
	\item The subroutines to evaluate the Ansatz' free energy (Eq.\ (\ref{eq:variational-free-energy})): the objective function is composed of the average energy of $\rho$ with respect to $H$ and the von Neumann entropy of $\rho$. The latter is the more term challenging to evaluate. While Wu \& Hsieh consider a decomposition into several Renyi entropies, measured using swap tests \cite{subacsi19}; Wang et al.\ and Chowdhury et al.\ use a truncation of its Taylor series, where the terms are computed via (high-order) swap tests and amplitude estimation \cite{brassard02} respectively.
	\item The optimization method: here one can opt for a gradient-descent approach, where several evaluations of the objective function and finite-difference formulas can be used to estimate its gradient; or else use gradient-free methods such as Nelder-Mead or Powell's method \cite{nocedal06}.
\end{itemize} 

The approaches for free energy evaluation are especially useful for us since they can directly be used for the merit function evaluation subroutine of Theorem \ref{thm:merit-function-evaluation}. We choose to use the approach of Chowdhury et al., that we restate in Appendix \ref{sec:deferred-thms} for completeness. Note however that, depending on the RL application, a different approach might be more suitable, as different approximation methods correspond to different effective models. Further investigations should reveal which methods are more interesting in a RL context.

\subsection{Quantum algorithms for merit function and gradient evaluation\label{sec:function+gradient-evaluation}}

As mentioned above, the neural network formulation of the RBM and DEM free energies (as DEBNs) allows for efficient exact computation of the modeled merit functions and their gradients. As for DBMs and QBMs, it is easy to see from the formulation of their free energies (Eqs.\ \ref{eq:DBM-free-energy} and \ref{eq:QBM-free-energy}) that these should be estimated via sampling from their Gibbs distributions. More specifically, these are the Gibbs distributions of the \emph{clamped} Hamiltonians, where state and action units have been fixed to a given configuration.

Regarding the gradients of these free energies however, the argument becomes slightly more intricate in the case of QBMs. As derived in Appendix \ref{sec:deferred-thms}, derivatives of the QBM free energy take the form $\partial_\theta F(\bm{s,a}) = - \text{Tr}[\Lambda_{\bm{s,a}} \partial_\theta e^{-H_\text{QBM}}]/\text{Tr}[\Lambda_{\bm{s,a}} e^{-H_\text{QBM}}]$. Due to the non-diagonal terms in $H_\text{QBM}$, we have that $H_\text{QBM}$ and $\partial_\theta H_\text{QBM}$ do not commute, making it impossible to use the identity $\partial_\theta e^{-H} = -\partial_\theta H e^{-H}$ that would apply for instance to DBM Hamiltonians. Nonetheless, due to our restriction to \emph{semi-transverse} QBMs and by the use of Duhamel's formula  \cite{amin18}, we can show that the gradient of the QBM free energy takes the form of expectation values $\text{Tr}[(\partial_\theta H_\text{QBM}) \rho_{\bm{s,a}}]$ where $\partial_\theta H_\text{QBM}$ are simply (products of) Pauli $\sigma^z$ and $\sigma^x$ operators appearing in $H_\text{QBM}$.

Hence, merit function and gradient evaluation boil down to Gibbs state preparation again, but with fixed state and action units $\bm{s,a}$. We formalize the subroutines used for these evaluations in the following theorems.

\begin{theorem}[DBM/QBM merit function evaluation]\label{thm:merit-function-evaluation}
Let $H^{\bm{\theta}}_\text{DBM/QBM}$ be a Hamiltonian describing a deep Boltzmann machine as specified by Eq.\ (\ref{eq:DBM-Hamiltonian}) or a semi-transverse quantum Boltzmann machine as specified by Eq.\ (\ref{eq:QBM-Hamiltonian}) and $(\bm{s},\bm{a})$ a given assignment of the state and action units. The approximate free energy $F^{\bm{\theta}}_\text{DBM/QBM}(\bm{s},\bm{a}) = -M^{\bm{\theta}}(\bm{s},\bm{a})$ associated to this state-action pair can be estimated to precision $\varepsilon$ in $$\widetilde{\mathcal{O}}\left(T_\text{Gibbs}\cdot\left(\frac{\norm{\bm{\theta}}_1}{\varepsilon} + \frac{1}{p_{\min}^2\varepsilon}\right)\right)$$ where $T_\text{Gibbs}$ is the complexity of preparing the Gibbs state of the clamped Hamiltonian, $\norm{\bm{\theta}}_1 = \sum_{i,k} \abs{w_{i,k}}s_i + \sum_{j,k} \abs{w_{j,k}}a_j + \sum_{k,k'} \abs{w_{k,k'}}\ (+\sum_k \abs{\Gamma_k})$ and $p_{\min}$ is the cut-off used in the estimation of the entropy term of the free energy.
\end{theorem}
\begin{proof}Using one of the subroutines in Theorem \ref{thm:gibbs-state-preparation}, we can prepare a purified Gibbs state of the clamped model Hamiltonian at inverse temperature $\beta = 1$ in time $T_\text{Gibbs}$. We then use the subroutines of Theorem \ref{thm:chowdhury-et-al_est-avg-energy} and Theorem \ref{thm:chowdhury-et-al_est-entropy} (Appendix \ref{sec:deferred-thms}) to evaluate the average energy and von Neumann entropy of this state using $\widetilde{\mathcal{O}}\left(\norm{\bm{\theta}}_1/\varepsilon\right)$ and $\widetilde{\mathcal{O}}\left(1/(p_{\min}^2\varepsilon)\right)$ calls to the Gibbs state preparation unitary respectively. Note that the success probability of these subroutines can be amplified to $1-\delta$ using $\log(\delta^{-1})$ repetitions and taking the median of the outcomes, according to the powering lemma \cite{jerrum86}.
\end{proof}

\begin{theorem}[Gradient evaluation of a DBM/QBM merit function]\label{thm:gradient-merit-function}
Let $H^{\bm{\theta}}_\text{DBM/QBM}$ be a Hamiltonian describing a deep Boltzmann machine as specified by Eq.\ (\ref{eq:DBM-Hamiltonian}) or a semi-transverse quantum Boltzmann machine as specified by Eq.\ (\ref{eq:QBM-Hamiltonian}) and $(\bm{s},\bm{a})$ a given assignment of the state and action units. The gradient of the free energy $\nabla_{\bm{\theta}} F^{\bm{\theta}}_\text{DBM/QBM}(\bm{s},\bm{a}) = -\nabla_{\bm{\theta}} M^{\bm{\theta}}(\bm{s},\bm{a})$ associated to this state-action pair can be estimated to precision $\varepsilon$ in $$\widetilde{\mathcal{O}}\left(T_\text{Gibbs}\cdot\min\left\{\frac{K + K_L^2(L-1)}{\varepsilon}, \frac{1}{\varepsilon^2}\right\}\right)$$ where $T_\text{Gibbs}$ is the complexity of preparing the Gibbs state of the clamped Hamiltonian, $K$ is the number of hidden units in the model and $K_L$ is the number of hidden units in each of the $L$ layers of the model.
\end{theorem}
\begin{proof}
We count the number of components of the gradient to of the order $\mathcal{O}\left(K + K_L^2(L-1)\right)$. Using one of the subroutines in Theorem \ref{thm:gibbs-state-preparation}, we can prepare a purified Gibbs state of the clamped model Hamiltonian at inverse temperature $\beta = 1$ in time $T_\text{Gibbs}$. From there, we can take two approaches:
\begin{itemize}[leftmargin=3mm, topsep=1mm]
	\item Measuring the hidden units of the Gibbs state in the computational basis gives us a vector $\bm{h}$ sampled from the distribution $P(\bm{h}|\bm{s,a})$. For the QBM, we also need to estimate (in a separate round of measurements) the expected values of Pauli $\sigma^x_k$ for all hidden units $k$, which can be done by applying a Hadamard gate on each hidden unit before measuring them in the computational basis. By repeating these measurements $\mathcal{O}(1/\varepsilon^2\log((K + K_L^2(L-1))/\delta)$ times and averaging the resulting vector samples, we can estimate each component of the gradient to precision $\varepsilon$ with probability $1-\delta$ using Chebyshev's inequality and the powering lemma.
	\item We can use amplitude estimation to compute the expected value of each hidden unit, Hadamard-transformed hidden unit and pair of hidden units. But every single term has to be estimated individually, which amounts to a total number of calls to the Gibbs state preparation unitary in $\mathcal{O}((K + K_L^2(L-1))/\varepsilon\log(\delta^{-1}))$ to get all estimates to precision $\varepsilon$ with probability $1-\delta$. Note that this approach only applies to Gibbs state preparation subroutines that create a coherent encoding of the Gibbs sate (see Eq.\ (\ref{eq:gibbs-state}) and associated footnote) since we need a coherent preparation of our quantum state to apply amplitude estimation on it \cite{wocjan08b}. It is hence compatible with the quantum simulated annealing and variational approaches but not the Hamiltonian simulation method.
\end{itemize}
\end{proof}

Note that the clamped model Hamiltonians now act on a subspace of size $2^K$ for a model with $K$ hidden units. This, as well as a bound $\norm{H^{\bm{\theta}}}$ on the model's clamped energy that is potentially smaller (since action units are fixed), leads to different complexities $T_\text{Gibbs}$ in Theorem \ref{thm:gibbs-state-preparation} and Table \ref{table}. Moreover, while the preparation of the clamped Gibbs state $\rho_{\bm{s,a}}$ is straightforward for the Hamiltonian simulation and quantum simulated annealing approaches, we point out that one needs to train again a new Ansatz $\ket{\psi_{\bm{\theta}}}$ in the variational approach. This is because, in general, one has no way of efficiently clamping state and action units in the Ansatz trained on the full Hamiltonian (post-selection on measurements or amplitude amplification are possible but resource-consuming). We consider this to be an interesting question that could be explored in future work.

\section{Conclusion and Outlook\label{sec:conclusion}}

Recent years have seen a massive surge of interest in leveraging results from quantum information theory to enhance machine learning methods and models. Nonetheless, deep reinforcement learning, as one of the emerging fields in machine learning, remains largely unaffected by this trend. This is in part due to an effort made by the machine learning community to avoid bottlenecks resulting from application-specific algorithms in favor of general applicability.\\
In this paper, we consider deep energy-based models specifically designed to cope with large state and action spaces in complex reinforcement learning (RL) environments. We demonstrate that these models can indeed outperform standard RL architectures such as deep-Q networks due to their capacity to represent complex correlations between state-action pairs.\\ 
The enhanced representational power of deep energy-based models comes at the cost of sampling from an associated probability distribution. Even approximate sampling is a hard problem and therefore, gives rise to a computational bottleneck. In this paper, we address this bottleneck by leveraging results from quantum approximate sampling algorithms in order to enable various models for deep energy-based RL.\\ 
Our main result is a hybrid deep energy-based RL algorithm that enables quadratic quantum speed-ups with full-scale quantum computers and is applicable to various deep energy-based models. In the interest of deploying near-term quantum devices, we further prove that the same algorithm can be combined with deep Boltzmann machines and quantum Boltzmann machines to provide heuristic quantum speed-ups with near-term quantum computers.\\

The presented framework for quantum energy-based RL does not limit itself to the methods and models presented here. Following current advances in quantum machine learning, we expect genuinely quantum function approximators such as parametrized quantum circuits \cite{schuld20,farhi18,chen19,schuld20b}, quantum generative models \cite{gao18} or quantum kernels \cite{havlivcek19,schuld19,lloyd20} to present interesting features (e.g., the ability to represent quantum correlations) that could provide advantages in complex RL environments. Moreover, our study of sampling algorithms has been considering the restrictive setting of isolated (or context-free) sampling distributions. The RL setting however offers an online structure that could be exploited by the sampling algorithms. For instance, inspired by existing quantum algorithms that speed-up sampling from sequences of slowly evolving Markov chains \cite{orsucci18}, one could adapt the Markov chain-based approach of quantum simulated annealing to exploit the property that the agent's policy is slowly evolving as well during training. In sum, our framework opens up a new avenue of connecting a large class of quantum models and methods to be applied in the context of quantum RL.

\section*{Acknowledgments}
SJ, HPN, LMT, and HJB acknowledge support from the Austrian Science Fund (FWF) through the projects DK-ALM:W1259-N27 and SFB BeyondC F71. SJ also acknowledges the Austrian Academy of Sciences as a recipient of the DOC Fellowship. HJB was also supported by the Ministerium für Wissenschaft, Forschung, und Kunst BadenWürttemberg (AZ:33-7533.-30-10/41/1). This work was in part supported by the Dutch Research Council (NWO/OCW), as part of the Quantum Software Consortium program (project number 024.003.037). The computational results presented here have been achieved in part using the LEO HPC infrastructure of the University of Innsbruck.

\bibliographystyle{ieeetr}
\bibliography{references}

\clearpage
\appendix

\section{Introduction to reinforcement learning\label{sec:intro-RL}}
A reinforcement learning scenario consists in a cyclic agent-environment interaction, where \emph{states} (or, equivalently in the present context, \emph{percepts}) $\bm{s}$, \emph{actions} $\bm{a}$ and \emph{rewards} $r$ are exchanged between the two parties. Every cycle of interaction is timed by a \emph{timestep} $t$, such that a given interaction is described by a \emph{history} $\bm{h} = (\bm{s^{(0)}}, \bm{a^{(0)}}, r^{(1)}, \bm{s^{(1)}}, \bm{a^{(1)}}, r^{(2)}, \ldots)$.\\

\textbf{Markov decision processes. }We choose a common description of the environment in terms of a Markov decision process (MDP), defined by a set of \emph{discrete} states $\mathcal{S}$, a set of \emph{discrete} actions $\mathcal{A}$, a transition probability distribution $\mathcal{P}: \mathcal{S}\times\mathcal{A}\times\mathcal{S} \rightarrow \mathbb{R}_+$, a (stochastic) reward function $\mathcal{R} : \mathcal{S}\times\mathcal{A} \rightarrow \mathbb{R}$ and a discount factor of future rewards $\gamma \in [0,1]$. Given that the state (action) space is discrete, we can assume that each state $\bm{s}$ (action $\bm{a}$) is given a binary representation in $\{ 0,1 \}^{n_S}$ ($\{ 0,1 \}^{n_A}$), with $n_S = \lceil \log|\mathcal{S}| \rceil$ ($n_A = \lceil \log\abs{\mathcal{A}} \rceil$).\\
For every interaction history $\bm{h} = (\bm{s^{(0)}}, \bm{a^{(0)}}, r^{(1)}, \bm{s^{(1)}}, \bm{a^{(1)}}, r^{(2)}, \ldots)$, we call \emph{discounted return} the discounted sum of rewards\footnote{Note that for a discount factor $\gamma < 1$, this sum is $\varepsilon$-close to its first $\left\lceil \log_\gamma{\frac{\varepsilon(1-\gamma)}{|R|_{\textrm{max}}}} \right\rceil$ terms, where $|R|_{\max} = \max_{\bm{s},\bm{a}}{\mathcal{R}(\bm{s},\bm{a})}$.} $R = \sum_{t=1}^{\infty} \gamma^{t-1} r^{(t)}$, a figure of merit describing the performance of the agent on a given task. The goal of RL is to design agents that learn a policy $\pi: \mathcal{S}\times\mathcal{A} \rightarrow \mathbb{R}_+$ such that for any given state $\bm{s} \in \mathcal{S}$, the sampled action $\bm{a} \sim \pi(\bm{a}|\bm{s})$ maximizes the discounted return $R$ for the remainder of the interaction. Even though this policy may not be unique, we call it the optimal policy $\pi^*$.\\

\textbf{Value-based RL. }So-called value-based methods in RL choose to resort to an additional function that explicitly links policy and return. The action-value function $Q(\bm{s},\bm{a}) = \mathbb{E}_\pi[R|\bm{s^{(0)}}=\bm{s},\bm{a^{(0)}}=\bm{a}]$, for instance, is the expected return given a state $\bm{s}$, in which the agent takes the action $\bm{a}$ and follows a policy $\pi$ thereafter. The Bellman equation \cite{bellman57} allows us to express $Q$ recursively:
\begin{equation}\label{eq:bellman}
    Q(\bm{s},\bm{a}) = \mathcal{R}(\bm{s},\bm{a}) + \gamma \sum_{\bm{s'},\bm{a'}} \mathcal{P}(\bm{s'}|\bm{s},\bm{a})\pi(\bm{a'}|\bm{s'}) Q(\bm{s'},\bm{a'})
\end{equation}
This recursive definition allows us to estimate the $Q$-function (that is, store \emph{approximate} values $\Hat{Q}(\bm{s},\bm{a})$ of the action-value function for all \emph{experienced} state-action pairs $(\bm{s},\bm{a})$) from sample interactions with the environment while following a policy $\pi$. Through accumulated learning experience, the approximate values $\Hat{Q}(\bm{s},\bm{a})$ get closer and closer to those of the true $Q$-function for that policy.\\
In order to derive a policy from the stored $\Hat{Q}$-values, a common transformation takes the form of a Boltzmann (or Gibbs) policy:
\begin{equation}\label{eq:boltzmann-policy}
    \pi(\bm{a}|\bm{s}) = \frac{e^{\beta \Hat{Q}(\bm{s},\bm{a})}}{\sum_{\bm{a'}}e^{\beta \Hat{Q}(\bm{s},\bm{a'})}}
\end{equation}
where $\beta > 0$ is an inverse temperature parameter that allows us to adjust the degree of exploration (v.s. exploitation) of the policy.\\
It can be shown \cite{sutton98} that a general procedure which alternates between updating the stored values $\Hat{Q}(\bm{s},\bm{a})$ according to (one of several) update rule(s) derived from the Bellman equation and using these values to update the policy $\pi$ of the agent according to Eq.\ (\ref{eq:boltzmann-policy}) leads to simultaneous convergence of $\Hat{Q}$ to the true $Q$-function and of $\pi$ to the optimal policy $\pi^*$. Two common update rules, belonging to a larger family of \emph{temporal difference (TD) learning} methods, are SARSA:
\begin{align}
	\Hat{Q}(\bm{s^{(t)}},\bm{a^{(t)}}) \leftarrow& (1-\alpha) \Hat{Q}(\bm{s^{(t)}},\bm{a^{(t)}})\\ &+ \alpha [r^{(t+1)} + \gamma \Hat{Q}(\bm{s^{(t+1)}},\bm{a^{(t+1)}})]\nonumber
\end{align}
and $Q$-learning:
\begin{align}
	\Hat{Q}(\bm{s^{(t)}},\bm{a^{(t)}}) \leftarrow& (1-\alpha) \Hat{Q}(\bm{s^{(t)}},\bm{a^{(t)}})\\ &+ \alpha [r^{(t+1)} + \gamma\max_{\bm{a}} \Hat{Q}(\bm{s^{(t+1)}},\bm{a})]\nonumber
\end{align}
where $\alpha$ is a learning rate in $[0,1]$. It can be noticed that the target estimates for $\Hat{Q}$-values in $Q$-learning differ from SARSA by that they are sampled from a $0$-temperature (or greedy) policy $ \pi(\bm{a}|\bm{s}) = \mathbbm{1}_{\bm{a}}\left(\textrm{argmax}_{\bm{a'}} \Hat{Q}(\bm{s},\bm{a'})\right)$.\\

\begin{figure}[h]
	\centering
	\includegraphics[width=1\columnwidth]{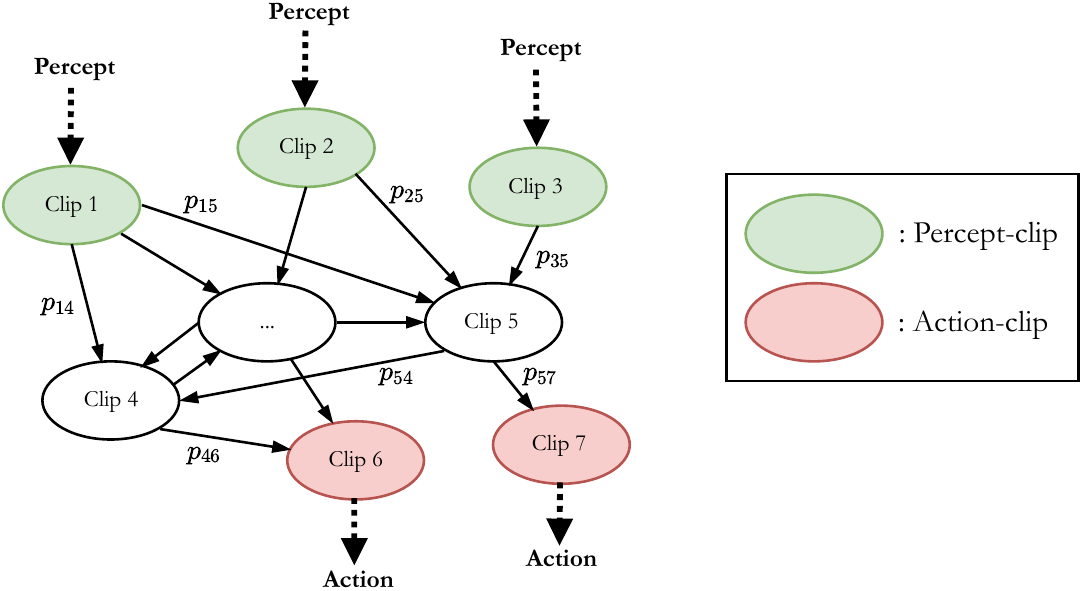}
	\caption{\textbf{A generic Episodic and Compositional Memory}}
	\label{fig:ECM}
\end{figure}

\textbf{Projective simulation. }Projective simulation (PS) \cite{briegel12} is a physics-inspired approach for the design of autonomous learning agents which is different from the value-based methods presented above. The central component of a PS agent is its so-called Episodic and Compositional Memory (ECM), formally represented as a stochastic network of \emph{clips} (see Fig.\ \ref{fig:ECM}). Clips (the nodes of the stochastic network), represent short episodes, i.e., remembered sequences of percepts, rewards and actions, built upon experience of the agent through its interactions with its environment and through certain compositional methods applied on more basic clips: \emph{percept-clips}, associated to single percepts perceived by the agent and \emph{action-clips}, associated to single actions it can perform. In PS, the deliberation mechanism (that is, the selection of an action given a percept) is based on a random walk process, initiated at the percept-clip associated to the percept received at a certain timestep. There are several possibilities in which the random walk process can be realized, but in the basic PS, this consists in hopping between clips until an action-clip is hit. The action associated with that action-clip is then executed on the environment and its associated reward is used to update the weights (or transitions probabilities) of the ECM to increase the probability of the path that led to that action.\\

\begin{figure}[h]
	\centering
	\includegraphics[width=1\columnwidth]{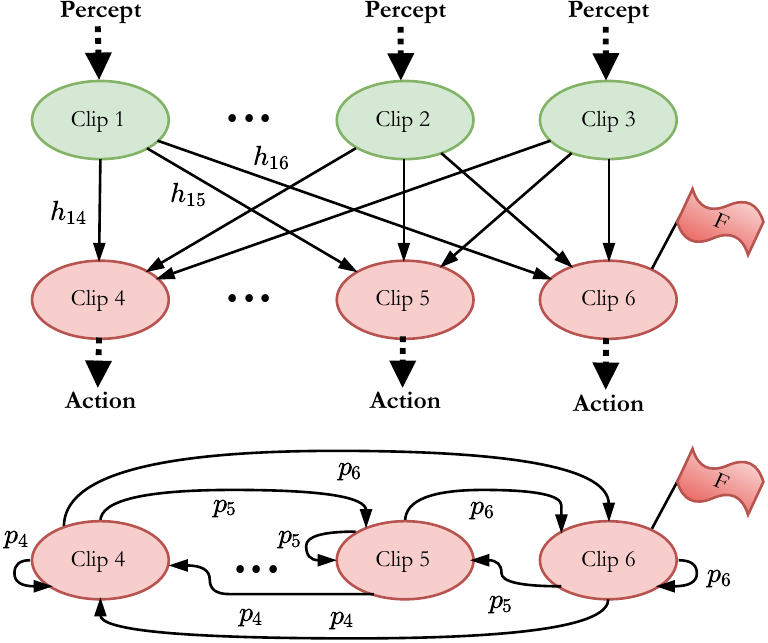}
	\caption{\textbf{A two-layered ECM with flagged actions and its associated action-clip Markov chain}}
	\label{fig:two-layered-ECM}
\end{figure}

\textbf{Two-layered PS. }The simplest architecture of an ECM is given by the so-called \emph{two-layered} PS (see Fig.\ \ref{fig:two-layered-ECM}): it is restricted to one layer of percept-clips linked to another layer of action-clips through weighted connections that we refer to as \emph{h}-values. These are real-valued weights constrained to be greater than (and initialized at) $1$. When normalized according to, e.g., a Boltzmann distribution:
\begin{equation}\label{eq:boltzmann-PS}
    \pi(\bm{a}=j|\bm{s}=i) = \frac{e^{\beta h_{ij}}}{\sum_{j'}e^{\beta h_{ij'}}}
\end{equation}
they give rise to the policy of the PS agent.\\
Learning takes place by updates of (all) h-values according to the following update rule:
\begin{equation}\label{eq:h-values}
    h_{ij} \leftarrow  h_{ij} - \gamma_\textrm{ps}(h_{ij} - 1) + g_{ij}\lambda
\end{equation}
where $\gamma_\textrm{ps} \in [0,1]$ is an internal damping (or forgetfulness) parameter allowing the agent to adapt to changing environments, $\lambda$ is the last \emph{positive} reward experienced by the agent, and the edge-glow $g_{ij}$ is a varying rescaling factor accounting for delayed rewards.\\
Even though these $h$-values do not derive from the Bellman equation (Eq.\ (\ref{eq:bellman})) and hence cannot be described as a value function, we can see a clear connection between the policies in Eqs.\ (\ref{eq:boltzmann-policy}) and (\ref{eq:boltzmann-PS}): both methods rely on storing one real value per state-action pair and normalizing these values for each state/percept to derive a policy. Elaborating on this simple connection, it was shown that these methods have a comparable performance on benchmarking tasks \cite{melnikov18}.\\

\textbf{Reflecting PS. }The PS framework includes many additions to the basic architecture presented above. One of such additions consists in adding a short-term memory structure to the ECM network, in the form of flags assigned to the action clips. A flag indicates if the corresponding action, for a given percept, recently lead to a reward -- a relevant information for agents in changing environments (not MDPs). The agent can use this additional information during his random walk on the ECM to deliberate on its next action: if the encountered action-clip isn't flagged, the random walk process can be reiterated, until a flagged action is hit or the process has been reiterated a maximum number of times $R_{\textrm{max}}$. This new deliberation process, called \emph{reflecting}, can be equivalently described by a Markov chain (or random walk) on the space of action clips (see Fig.\ \ref{fig:two-layered-ECM}), where the transition probabilities are defined according to the last percept $\bm{s_i}$ perceived by the agent. In the case of a two-layered PS agent, these transitions probabilities are:
\begin{equation}\label{eq:rPS-probs}
    p_j = \pi(\bm{a}=j|\bm{s}=i) = \frac{e^{\beta h_{ij}}}{\sum_{j'}e^{\beta h_{ij'}}}
\end{equation}
The selection of a flagged action has then a time complexity governed by the \emph{hitting time} of this Markov chain, which can benefit from a quadratic speed-up using quantum walks \cite{paparo14}. We define these notions in the next section.

\section{Random and quantum walks\label{sec:walks}}

\begin{figure}[h]
	\centering
	\includegraphics[width=1\columnwidth]{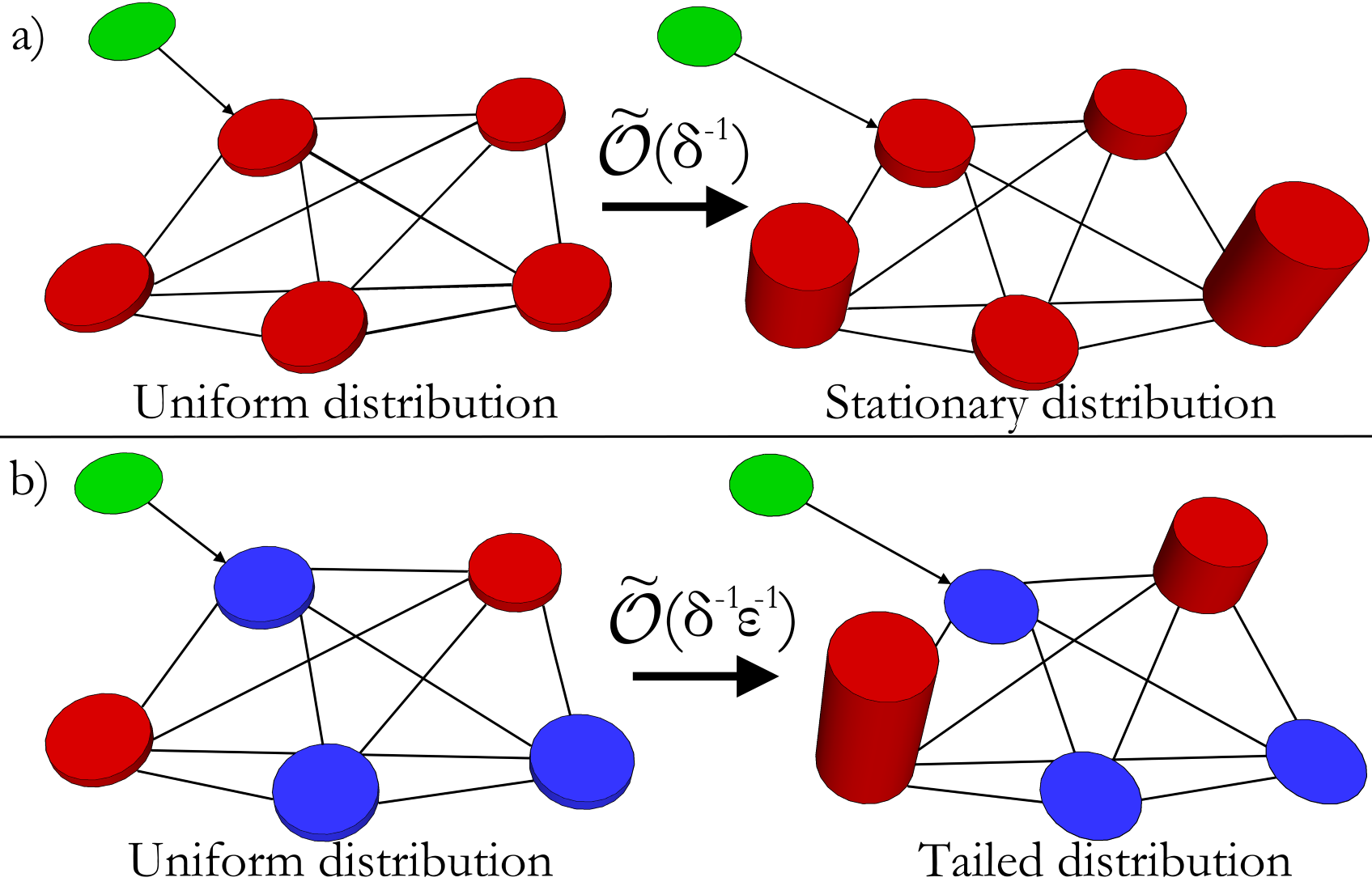}
	\caption{\textbf{Mixing-time v.s.\ Hitting-time random walks} Note that the graphs are generally non-symmetric (that is, the connections are directed), but were represented symmetric for simplicity. Marked elements are colored in red.}
	\label{fig:walk}
\end{figure}

\textbf{Mixing time. }A discrete-time random walk on a graph is described by a Markov chain (MC) that is specified by a transition matrix $P$ gathering the transition probabilities $[P]_{i,j}=P_{ij}$ from node $i$ to node $j$ of the graph. When the MC is ergodic (i.e., irreducible and aperiodic), it has a unique stationary distribution $\bm{\pi}=P\bm{\pi}$ over the nodes of the graph. $\bm{\pi}$ can be approximated by repeated applications $P^t\bm{\pi_0}$ of $P$ to any initial distribution $\bm{\pi_0}$, or equivalently, by applying $t$ steps of the random walk $P$ to any initial distribution $\bm{\pi_0}$ over the nodes of the graph. The minimal number of applications of $P$ needed to obtain an $\epsilon'$-close approximation of $\bm{\pi}$ (in $l_p$-norm) starting from \emph{any} $\bm{\pi_0}$ is called the \emph{mixing time} of $P$:
\begin{equation}
t_{\epsilon'}^{mix} = \mathrm{min}\{t : \forall \bm{\pi_0}, \ ||P^t \bm{\pi_0} - \bm{\pi}||_p \ \leq \epsilon'\}
\end{equation}
When the MC is moreover reversible, its mixing time can be related to its spectral properties. More especially, we have \cite{aldous82}:
\begin{equation}
\left\lfloor\frac{1-\delta}{2\log(2\epsilon')}\frac{1}{\delta}\right\rfloor\ \leq\ t_{\epsilon'}^{mix} \leq \ \left\lceil\frac{1}{\delta}\ \mathrm{log}\left(\frac{1}{\epsilon'\pi^*}\right)\right\rceil
\end{equation}
where $\delta = 1 - \max_{i>0}|\lambda_i| \in (0,1]$ (i.e., the difference between the largest and second largest eigenvalues in absolute value of the matrix $P$) is the \emph{spectral gap} of $P$ and $\pi^* = \min_i\{\pi_i\}$. We say that the mixing time of a reversible MC $P$ is of the order $\tilde{\mathcal{O}}(\delta^{-1})$\footnote{Intuitively, this has to do with the fact that all eigenvalues modulus of $P$ are strictly less than $1$ (except for the eigenvalue 1 associated to the stationary distribution), and repeated applications of $P$ make its (non-eigenvalue-1) eigenspaces fade away with a geometric rate corresponding to their eigenvalue modulus. The slowest fading rate is given by the second largest eigenvalue modulus.}.\\

\textbf{Hitting time. }Now suppose that the graph contains an (unknown) subset of marked nodes and that we are interested in finding one of these through a random walk on the graph. Let's call $\varepsilon \in [0,1]$ the proportion of marked elements in the stationary distribution $\bm{\pi}$ of $P$. Then note that a sample from $\bm{\pi}$ has probability $\varepsilon$ of being marked, which implies that on average $1/\varepsilon$ samples from $\bm{\pi}$ are needed to find a marked element. More generally, we are interested in preparing the \emph{tailed distribution} of marked elements at stationarity $\bm{\pi}'$, which is the renormalized stationary distribution with support only on the marked nodes. The average number of random walk steps needed to prepare $\bm{\pi}'$ starting from any distribution $\bm{\pi_0}$ is called the \emph{hitting time} of $P$ and satisfies:
\begin{equation}
t_{\epsilon'}^{hit} = t_{\epsilon'}^{mix}/\varepsilon
\end{equation}
In the case of a reversible MC, the hitting time is of the order $\tilde{\mathcal{O}}(\delta^{-1}\varepsilon^{-1})$.\\

\textbf{Metropolis-Hastings (MH). }The MH algorithm is a general Markov chain Monte Carlo method that allows one to sample from a probability distribution over a configuration space $\mathcal{X}$ given only access to its corresponding density function $f(\bm{x})$ without having to evaluate its value for all $\bm{x}$. When dealing with the energy-based models listed in Sec.\ \ref{sec:hybrid-EBRL}, $\mathcal{X}$ can either be the space $\mathcal{A}$ of all actions and $f(\bm{x}) = -F^{\bm{\theta}}(\bm{s,a})$ for a fixed state $\bm{s}$, or $\mathcal{X}$ is the space $\mathcal{A}\times\mathcal{H}$ of all configurations of action and hidden units and $f(x) = - E^{\bm{\theta}}(\bm{s,a,h})$. MH constructs a reversible MC $P$ having the desired distribution as stationary distribution by transforming a candidate transition matrix $K$ into $P$. This construction starts by introducing an acceptance ratio $R_{\bm{x}\bm{x'}}=\min\left\{ 1,\frac{\exp(\beta f(\bm{x'}))K_{\bm{x'}\bm{x}}}{\exp(\beta f(\bm{x}))K_{\bm{x}\bm{x'}}} \right\}$ for each pair $(\bm{x},\bm{x'}) \in \mathcal{X}$ which is used to define $P$ as follows:
\begin{equation}
	P_{\bm{x}\bm{x'}} = \left\{
    \begin{array}{cl}
        R_{\bm{x}\bm{x'}}K_{\bm{x}\bm{x'}} & \mbox{if }\bm{x'}\neq\bm{x}\\
        1 - \sum_{\bm{x''}\neq\bm{x}} R_{\bm{x}\bm{x''}}K_{\bm{x}\bm{x''}} & \mbox{if }\bm{x'}=\bm{x}
    \end{array}
\right.
\end{equation}
In practice, $K$ is chosen such that the sets $N(\bm{x})=\{ \bm{x'}: K_{\bm{x}\bm{x'}} > 0 \}$ define a neighborhood structure on $\mathcal{X}$ and $\max_{\bm{x}\in\mathcal{X}} N(\bm{x}) \ll \abs{\mathcal{X}}$. A possible choice for $N(\bm{x})$ and $K(\bm{x},\cdot)$ is, e.g., a Hamming ball around $\bm{x}$ and the uniform distribution over it, respectively.\\
If $K_{\bm{x}\bm{x}} > 0 \ \ \forall \bm{x}\in\mathcal{X}$ and the graph defined by $N$ is connected, then MH \cite{hastings70} ensures that $P$ is reversible and has as stationary distribution:
\begin{equation}\label{eq:policy-MH}
	\pi_\beta(\bm{x}) = \frac{e^{\beta f(\bm{x})}}{\sum_{\bm{a'}}e^{\beta f(\bm{x'})}}
\end{equation}
\\
\textbf{Gibbs sampling (GS). }Gibbs sampling \cite{geman84} is a particular instance of the Metropolis-Hastings algorithm where the neighborhood $N(\bm{x})$ is a Hamming ball of radius $1$ around $\bm{x}$ (including $\bm{x}$ itself) and the candidate transition matrix $K$ is given by:
\begin{equation}
	K_{\bm{x}\bm{x'}} = \left\{
    \begin{array}{cl}
       \frac{1}{\log{\abs{\mathcal{X}}}}\frac{e^{\beta f(\bm{x'})}}{e^{\beta f(\bm{x})}+e^{\beta f(\bm{x'})}} & \mbox{if }\bm{x'}\in N(\bm{x})\\
        0 & \mbox{if }\bm{x'}\notin N(\bm{x})
    \end{array}
\right. 
\end{equation}
This $K$ has the special property that its corresponding acceptance ratio $R_{\bm{x}\bm{x'}}=1$ for all $\bm{x'}\in N(\bm{x})$. In particular, the conditions for reversibility of $P$ are always met.\\


\textbf{Szegedy walks. }Szegedy-type quantum walks \cite{szegedy04,magniez11} are a quantum analogue of classical random walks. They generalize an MC $P$ to a unitary $W(P)$ called the \emph{Szegedy walk operator}. Here we describe how this unitary can be constructed and explain its interesting properties.\\
Let $\mathcal{H}$ be the Hilbert space that shares the same basis vectors as the MC $P=[P_{ij}]$. So-called \emph{quantum diffusion operators} $U_P$ and $V_P$ act on $\mathcal{H}\otimes\mathcal{H}$ as:
\begin{equation*}
    \begin{cases}
		U_P \ket{i}\ket{0} = \ket{i}\sum_j \sqrt{P_{ij}}\ket{j}\\
		V_P \ket{0}\ket{i} = \sum_j \sqrt{P^*_{ij}}\ket{j}\ket{i}
    \end{cases}
\end{equation*}
where $P^*$ is the time-reversed MC\footnote{Note that the asterisk on $P^*$ does not indicate complex conjugation.} associated to $P$. Its elements are given by $P^*_{ij}=P_{ji}\pi_i/\pi_j$, where $\bm{\pi}=(\pi_i)$ is the stationary distribution of $P$. When $P^*=P$, we say that the MC is \emph{reversible}. Although Szegedy walks can be defined if this property is not fulfilled, one would require additional access to $P^*$ to be able to implement $V_P$. Because this is not usually the case (the expression of $P^*$ involves the stationary distribution of $P$), the reversibility of $P$ is then crucial for the implementation of $V_P$.\\
With the quantum diffusion operators $U_P$ and $V_P$ at hand, we can construct the Szegedy walk operator $W(P)$, which is the reflection over the spaces $A$ and $B$:
\begin{equation*}
    \begin{cases}
		A =\textrm{span}\{ U_P \ket{i}\ket{0} : i=1,\ \ldots,\ N \}\\
		B =\textrm{span}\{ V_P \ket{0}\ket{i} : i=1,\ \ldots,\ N \}
    \end{cases}
\end{equation*}
\begin{equation*}
    W(P) = \textrm{ref}(B)\ \textrm{ref}(A)
\end{equation*}
where
\begin{equation*}
    \begin{cases}
		\textrm{ref}(A) = U_P\ \left[ \mathbbm{1}_N \otimes \textrm{ref}(0) \right]\ U_P^\dagger\\
		\textrm{ref}(B) = V_P\ \left[ \textrm{ref}(0) \otimes \mathbbm{1}_N \right]\ V_P^\dagger
    \end{cases}
\end{equation*}
and
\begin{equation*}
    \textrm{ref}(0) = 2{|{0}\rangle}{\langle{0}|} - \mathbbm{1}_N
\end{equation*}
We refer to $A+B$ as the busy subspace and to its orthogonal complement $A^\perp \cap B^\perp$ as the idle subspace. $W(P)$ acts as the identity on the idle subspace, but, on the busy subspace, its unique eigenvector with eigenvalue $1$ (or, equivalently, eigenphase $0$) is
\begin{equation*}
    \ket{\pi}\ket{0} = \sum_{i} \sqrt{\pi_i}\ket{i}\ket{0}
\end{equation*}
This follows from Refs. \cite{szegedy04,magniez11}.\\
As we were interested in the spectral gap $\delta$ of $P$ classically, we are interested in the \emph{phase gap} $\Delta = 2\min_{i>0}\abs{\theta_i}$, i.e., the smallest non-zero eigenphase of $W(P)$, in the quantum case. Given that the eigenvalues of $P$ and the eigenphases of $W(P)$ are related by $\abs{\lambda_i} = \cos(\theta_i)$ $\forall i$, it can be shown \cite{szegedy04,magniez11} that:
\begin{equation*}
    \Delta \geq \sqrt{2\delta}
\end{equation*}
i.e., the phase gap of $W(P)$ is quadratically larger that the spectral gap of $P$. This is a crucial property for Szegedy quantum walks as a ``quantum search'' \cite{brassard02,yoder14} for $\ket{\pi}\ket{0}$ on the busy subspace of $W(P)$ is governed by its phase gap as $\tilde{\mathcal{O}}(\Delta^{-1})$. This is quadratically faster than sampling from $\bm{\pi}$ classically.\\

Note that, in the case of a two-layered reflecting PS agent, the quantum state $\ket{\pi}$ can be easily prepared, as the MC defined by the transition probabilities of Eq. (\ref{eq:rPS-probs}) has a trivial mixing time $\delta=1$. This makes it a purely hitting-type random walk, for which one is interested in preparing $\ket{\pi_F} = \frac{1}{Z_F}\sum_{i \in F} \sqrt{\pi_i}\ket{i}$, where $F$ is the set of flagged actions associated to a given percept. To this end, the authors of Ref.\ \cite{paparo14} use Szegedy walks and amplitude amplification \cite{brassard02,yoder14} to prepare $\ket{\pi_F}$ \emph{from} $\ket{\pi}$. The complexity of this preparation is in $\tilde{\mathcal{O}}(\sqrt{\varepsilon^{-1}})$, where $\varepsilon = \abs{\braket{\pi}{\pi_F}}^2$. This is quadratically faster than the complexity of the classical random walk, in $\mathcal{O}(\varepsilon^{-1})$.\\
The quantum deliberation relying on quantum simulated annealing (Sec.\ \ref{sec:quantum-gibbs-preparation}) does not rely on flagged actions, and hence is a mixing-type random walk.

\section{Boltzmann machines and energy-based models\label{sec:bm_energy-based}}

\textbf{Restricted Boltzmann machines. }Boltzmann machines \cite{ackley85} are energy-based generative models, meaning that they can learn probability distributions $\hat{p}(\bm{x})$ over a set $\mathcal{X} \subset \{0,1\}^{n}$, by encoding a parametrized distribution $p^{\bm{\theta}}(\bm{x})$. This probability distribution is specified by an energy function $E^{\bm{\theta}}(\bm{x})$ defined over the entire set $\mathcal{X}$ and normalized by a Boltzmann/Gibbs distribution:
\begin{equation}
	P^{\bm{\theta}}(\bm{x}) = \frac{e^{-E^{\bm{\theta}}(\bm{x})}}{\sum_{\bm{x'}}e^{-E^{\bm{\theta}}(\bm{x'})}}
\end{equation}
This energy function takes the form of an Ising model spin Hamiltonian which accounts for pair interactions $w_{ik}$ and individual external magnetic fields $b_i$ ($(\{w_{ik}\}_{i,k},\{b_i\}_i)=\bm{\theta} \in \mathbb{R}^d$) acting between and over spins $x_i$ and $x_k$. Different configurations of these spins correspond to different binary vectors $\bm{x} \in \mathcal{X}$. A commonly used type of BMs are the so-called restricted Boltzmann machines \cite{fischer12,hinton12}, characterized by their bipartite structure (see Fig.\ \ref{fig:RBM-RL}), which restricts the interactions $w_{ik}$ in the Ising model Hamiltonian/energy function. RBMs divide the spins $\{x_i\}_i$ between two subsets $\bm{v}=\{v_i\}_i$ and $\bm{h}=\{h_k\}_k$, respectively called visible and hidden units. The visible units represent the data that we want to encode in the model, while the hidden (or latent) units allow to have more degrees of freedom in the specification of the energy function and allow to represent more complex distributions over the visible units. The RBM model effectively encodes the probability distribution:
\begin{equation}\label{eq:prob-RBM}
	P^{\bm{\theta}}(\bm{v}) = \frac{\sum_{\bm{h}}e^{-E^{\bm{\theta}}(\bm{v},\bm{h})}}{\sum_{\bm{v'},\bm{h'}}e^{-E^{\bm{\theta}}(\bm{v'},\bm{h'})}} = \frac{e^{-F^{\bm{\theta}}(\bm{v})}}{\sum_{\bm{v'}}e^{-F^{\bm{\theta}}(\bm{v'})}}
\end{equation}
where
\begin{equation}
	-E^{\bm{\theta}}(\bm{v},\bm{h})  = \sum_{i,k} w_{ik}v_ih_k + \sum_{i} b_iv_i + \sum_{k} b_kh_k 
\end{equation}
and
\begin{align*}
	F^{\bm{\theta}}(\bm{v}) &= -\mathrm{log}\left(\sum_{\bm{h'}}{e^{-E^{\bm{\theta}}(\bm{v},\bm{h'})}}\right)\\
    &=-\mathrm{log}\left(\frac{e^{-E^{\bm{\theta}}(\bm{v},\bm{h})}}{P(\bm{h} | \bm{v})}\right) \quad \forall \bm{h}\\
    &=E^{\bm{\theta}}(\bm{v},\bm{h}) + \mathrm{log}\left(P(\bm{h} | \bm{v})\right) \quad \forall \bm{h}\\
    &=\sum_{\bm{h}}{P(\bm{h} | \bm{v})\left[E^{\bm{\theta}}(\bm{v},\bm{h}) + \mathrm{log}\left(P(\bm{h} | \bm{v})\right)\right]}\\
    &=E^{\bm{\theta}}(\bm{v},\langle \bm{h}\rangle_{P(\bm{h} | \bm{v})}) - H(\bm{h} | \bm{v}) 
\end{align*}
is the so-called \emph{free energy} of the RBM (the average energy after tracing out the hidden units).\\

\textbf{Derivation of the free-energy neural network. }We derive below the expression of an RBM free energy as a feedforward neural network. This derivation is inspired by Ref.\ \cite{martens13}.\\
Starting from the expression of an RBM probabiliy distribution:
\begin{widetext}
\begin{align*}
P^{\bm{\theta}}(\bm{v}) &= \frac{1}{Z^{\bm{\theta}}} \sum_{\bm{h}} \exp\left(\sum_{i,k} w_{ik}v_ih_k +\sum_i b_iv_i + \sum_k b_kh_k\right)\\
&= \frac{1}{Z^{\bm{\theta}}} \exp\left(\sum_i b_iv_i\right) \prod_k \sum_{h_k \in \{0,1\}}\exp\left(\sum_{i} w_{ik}v_ih_k + b_kh_k\right)\ \# \sum_{\bm{x}\in\{0,1\}^2}e^{ax_1 + bx_2} = (1+e^a)(1+e^b)\\
&= \frac{1}{Z^{\bm{\theta}}} \exp\left(\sum_i b_iv_i\right) \prod_k  \left(1+\exp\left(\sum_{i} w_{ik}v_i+ b_k\right)\right)\ \# \sum_{\bm{x}\in\{0,1\}^2}e^{ax_1 + bx_2} = (1+e^a)(1+e^b)\\
&= \frac{1}{Z^{\bm{\theta}}} \exp\left(\sum_i b_iv_i\right) \exp \left( \sum_k \log\left(1+\exp\left(\sum_{i} w_{ik}v_i+ b_k\right)\right) \right)\ \#\ \prod_k a_k = \exp\left(\sum_k \log(a_k)\right)\\
&= \frac{1}{Z^{\bm{\theta}}} \exp\left(\sum_i b_iv_i + \sum_k \textrm{softplus}\left(\sum_{i} w_{ik}v_i+ b_k\right)\right)\ \#\ \textrm{softplus}(x) = \log (1 + \exp(x))\\
&= \frac{1}{Z^{\bm{\theta}}} \exp\left(-F^{\bm{\theta}}(\bm{v})\right)
\end{align*}
\end{widetext}
we end up with:
\begin{equation}
-F^{\bm{\theta}}(\bm{v}) = \sum_i b_iv_i + \sum_k \textrm{softplus}\left(\sum_{i} w_{ik}v_i+ b_k\right)
\end{equation}
which takes the form of the feedforward neural network depicted in Fig.\ \ref{fig:RBM-FF}.\\

\begin{figure}
	\centering
	\includegraphics[width=1\columnwidth]{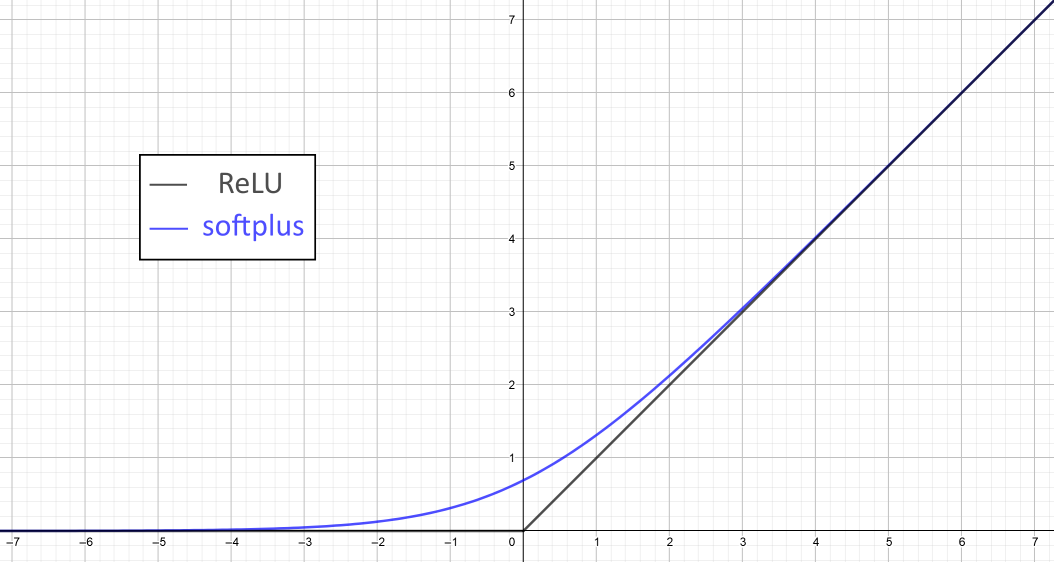}
	\caption{\textbf{A comparison between the commonly used ReLU activation function $\max(0,x)$ and the softplus function $x+\log(1+e^{-x}) =\log(1+e^x)$ appearing in the expression of the free energy of an RBM}}
	\label{fig:ReLU}
\end{figure}

\textbf{Eigendecomposition of $H_\text{QBM}$. }We write the QBM Hamiltonian in the computational basis $(\bm{v},\bm{h})$, where we divided the Hilbert space on which it acts into visible and hidden units $\mathcal{V}\otimes\mathcal{H}$. Separating the $\sigma^z$ and $\sigma^{z}\sigma^{z}$ operators from the $\sigma^{x}$ operators in the Hamiltonian of Eq.\ (\ref{eq:QBM-Hamiltonian}) allows to express the latter as:
\begin{multline*}
H_\text{QBM} = 
\begin{pmatrix}
H^{z,zz}_1 & 0 & \cdots & 0 \\
0 & H^{z,zz}_2 & \cdots & 0 \\
\vdots  & \vdots  & \ddots & \vdots  \\
0 & 0 & \cdots & H^{z,zz}_{\abs{\mathcal{V}}\times\abs{\mathcal{H}}} 
\end{pmatrix}\\
\hfill + I_{\bm{v}} \otimes
\begin{pmatrix}
H^{x}_{1,1} & \cdots & H^{x}_{1,\abs{\mathcal{H}}} \\
\vdots  & \ddots & \vdots  \\
H^{x}_{\abs{\mathcal{H}},1} & \cdots & H^{x}_{\abs{\mathcal{H}},\abs{\mathcal{H}}} 
\end{pmatrix}\\
 =
\begin{pmatrix}
M_1 & 0 & \cdots & 0 \\
0 & M_2 & \cdots & 0 \\
\vdots  & \vdots  & \ddots & \vdots  \\
0 & 0 & \cdots & M_{\abs{\mathcal{V}}}
\end{pmatrix}
= \sum_{\bm{v}} \ket{\bm{v}}\bra{\bm{v}} \otimes M_{\bm{v}}
\end{multline*}
where $\left\{M_1, \ldots, M_{\abs{\mathcal{V}}}\right\}$ are non-diagonal matrices acting on the subspaces $\left\{\text{span}\{\bm{v_1}\}\otimes\mathcal{H}, \ldots, \text{span}\{\bm{v}_{\abs{\mathcal{V}}}\}\otimes\mathcal{H}\right\}$ respectively, for $\bm{v_{i}}$ computational basis states of visible units. This block diagonal form implies that the eigenvectors of $H_\text{QBM}$ are block vectors of the form $\ket{\bm{v},\bm{h_v}}$, where $\ket{\bm{v}}$ is a computational basis state and $\ket{\bm{h_v}}$ can be a superposition of hidden units configurations.\\

\textbf{Derivation of the PS error function. }We showed in Sec.\ \ref{sec:EB-RL} how to derive an error function from the $Q$-learning update rule. Due to its shared origin, deriving an error function from the SARSA update rule can be done similarly. In contrast, $h$-values are not linked to the Bellmann equation and hence do not represent estimates of the (bounded) return. Instead, they are allowed to accumulate rewards indefinitely (and hence diverge). Diverging $h$-values can still lead to converging policies after normalization but they constitute a problem for the derivation of an error function, as this error can never be minimized to 0, hence making neural networks harder to train. One mechanism in PS that allows us to avoid this divergence is damping, and we show here how it can be translated into a regularization term in the PS error function.\\

Let's take the update rule of PS (Eq. (\ref{eq:h-values})) with damping $\gamma_\textrm{ps} = 0$ and edge-glow $g_{ij} = 1$:
\begin{equation}
h_{ij} \leftarrow  h_{ij} + \lambda
\end{equation}
This setup cancels the damping mechanism and supposes that edge-glow values $g_{ij} \neq 1$ are absorbed in the reward $\lambda$. This update rule gives the following error function:
\begin{equation}
\mathcal{E}_\textrm{PS}(\bm{s^{(t)}},\bm{a^{(t)}}) = r^{(t+1)}
\end{equation}
which is never 0 as long as the given transition is rewarded (which is the case in a static environment) and leads to diverging $h$-values.\\
To avoid this divergence, we take into account damping by effectively introducing a regularization term:
\begin{equation}\label{eq:error-ps}
\mathcal{E}_\textrm{PS}(\bm{s^{(t)}},\bm{a^{(t)}}) = r^{(t+1)} - \gamma_\textrm{ps}h_{ij}
\end{equation}
such that the error is now 0 when $h_{ij} = r^{(t+1)} / \gamma_\textrm{ps}$, hence effectively bounding the $h$-values. Note that $\mathcal{E}_\textrm{PS}(\bm{s^{(t)}},\bm{a^{(t)}}) = h_{ij}^{(t+1)} - h_{ij}^{(t)}$, so this new error corresponds to the update rule:
\begin{equation}
h_{ij} \leftarrow  (1-\gamma_\textrm{ps})h_{ij} + \lambda
\end{equation}
which is the original PS update rule (Eq.\ (\ref{eq:h-values}), with the edge-glow fixed to 1).

\section{Deferred theorems and derivations for quantum subroutines\label{sec:deferred-thms}}

\begin{theorem}[Harrow and Wei \cite{harrow20} Theorem 10]\label{thm:harrow-wei}
Assume that we are given a prior distribution $\Pi_0(x)$ and a likelihood function $L(x)\geq 0$, so that we can parametrize the partition function $Z(\beta') =\sum_x \Pi_0(x)e^{-\beta'L(x)}$ at each inverse temperature $\beta'\in[0,1]$. Assume that we can generate the state $\ket{\Pi_0}$ corresponding to the coherent encoding of the prior, and assume that for every inverse temperature $\beta'$ we have a Markov chain $P_{\beta'}$ with stationary distribution $\ket{\Pi_{\beta'}}$ and spectral gap lower-bounded by $\delta$. Then, for any $\varepsilon>0$, $\eta >0$, there is a quantum algorithm that, with probability at least $1-\eta$, produces a state $\ket{\widetilde{\Pi_1}}$ so that $\left\lVert\ket{\widetilde{\Pi_1}}-\ket{\Pi_1}\right\rVert \leq \varepsilon$ for $\ket{\Pi_1}$ the coherent encoding of the posterior distribution $\Pi_1(x) = \Pi_0(x)\frac{e^{-L(x)}}{Z(1)}$. The algorithm uses 
\begin{equation*}
\widetilde{\mathcal{O}} \left( \sqrt{\mathbb{E}_{\Pi_0} [L(x)] / \delta} \right) 
\end{equation*}
total steps of the quantum walk operators corresponding to the Markov chains $M_{\beta'}$.
\end{theorem}

In the case of a DEBN, we use Theorem \ref{thm:harrow-wei} with $x = \bm{a}$,  $L(x) = - \beta M^{\bm{\theta}}_-(\bm{s,a})$ and $\Pi_0(x) = \frac{1}{\abs{\mathcal{A}}}$. $M^{\bm{\theta}}_-(\bm{s,a})$ is obtained from $M^{\bm{\theta}}(\bm{s,a})$ by the transformation $M^{\bm{\theta}}_-(\bm{s,a}) = M^{\bm{\theta}}(\bm{s,a}) - \max \left\{ 0,\max_{\bm{a}} M^{\bm{\theta}}(\bm{s,a})\right\}$, such that $- 2 \max_{\bm{a}} \abs{M^{\bm{\theta}}(\bm{s,a})} \leq M^{\bm{\theta}}_-(\bm{s,a}) \leq 0$, i.e., $0 \leq L(x)\leq 2\beta\max_{\bm{a}}\norm{M^{\bm{\theta}}}$. Hence, sampling from the distribution $\pi_\beta(\bm{a}|\bm{s}) = \frac{e^{\beta M^{\bm{\theta}}_-(\bm{s,a})}}{Z_{-,\beta}(\bm{s})} = \frac{e^{\beta M^{\bm{\theta}}(\bm{s,a})}}{Z_\beta(\bm{s})}$ has complexity in $\widetilde{\mathcal{O}} \left( \sqrt{\mathbb{E}_{\Pi_0} [L(x)] / \delta} \right) = \widetilde{\mathcal{O}} \left( \sqrt{\frac{1}{\abs{\mathcal{A}}}\sum_{\bm{a}} L(\bm{a}) / \delta} \right) \leq  \widetilde{\mathcal{O}} \left( \sqrt{\beta \norm{M^{\bm{\theta}}}  / \delta} \right)$.\\

Note: While the value $\norm{M^{\bm{\theta}}} = \max_{\bm{a}}\abs{M^{\bm{\theta}}(\bm{s,a})}$ is not needed to implement the quantum walk operator of $P_\beta'$ associated to $M_-(\bm{s,a})$ (see Appendix \ref{sec:coherent-access}), it is actually needed in a subroutine behind the algorithm of Theorem \ref{thm:harrow-wei} (see Algorithm 1 in Ref. \cite{harrow20}).\\

A similar approach can be taken for DBMs, where now $x = (\bm{a,h})$, $L(x) = \beta E_-(\bm{s,a,h})$ and $\Pi_0(x) = \frac{1}{\abs{\mathcal{A}}2^K}$. The resulting complexity is now in $\widetilde{\mathcal{O}} \left( \sqrt{\frac{1}{\abs{\mathcal{A}}2^K}\sum_{\bm{a,h}} L(\bm{a,h}) / \delta} \right) \leq  \widetilde{\mathcal{O}} \left( \sqrt{\beta \norm{M^{\bm{\theta}}}  / \delta} \right)$ (for a different spectral gap $\delta$).\\

 \begin{theorem}[Gibbs state preparation with sparse access \cite{van20} Lemma 44]
\label{thm:van-et-al}
    Let $H \in \mathbb{C}^{N\times N}$ be a $d$-sparse Hamiltonian. Suppose that $I \preceq H$, we are given $C\in \mathbb{R}_+$ such that $\norm{H}\leq 2C$, and we know a lower bound $z\leq Z=\tr[e^{-H}]$.
    If $\epsilon \in (0,1/3)$, then we can prepare a purified Gibbs state $\ket{\gamma}_{AB}$ such that
    \begin{equation*}
      \norm{\mathrm{Tr}_B \left[\ket{\gamma}\bra{\gamma}_{AB} \right] - \frac{e^{-H}}{Z}} \leq \epsilon
    \end{equation*}
    using
    \begin{equation*}
      \mathcal{O} \left(\sqrt{\frac{N}{z}} Cd \log \left( \frac{C}{\epsilon}\right) \log \left( \frac{1}{\epsilon} \right) \right)
    \end{equation*}
    queries, and
    \begin{equation*}
      \mathcal{O} \left(\sqrt{\frac{N}{z}}  Cd \log \left( \frac{C}{\epsilon}\right) \log \left( \frac{1}{\epsilon} \right)  \left[ \log(N) + \log^{5/2} \left( \frac{C}{\epsilon} \right) \right] \right)
    \end{equation*}
    gates.
\end{theorem}

For the DEBN, we have $N=2^{A}$, while $N=2^{AK}$ for the DBM and QBM, where $K$ is the number of hidden units of these models. We enforce the condition $I \preceq H$ by choosing $H = \beta H_\text{model} - (\beta\tilde{\lambda}_\text{min} - 3/2)I$, where $\tilde{\lambda}_\text{min}$ is either bounded by the weights of the model or estimated via a minimum finding subroutine. This choice also gives us natural bounds $z = e^{-2}$ and $C = \beta\norm{H_\text{model}}$. Since the DEBN and DBM Hamiltonians are diagonal, their sparsity is $d=1$. As for the QBM Hamiltonian, the $\sigma^{x}$ terms each add at most one non-zero entry to each row/column, which leads to a sparsity $d = 1 + K$. Hence the complexity of preparing the purified Gibbs states of the DEBN and DBM/QBM is in $\widetilde{\mathcal{O}} \left(\sqrt{2^{A}} \beta\norm{H_\text{model}}\right)$ and $\widetilde{\mathcal{O}} \left(\sqrt{2^{AK}} \beta\norm{H_\text{model}}\right)$ respectively.

\begin{theorem}[Average energy estimation \cite{chowdhury20} Theorem 2]
\label{thm:chowdhury-et-al_est-avg-energy}
Let $H$ be a Hamiltonian that can be represented as a linear combination of unitaries $H = \sum_{k=1}^{K} \alpha_k U_k$ with $\abs{\alpha}_1 = \sum_k \abs{\alpha_k}$ for $K\in O(\log N)$, and let $U_{\psi}$ be a unitary that prepares a state $\ket{\psi}$. The number of applications of the unitaries $U_P$, $U_S$ and the unitary $U_{\psi}$ to estimate the average energy $\bra{\psi}H\ket{\psi}$ within error $\epsilon$ with probability at least $2/3$ is in
$\mathcal{O}\left(\left\|\alpha\right\|_1/\epsilon\right)$.
\end{theorem}

\begin{theorem}[von Neumann entropy estimation \cite{chowdhury20} Theorem 1]
\label{thm:chowdhury-et-al_est-entropy}
Given access to an oracle $ U_{\rho}$ that prepares a purification $ \sum_j \sqrt{p_j} \ket{j}\ket{\psi_j} $ of a density matrix $ \rho = \sum_j p_j \ketbra{\psi_j}{\psi_j} $ where $p_j \ge p_{\min}$, the von Neumann entropy of the state $\rho$ can be estimated within error $\epsilon$ with probability greater than $2/3$ with $\widetilde{\mathcal{O}}(1/(p_{\min}^2 \epsilon))$ preparations of the purified density operator.
\end{theorem}

\textbf{Derivation of the gradient expression for DBMs. }This derivation is relatively easy because we can use the property $\partial_\theta e^{-H} = -\partial_\theta H e^{-H}$ with a diagonal Hamiltonian $H$. Define $Z_{\bm{s,a}} = \sum_{\bm{h}} e^{-H(\bm{s,a,h})}$, then
\begin{align*}
\frac{\partial F}{\partial \theta}(\bm{s,a}) &= - \frac{\partial}{\partial \theta} \log\left( Z_{\bm{s,a}} \right)\\
&= -\frac{1}{Z_{\bm{s,a}}} \frac{\partial Z_{\bm{s,a}}}{\partial \theta}\\
&= -\frac{1}{Z_{\bm{s,a}}} \sum_{\bm{h}} \frac{\partial e^{-H(\bm{s,a,h})}}{\partial \theta}\\
&= \frac{1}{Z_{\bm{s,a}}} \sum_{\bm{h}} \frac{\partial H(\bm{s,a,h})}{\partial \theta} e^{-H(\bm{s,a,h})}\\
&= \sum_{\bm{h}} P(\bm{h}|\bm{s,a}) \frac{\partial H(\bm{s,a,h})}{\partial \theta}\\
&= \langle \partial_\theta H \rangle_{P(\bm{h}|\bm{s},\bm{a})}
\end{align*}
where $\partial_\theta H$ is the $\sigma^{z}_{l}$ or $\sigma^{z}_{l}\sigma^{z}_{l'}$ operator associated to the parameter $\theta$ (or equivalently the random variables associated to units $l$ and $l'$ of the DBM), depending if the latter is a bias or weight respectively. $P(\bm{h}|\bm{s,a})$ is the Gibbs distribution obtained by fixing the operators $\sigma^{z}$ associated to state and action units to the configuration $(\bm{s,a})$.\\

\textbf{Derivation of the gradient expression for QBMs. }Since the QBM Hamiltonian involves both $\sigma^x$ and $\sigma^z$ terms, we have that $H$ and $\partial_\theta H$ do not commute, and hence $\partial_\theta e^{-H} \neq -\partial_\theta H e^{-H}$. This makes the derivation of the gradient a bit more involved.\\ Define $Z_{\bm{s,a}} = \text{Tr}\left[\Lambda_{\bm{s,a}}e^{-H}\right]$, then
\begin{align*}
\frac{\partial F}{\partial \theta} &= - \frac{\partial}{\partial \theta} \log\left( Z_{\bm{s,a}} \right)\\
&= -\frac{1}{Z_{\bm{s,a}}} \frac{\partial Z_{\bm{s,a}}}{\partial \theta}\\
&= - \frac{\text{Tr}\left[\Lambda_{\bm{s,a}}\partial_\theta e^{-H}\right]}{\text{Tr}\left[\Lambda_{\bm{s,a}}e^{-H}\right]}
\end{align*}
Since $\partial_\theta e^{-H} \neq -\partial_\theta H e^{-H}$, we now rely on Duhamel's formula to express:
\begin{align*}
\partial_\theta e^{-H} &= - \int_0^1 e^{-\tau H} \partial_\theta H e^{(\tau-1) H} d\tau \\
\Rightarrow \Lambda_{\bm{s,a}}\partial_\theta e^{-H} &=  - \int_0^1 \Lambda_{\bm{s,a}} e^{-\tau H} \partial_\theta H e^{(\tau-1) H} d\tau
\end{align*}
We then use the circular permutation property of the trace to get:
\begin{align*}
\text{Tr}\left[\Lambda_{\bm{s,a}}\partial_\theta e^{-H}\right] = - \int_0^1 \text{Tr}\left[\partial_\theta H e^{(\tau-1) H} \Lambda_{\bm{s,a}} e^{-\tau H}\right] d\tau
\end{align*}
Using the eigendecompostion of $H_{QBM} = \sum_{\bm{v},\bm{h_{v}}} \lambda_{\bm{v},\bm{h_{v}}} \ket{\bm{v,h_v}}\bra{\bm{v,h_v}}$ (Appendix \ref{sec:bm_energy-based}), we find:
\begin{align*}
e^{(\tau-1) H} \Lambda_{\bm{v}} e^{-\tau H} &= \sum_{\bm{v'},\bm{h_{v'}}}e^{(\tau-1) \lambda_{\bm{v',h_{v'}}}} \ket{\bm{v',h_{v'}}}\bra{\bm{v',h_{v'}}}\\ &\quad\ \Lambda_{\bm{v}} \sum_{\bm{v'},\bm{h_{v'}}}e^{-\tau \lambda_{\bm{v',h_{v'}}}} \ket{\bm{v',h_{v'}}}\bra{\bm{v',h_{v'}}}\\
&= \sum_{\bm{h_{v}}} e^{-\lambda_{\bm{v,h_v}}}\ket{\bm{v,h_{v}}}\bra{\bm{v,h_{v}}}\\
&= \Lambda_{\bm{v}} e^{-H} \Lambda_{\bm{v}}
\end{align*}
and hence
\begin{align*}
\text{Tr}\left[\Lambda_{\bm{s,a}}\partial_\theta e^{-H}\right] &= - \int_0^1 \text{Tr}\left[\partial_\theta H \Lambda_{\bm{s,a}} e^{-H} \Lambda_{\bm{s,a}}\right] d\tau\\
&= - \text{Tr}\left[\partial_\theta H \Lambda_{\bm{s,a}} e^{-H} \Lambda_{\bm{s,a}}\right]
\end{align*}
and
\begin{align*}
\frac{\partial F}{\partial \theta} &=\frac{\text{Tr}\left[\partial_\theta H \Lambda_{\bm{s,a}} e^{-H} \Lambda_{\bm{s,a}}\right]}{\text{Tr}\left[ \Lambda_{\bm{s,a}} e^{-H}\right]}\\
&= \text{Tr}[(\partial_\theta H) \rho_{\bm{s,a}}]\\ 
&= \langle \partial_\theta H \rangle_{\rho_{\bm{s,a}}}
\end{align*}
where now $\partial_\theta H$ is the $\sigma^{x}_{k}$, $\sigma^{z}_{l}$ or $\sigma^{z}_{l}\sigma^{z}_{l'}$ operator associated to the parameter $\theta$.

\section{Coherent access to the transition matrix\label{sec:coherent-access}}

We want to implement the operation $\ket{i}_1 \ket{0}_2 \ket{0}_{aux} \rightarrow \ket{i}_1 \sum_{j \in N(i)} \sqrt{P_{ij}} \ket{j}_2 \ket{0}_{aux}$ coherently for all actions $i$ and for the transition matrix $P$ and neighborhood structure $N$ given by Metropolis-Hastings or Gibbs sampling (Appendix \ref{sec:walks}). In order to achieve this, we assume access to the description of the classical circuits that compute the operations $\ket{i}_1 \ket{j}_2 \ket{0} \rightarrow \ket{i}_1 \ket{j}_2 \ket{P_{ij}}$ and $\ket{i}\ket{0} \rightarrow \ket{i} \bigotimes_{j \in N(i)} \ket{j}$, having complexity $C(P)$ and $\mathcal{O}(\textrm{poly}(M))$ respectively. Using either Fredkin or Toffoli gates, it is possible to design quantum circuits that implement coherently the operation $\ket{i} \ket{0}^{\otimes M} \ket{0}_{aux} \rightarrow \ket{i} \ket{P_{i1}} \ldots \ket{P_{iM}} \ket{0}_{aux}$ with comparable complexities by \emph{reversibilization} (Sec.\ 3.2.5.\ of \cite{nielsen10}), where $M=\max_{a\in\mathcal{A}}N(a)$ and the $j\in N(i)$ have been renamed $1, \ldots, M$ for each $i$ for simplicity. Using this circuit and a procedure called \emph{coherent controlization} \cite{dunjko15,grover02}, it is possible to implement the coherent access to the columns of the transition matrix $P$ using $\mathcal{O}(M\times C(P)+\textrm{poly}(M))$ operations, same as classically. This implementation goes as follows:
\begin{widetext}
\begin{align*}
&\ket{i}_{1} \ket{0}^{\otimes M}_{aux} \ket{0}^{\otimes M-2}_{aux} \ket{0}^{\otimes\log(M)}_2\\
&\rightarrow \ket{i}_{1} \ket{P_{i1}} \ldots \ket{P_{iM}} \ket{0}^{\otimes M-2}_{aux} \ket{0}^{\otimes\log(M)}_2\ \#\ \textrm{cost } M\times C(P) + \mathcal{O}(\textrm{poly}(M))\\
&\rightarrow \ket{i}_{1} \ket{P_{i1}} \ldots \ket{P_{iM}} \ket{P_{i1} + \ldots + P_{i\frac{M}{2}}} \ldots \ket{P_{i,\frac{M}{2}-1} + P_{i\frac{M}{2}}} \ket{0}^{\otimes\log(M)}_2\ \#\ \textrm{cost } M-2 \textrm{ additions}\\
&\rightarrow \ket{i}_{1} \ket{P_{i1}} \ldots \ket{P_{iM}} \ket{P_{i1} + \ldots + P_{i\frac{M}{2}}} \ldots \ket{P_{i,\frac{M}{2}-1} + P_{i\frac{M}{2}}} \sum_{j=1}^{M}\sqrt{P_{ij}}\ket{j}_{2}\ \#\ \textrm{cost of CC: } \mathcal{O}(\textrm{poly}(M))\\
&\rightarrow \ket{i}_{1} \ket{0}^{\otimes M}_{aux} \ket{0}^{\otimes M-2}_{aux} \sum_{j=1}^{M}\sqrt{P_{ij}}\ket{j}_{2}\ \#\ \textrm{cost } M\times C(P) + \mathcal{O}(\textrm{poly}(M))
\end{align*}
\end{widetext}

\section{Details of the numerical simulations\label{sec:app_num_sim}}

\textbf{Hyperparameters. }For the simulations of GridWorld, we chose ADAM as an optimizer with a learning rate of $\alpha=0.0001$. 
The training batches of the main network contain $100$ samples and the replay memory has a capacity of $5000$ interactions. 
The target network is updated every $100$ trials while the experience replay memory is updated every $100$ timesteps (in a First-In-First-Out manner).
We chose an exponential annealing schedule for the glow parameters of the PS update rule that starts with a value of $0.9$ and ends with a value of $0.99$.
The agent's policy is a softmax policy and for the $\beta$-parameter we choose a hyperbolic tangent schedule that starts with a value of $0.001$ and ends with a value of $0.8$. 
The neural networks have $64$ units per layer.

For the simulations of CartPole, in order to have a fair comparison of the performance of neural networks with different depths, we set their number of units per hidden layer such that all neural networks have approximately the same total number of weights and biases. More precisely, we set the number of units per layer to be $10$, $19$, $73$ for neural networks with respectively $5$, $2$ and $1$ hidden layer(s), leading to a total number of weights and biases of $525$, $528$, $526$ respectively.\\
For neural networks with $5$ hidden layers, we chose ADAM as an optimizer with a learning rate of $\alpha=0.001$.  
The training batches of the main network contain $100$ samples and the replay memory has a capacity of $2000$ interactions. 
The target network is updated every $10000$ trials while the experience replay memory is updated every $5$ timesteps (in a First In First Out manner).
For neural networks with $1$ and $2$ hidden layer(s), we chose ADAM as an optimizer with a learning rate of $\alpha=0.01$.  
The training batches of the main network contain $200$ samples and the replay memory has a capacity of $2000$ interactions. 
The target network is updated every 5000 trials while the experience replay memory is updated every $5$ timesteps (in a First In First Out manner).
For all simulations, we chose an exponential annealing schedule for the glow parameters of the PS update rule that starts with a value of $0.5$ and ends with a value of $0.99$.
The agent's policy is a softmax policy and for the $\beta$-parameter we choose a hyperbolic tangent schedule that starts with a value of $0.001$ and ends with a value of $1$.\\

For the simulations of the policy-sampling experiment, in order to have a fair comparison of performance between DEBNs and DQNs, we chose to keep their total number of weights fixed to $10440$, i.e., such that the mimimal number of hidden units par layer of the DQNs are $10$ in all instances of the experiment. For a given number of actions, we adapt the size of the two hidden layers of the networks such that both layers have the same number of hidden units.
We chose ADAM as an optimizer with a learning rate of $\alpha=0.001$.  
The training batches of the neural network contain $10$ samples. 
The agent's policy is a softmax policy with a $\beta$-parameter set to $1$.\\

For the simulations of the reward-discounting experiment, in order to have a fair comparison of performance between DEBNs and DQNs, we chose to keep their total number of weights fixed to $1900$, i.e., such that the mimimal number of hidden units par layer of the DQNs are $10$ in all instances of the experiment. For a given number of actions, we adapt the size of the two hidden layers of the networks such that both layers have the same number of hidden units.
We chose ADAM as an optimizer with a learning rate of $\alpha=0.01$.  
The training batches of the main network contain $10$ samples and the replay memory has a capacity of $2000$ interactions. 
The target network is updated every $100$ timesteps while the experience replay memory is updated (along with the main network) every $10$ timesteps (in a First In First Out manner).\\

For the simulations of the circular GridWorld experiment, in order to have a fair comparison of performance between DEBNs and DQNs, we chose to keep their total number of weights fixed to $2460$, i.e., such that the mimimal number of hidden units par layer of the DQNs are $10$ in all instances of the experiment. For a given number of actions, we adapt the size of the two hidden layers of the networks such that both layers have the same number of hidden units.
We chose ADAM as an optimizer with a learning rate of $\alpha=0.01$.  
The training batches of the main network contain $100$ samples and the replay memory has a capacity of $5000$ interactions. 
The target network is updated every $500$ timesteps while the experience replay memory is updated (along with the main network) every $10$ timesteps (in a First In First Out manner).
The agent's policy is a softmax policy and for the $\beta$-parameter we choose a linear schedule that starts with a value of $1$ and ends with a value of $100$. 
\end{document}